\documentclass[11pt, letterpaper]{article}
\pdfoutput=1
\usepackage{fullpage} 
\usepackage{xspace,xcolor,graphicx,tabularx} 
\usepackage{mathtools,amsthm,amssymb} 
\usepackage{enumitem} 
\usepackage{epsfig}
\setenumerate[0]{label={\normalfont (\roman*)}}

\usepackage[T1]{fontenc}
\usepackage[UKenglish]{babel}
\usepackage[scaled=0.86]{helvet}
\usepackage{mathptmx}
\DeclareMathAlphabet{\mathcal}{OMS}{ntxm}{m}{n}
\usepackage{dsfont} 
\let\mathbb\relax
\let\mathbb\mathds
\usepackage{stmaryrd} 
\makeatletter
\renewcommand{\paragraph}{%
  \@startsection{paragraph}{4}%
  {\z@}{2.25ex \@plus 1ex \@minus .2ex}{-1em}%
  {\normalfont\normalsize\bfseries}%
}
\makeatother
\usepackage{authblk}

\definecolor{linkblue}{HTML}{001487}
\usepackage[colorlinks=true,allcolors=linkblue]{hyperref}
\usepackage{url}
\usepackage[nameinlink,capitalize,noabbrev]{cleveref}

\newtheorem{theorem}{Theorem}[section]
\newtheorem*{theorem*}{Theorem}
\newtheorem{proposition}[theorem]{Proposition}

\newtheorem{lemma}[theorem]{Lemma}

\newtheorem{corollary}[theorem]{Corollary}
\theoremstyle{remark}
\newtheorem{remark}[theorem]{Remark}
\theoremstyle{definition}
\newtheorem{definition}[theorem]{Definition}
\newtheorem{assumption}[theorem]{Assumption}

\numberwithin{equation}{section}
\newcommand\numberthis{\addtocounter{equation}{1}\tag{\theequation}}

\newcommand{\setft}[1]{\textnormal{#1}}
\newcommand{\eps}{\epsilon}

\newcommand{\C}{\ensuremath{\mathds{C}}}

\newcommand{\N}{\ensuremath{\mathds{N}}}

\newcommand{\Z}{\ensuremath{\mathds{Z}}}
\newcommand{\bits}{\ensuremath{\{0, 1\}}}

\newcommand{\ot}{\ensuremath{\otimes}}
\newcommand{\deq}{\coloneqq}
\newcommand{\tr}[1]{\mathrm{Tr}\!\left[ #1 \right]}

\newcommand{\pr}[1]{\mathrm{Pr}\!\left[ #1 \right]}
\newcommand{\prs}[2]{\mathrm{Pr}_{#1}\!\left[ #2 \right]}
\newcommand{\norm}[1]{\left\lVert#1\right\rVert}

\DeclareMathOperator{\poly}{poly}
\DeclareMathOperator{\negl}{negl}

\DeclareMathOperator{\img}{img}
\DeclareMathOperator{\rank}{rank}
\DeclareMathOperator{\E}{\mathds{E}}

\newcommand{\sth}{{\setft{~s.t.~}}}

\newcommand{\ket}[1]{|#1\rangle}
\newcommand{\bra}[1]{\langle#1|}
\newcommand{\proj}[1]{\ket{#1}\!\bra{#1}}
\DeclarePairedDelimiterX\braket[2]{\langle}{\rangle}{#1 \delimsize\vert #2}

\newcommand{\cA}{\ensuremath{\mathcal{A}}}

\newcommand{\cH}{\ensuremath{\mathcal{H}}}

\newcommand{\cK}{\ensuremath{\mathcal{K}}}

\newcommand{\psioutput}{\ket{\psi_{\mathrm{output}}}}
\newcommand{\psiground}{\ket{\psi_{\mathrm{ground}}}}
\newcommand{\ontop}[2]{{\begin{array}{l} {#1} \\ {#2} \end{array}}}
\newcommand{\ketbra}[2]{|#1\rangle\langle#2|}
\newcommand{\stackket}[2]{{\Big | \hskip -5pt \ontop{#1}{#2} \hskip -3pt  \Big \rangle}}
\newcommand{\stackbra}[2]{{\Big \langle \hskip -5pt \ontop{#1}{#2} \hskip -3pt  \Big |}}
\newcommand{\stackketbra}[2]{{\stackket{#1}{#2} \stackbra{#1}{#2} }}

\DeclareSymbolFont{greekletters}{OML}{ntxmi}{m}{it}
\DeclareMathSymbol{\alpha}{\mathord}{greekletters}{"0B}
\DeclareMathSymbol{\beta}{\mathord}{greekletters}{"0C}
\DeclareMathSymbol{\gamma}{\mathord}{greekletters}{"0D}
\DeclareMathSymbol{\delta}{\mathord}{greekletters}{"0E}
\DeclareMathSymbol{\epsilon}{\mathord}{greekletters}{"0F}
\DeclareMathSymbol{\zeta}{\mathord}{greekletters}{"10}
\DeclareMathSymbol{\eta}{\mathord}{greekletters}{"11}
\DeclareMathSymbol{\theta}{\mathord}{greekletters}{"12}
\DeclareMathSymbol{\iota}{\mathord}{greekletters}{"13}
\DeclareMathSymbol{\kappa}{\mathord}{greekletters}{"14}
\DeclareMathSymbol{\lambda}{\mathord}{greekletters}{"15}
\DeclareMathSymbol{\mu}{\mathord}{greekletters}{"16}
\DeclareMathSymbol{\nu}{\mathord}{greekletters}{"17}
\DeclareMathSymbol{\xi}{\mathord}{greekletters}{"18}
\DeclareMathSymbol{\pi}{\mathord}{greekletters}{"19}
\DeclareMathSymbol{\rho}{\mathord}{greekletters}{"1A}
\DeclareMathSymbol{\sigma}{\mathord}{greekletters}{"1B}
\DeclareMathSymbol{\tau}{\mathord}{greekletters}{"1C}
\DeclareMathSymbol{\upsilon}{\mathord}{greekletters}{"1D}
\DeclareMathSymbol{\phi}{\mathord}{greekletters}{"1E}
\DeclareMathSymbol{\chi}{\mathord}{greekletters}{"1F}
\DeclareMathSymbol{\psi}{\mathord}{greekletters}{"20}
\DeclareMathSymbol{\omega}{\mathord}{greekletters}{"21}
\DeclareMathSymbol{\varepsilon}{\mathord}{greekletters}{"22}
\DeclareMathSymbol{\vartheta}{\mathord}{greekletters}{"23}
\DeclareMathSymbol{\varpi}{\mathord}{greekletters}{"24}
\DeclareMathSymbol{\varrho}{\mathord}{greekletters}{"25}
\DeclareMathSymbol{\varsigma}{\mathord}{greekletters}{"26}
\DeclareMathSymbol{\varphi}{\mathord}{greekletters}{"27}

\newcommand{\roundp}[1]{\lfloor #1 \rfloor_p}
\newcommand{\keyi}{\cK^{\rm inj}}
\newcommand{\keyl}{\cK^{\rm lossy}}
\newcommand{\keyh}{\cK^{\rm hash}}
\newcommand{\ind}{\bold{1}}
\newcommand{\lwe}{\setft{LWE}}
\newcommand{\hi}{{\rm high}}
\newcommand{\lo}{{\rm low}}

\newcommand{\tab}[1]{\hyperref[tab:#1]{Table~\ref{tab:#1}}}

\newlength{\bigcirclen}

\def\statea{{\bigcirc}}
\def\stateba{{\settowidth{\bigcirclen}{$\bigcirc$}\makebox[0pt][l]{\makebox[\bigcirclen][c]{$\uparrow$}}\bigcirc}}
\def\statebb{{\settowidth{\bigcirclen}{$\bigcirc$}\makebox[0pt][l]{\makebox[\bigcirclen][c]{$\downarrow$}}\bigcirc}}
\def\stateb{{\settowidth{\bigcirclen}{$\bigcirc$}\makebox[0pt][l]{\makebox[\bigcirclen][c]{$\uparrow$}}\makebox[0pt][l]{\makebox[\bigcirclen][c]{$\downarrow$}}\bigcirc}}

\def\stateca{{\settowidth{\bigcirclen}{$\bigcirc$}\makebox[0pt][l]{\makebox[\bigcirclen][c]{$\Uparrow$}}\bigcirc}}
\def\statecb{{\settowidth{\bigcirclen}{$\bigcirc$}\makebox[0pt][l]{\makebox[\bigcirclen][c]{$\Downarrow$}}\bigcirc}}
\def\statec{{\settowidth{\bigcirclen}{$\bigcirc$}\makebox[0pt][l]{\makebox[\bigcirclen][c]{\raisebox{.035em}{\fontsize{9}{0}\selectfont$\Updownarrow$}}}\bigcirc}}
\def\stated{{\settowidth{\bigcirclen}{$\bigcirc$}%
\makebox[0pt][l]{\makebox[\bigcirclen][c]{$\times$}}\bigcirc}}

\def\stateea{{\settowidth{\bigcirclen}{$\bigcirc$}\makebox[0pt][l]{\makebox[\bigcirclen][c]{$\upuparrows$}}\bigcirc}}
\def\stateeb{{\settowidth{\bigcirclen}{$\bigcirc$}\makebox[0pt][l]{\makebox[\bigcirclen][c]{$\downdownarrows$}}\bigcirc}}
\def\statee{{\settowidth{\bigcirclen}{$\bigcirc$}\makebox[0pt][l]{\makebox[\bigcirclen][c]{$\upuparrows$}}\makebox[0pt][l]{\makebox[\bigcirclen][c]{$\downdownarrows$}}\bigcirc}}

\def\statem{{\settowidth{\bigcirclen}{$\bigcirc$}\makebox[0pt][l]{\makebox[\bigcirclen][c]{\scriptsize{M}}}\bigcirc}}

\newtheoremstyle{break}
  {\topsep}{\topsep}
  {\itshape}{}
  {\bfseries}{}
  {\newline}{}
\theoremstyle{break}

\begin{document}

\title{Public-key pseudoentanglement and the \\ hardness of learning ground state entanglement structure}

\author[1]{Adam Bouland\footnote{abouland@stanford.edu}}
\author[2]{Bill Fefferman\footnote{wjf@uchicago.edu}}
\author[2]{Soumik Ghosh\footnote{soumikghosh@uchicago.edu}}
\author[3]{\\Tony Metger\footnote{tmetger@ethz.ch}}
\author[4]{Umesh Vazirani\footnote{vazirani@eecs.berkeley.edu}}
\author[1]{Chenyi Zhang\footnote{chenyiz@stanford.edu}}
\author[1]{Zixin Zhou\footnote{jackzhou@stanford.edu}}

\affil[1]{Stanford University}
\affil[2]{University of Chicago}
\affil[3]{ETH Zurich}
\affil[4]{UC Berkeley}

\date{\vspace{-1cm}}

\maketitle
\thispagestyle{empty}

\begingroup 
\hyphenpenalty 10000
\exhyphenpenalty 10000
\begin{abstract}
Given a local Hamiltonian, how difficult is it to determine the entanglement structure of its ground state?
We show that this problem is computationally intractable even if one is only trying to decide if the ground state is volume-law vs near area-law entangled.
We prove this by constructing strong forms of pseudoentanglement in a public-key setting, where the circuits used to prepare the states are public knowledge.
In particular, we construct two families of quantum circuits which produce volume-law vs near area-law entangled states, but nonetheless the classical descriptions of the circuits are indistinguishable under the Learning with Errors (LWE) assumption.
Indistinguishability of the circuits then allows us to translate our construction to Hamiltonians.
Our work opens new directions in Hamiltonian complexity, for example whether it is difficult to learn certain phases of matter. 
\end{abstract}
\endgroup

\newpage 
\thispagestyle{empty}
{
\hypersetup{linkcolor=black}
\setcounter{tocdepth}{3}
\tableofcontents
}

\newpage

\maketitle
\section{Introduction}
The central problem in Hamiltonian complexity is to understand the structure of ground states of local Hamiltonians and the difficulty of learning properties of them, e.g.~\cite{hastings2007area,schuch2008computational,landau2015polynomial,arad2017rigorous,anshu2023survey}. 
In this work we study the following question: given a local Hamiltonian $H$, how difficult is it to learn about the entanglement structure of its ground state? For example, given $H$, can you tell if its ground state is area-law or volume-law entangled?  We  call such questions about the qualitative features of the entanglement structure a \emph{Learning Ground State Entanglement Structure} (LGSES) problem.
This problem is related to both condensed matter physics -- where ground state entanglement structure is a central object of study -- and may also shed light on questions in quantum gravity regarding how entanglement could possibly be dual to other physical quantities  \cite{bouland2019computational,gheorghiu2020estimating}.

Whereas most effort in many-body physics has been directed towards the positive side of this question, i.e.~finding conditions under which properties of the ground state can be efficiently learnt or computed (see e.g.~\cite{cramer2010efficient,landau2015polynomial,huang2022provably}), here  our goal is to prove hardness results for the LGSES problem from cryptographic assumptions.
This explores the limitations that any algorithm for such problems must run up against.
In other words, hardness results for the LGSES problem point to qualitative features of Hamiltonian ground states that are inherently computationally intractable to compute.

To prove computational hardness results for the LGSES problem, we will relate it to a notion called \emph{pseudoentanglement}, a term recently introduced in~\cite{bouland2022quantum}.
Informally, pseudoentangled state ensembles consist of low-entanglement states that masquerade as high-entanglement states to computationally bounded observers; in other words, a pseudoentangled state only looks like a high-entanglement state to bounded observers, whereas its actual (information-theoretic) entanglement is low.
The main result of~\cite{bouland2022quantum} is that it is possible to create pseudoentangled states which hide vast differences in entanglement.
More concretely, they construct two ensembles of quantum states, $\Psi^{\hi}$ and $\Psi^{\lo}$, such that any state $\ket{\psi} \in \Psi^{\hi}$ has high entanglement entropy across any bipartition of the qubits and states in $\Psi^{\lo}$ have low entanglement, but any computationally bounded quantum algorithm that receives a state from $\Psi^{\hi}\cup \Psi^{\lo}$ cannot tell which kind of state it received.

This notion of pseudoentanglement is interesting in its own right and has been used for various applications in~\cite{bouland2022quantum}, but its relevance to Hamiltonian complexity is unclear.
This is because the settings are inherently different:
the notion of pseudoentanglement in~\cite{bouland2022quantum} involves a \emph{quantum} input, namely copies of the relevant quantum states, being given to the distinguisher.
If we tried to translate this into a setting involving Hamiltonian ground states, we would end up in a model where we study the properties of Hamiltonian ground states given quantum copies of these ground states, but without knowing the actual Hamiltonian itself.
In contrast, in Hamiltonian complexity we assume that we know a classical description of the Hamiltonian under consideration and would like to determine properties of its ground state.

Therefore, if we want to use some notion of pseudoentanglement to prove hardness results for the LGSES problem, we need to consider a model where the distinguisher is not just given quantum copies of the high- or low-entanglement states, but rather an efficient classical description of these states.
This leads to a different kind of pseudoentanglement, which we call \emph{public-key} pseudoentanglement since the description of the states (the ``key'') can be made public.
This notion is already implicit in earlier work by Gheorghiu and Hoban~\cite{gheorghiu2020estimating}, who gave a construction of public-key pseudoentanglement that we discuss in more detail below.

\begin{definition}[Public-key pseudoentanglement (implicit in \cite{gheorghiu2020estimating})] Two ensembles of $n$-qubit $\text{poly}(n)$-gate quantum circuits $\{C_k\}$, $\{C'_k\}$ are public-key pseudoentangled with entanglement gap $f(n)$ vs $g(n)$ if they are computationally indistinguishable to poly-time quantum algorithms, and yet with high probability over the ensembles, $C_k\ket{0^n}$ has entanglement $\geq f(n)$  and $C'_k\ket{0^n}$ has entanglement $\leq g(n)$ across one or more cuts of the system.
\end{definition}

The key difference between this definition and the one in~\cite{bouland2022quantum} is that here a classical description of the circuit used to prepare the states is given as input to the distinguisher, whereas in~\cite{bouland2022quantum} the distinguisher only receives copies of the output state of the circuit.
A distinguisher can therefore not only prepare copies of the output state, but also analyse the classical description of the circuits directly to gain additional information, making it  harder to ``hide'' the entanglement of pseudoentangled states.
Hence, public-key pseudoentanglement is a much stronger notion than that in~\cite{bouland2022quantum}, which we will call \emph{private-key} pseudoentanglement as the circuits are hidden.

There are generally two features we care about in a pseudoentanglement construction: the entanglement gap, which we want to be as large as possible in order to hide as much entanglement as possible, and the set of cuts across which this entanglement gap holds.
Informally, a more general set of cuts (i.e.~bipartitions of the qubits) across which the entanglement gap holds corresponds to hiding more qualitative information about the entanglement structure of the state; this is what we are most interested in for the LGSES problem.

The pioneering work of Gheorghiu and Hoban gave the first pseudoentanglement construction with entanglement gap $n$ vs $n-O(1)$ across a single cut based on the Learning With Errors (LWE) assumption \cite{gheorghiu2020estimating}. 
Using this they showed that it is difficult to learn \emph{fine-grained} properties of the ground state entanglement -- namely if the ground state has entanglement $n$ vs $n-O(1)$ across a fixed cut. 
The basic idea is that if one passes the circuit to prepare pseudoentangled states through a modified Kitaev clock construction~\cite{kitaev,nirkhe2018approximate}, the ground state of the resulting Hamiltonian has the same entanglement properties as the output state of the circuit.
We emphasize that the Kitaev clock Hamiltonian encodes the circuit used to prepare the state in plaintext, so this application necessarily requires a \emph{public-key} construction.

However, Gheorghiu and Hoban's construction only suffices to prove hardness of detecting very small differences in the entanglement of the ground state across a single cut.
In contrast, the LGSES problem asks about \emph{qualitative} or \emph{coarse-grained} features of the entanglement structure, as very fine-grained properties are usually not physically relevant.
This raises the following question: is it possible to get a public-key pseudoentanglement construction that hides qualitatively different entanglement structures?
This would lead to natural hardness statements for the LGSES problem.

\subsection{Near area-law public-key pseudoentanglement}

Our first result is to construct strong forms of public-key pseudoentanglement from LWE.
In particular we show it is possible to hide 1D near area-law vs volume-law entanglement.

\begin{theorem}[Volume vs area-law public-key pseudoentanglement (informal)] \label{thm:intromain}
    Assuming subexponential-time hardness of LWE, there exist public-key pseudoentangled ensembles with volume-law vs near area-law entanglement when the qubits are arranged on a 1D line. 

\end{theorem}

That is, in one case the states have entanglement $\Omega(\min(k, n - k))$ across any division of the qubits into $k$ vs $n-k$ qubits (volume-law), and in the other case the entanglement of any cut is $\leq |A| \text{polylog}(n)$ where $|A|$ is the area of the cut when the qubits are arranged on a 1D line (i.e.~the number of times the cut crosses the 1D line).
We call this \emph{near area-law entangled}.
This is optimal because if the polylog factor were changed to a log, then these states would be efficiently distinguishable from one another by standard Matrix Product State learning algorithms \cite{cramer2010efficient,landau2015polynomial}. 
Hiding such qualitatively different entanglement structures requires a completely different construction from \cite{gheorghiu2020estimating}.
We note that while in \cref{thm:intromain} we assume subexponential-time hardness of LWE, we also show that the standard LWE assumption implies a similar result.\footnote{In particular, the polylog correction factor to the area-law scaling is replaced by an $n^\epsilon$ correction, where $\epsilon$ can be any constant $>0$.}
We will discuss the application of our result to Hamiltonian complexity shortly, after we present its proof sketch.

\paragraph{Proof sketch for single-cut pseudoentanglement.}

Let us first consider a single cut partitioning the qubits into sets $A$ and $B$.
In this case, the most natural states to consider are of the form $\sum_{x \in \{0,1\}^n} \ket{x}_A \ket{h(x)}_B$ for some function $h$.
If one chooses $h$ to be injective, then this state has entanglement entropy $n$ across this cut; if one chooses $h$ to be $2^k$-to-1, then the entanglement is $n - k$.
\cite{gheorghiu2020estimating} used exactly these states with trapdoor claw-free functions from~\cite{randomness}, which are functions that are either injective or 2-to-1, but the two cases are computationally hard to distinguish (given a description of the function) assuming the hardness of LWE.
There are some additional subtleties arising from the fact that the trapdoor claw-free functions from~\cite{randomness} do not output numbers, but probability distributions, and are only approximately 2-to-1. 
We refer to~\cite{gheorghiu2020estimating} for a detailed analysis of this construction.

To increase the entanglement gap, we need to make the many-to-1 functions \emph{more compressing}.
The functions in \cite{gheorghiu2020estimating,randomness} additionally have a trapdoor; however, we observe that for pseudoentanglement, we can dispense with the trapdoor and only need so-called \emph{lossy functions}:
these are functions that are either injective or $2^k$-to-1, but again the two kinds of functions are hard to distinguish.
Starting from this observation, it turns out to be possible to combine the construction from~\cite{gheorghiu2020estimating} with ideas from a recent randomness generation protocol~\cite{mahadev2022efficient} to achieve an entanglement gap of $n$ vs $n^\delta$ for any $\delta > 0$.\footnote{We note that independently from and simultaneously to our work, Gheorghiu and Hoban updated their results to include a construction of this form, which achieves the aforementioned gap of $n^\delta$ vs $n$ across a fixed cut with $n$ qubits in one set and $\poly(n)$ in the other.}

The challenge with this approach based on~\cite{gheorghiu2020estimating} is that it appears fundamentally restricted to a single cut.
However, our goal is to obtain pseudoentangled states with near area-law vs volume-law entanglement structure, which requires low entanglement across \emph{exponentially many cuts} simultaneously.
This would require controlling the entanglement not just between regions $A$ and $B$, but also within these regions.
With the approach of~\cite{gheorghiu2020estimating} it seems difficult to appropriately modify the state in register $A$ without jeopardising the entanglement across the cut between $A$ and $B$.

For this reason, we need a different kind of state that allows us to control the entanglement across all cuts simultaneously.
It turns out that a useful class to consider are binary phase states as in~\cite{ji2018pseudorandom,brakerski2019pseudo,bouland2022quantum}, i.e.,~states of the form $\sum_{x\in\{0,1\}^n} (-1)^{f(x)} \ket{x}$.
Our approach will be as follows: we start from a phase state with high-entanglement across every cut.
We will then modify this state (in a computationally undetectable way) to have low entanglement across some particular cut.
This achieves the essentially same as the~\cite{gheorghiu2020estimating}-based construction above.\footnote{One advantage of this construction even for a single cut is that it allows an entanglement gap of $n$ vs $\poly \log n$ for an $O(n)$-qubit state, whereas in the \cite{gheorghiu2020estimating}-based construction the $B$-register had to have $\poly(n)$ qubits for a comparable entanglement gap.}
However, crucially our ``modification procedure'' is iterable: this means that we can perform essentially the same entanglement reduction operation across many cuts in sequence and end up with a state with low entanglement across every cut.

We first describe the construction and proof for a single cut.
For simplicity, let us first consider the cut between the first and second $n/2$ qubits.
One can easily compute that the reduced density matrix of the first half of the state is $\rho \propto T T^{\top}$ where $T_{ij}=(-1)^{f(i \parallel j)}$ and $i,j\in\{0,1\}^{n/2}$ are the first and second halves of the string $x$, i.e.,~$T$ is the truth table of the function $f$ written out in a matrix form. 
By a direct calculation, one can show an upper bound on the entanglement entropy of our phase state (i.e.~the von Neumann entropy of $\rho$) in terms of the rank of $T$, and a lower bound on the entanglement entropy in terms of the Frobenius norm of $T T^{\top}$.
Our strategy will therefore be to start from a ``$T$-matrix'' for a high-entanglement state, and then perform one of two modifications in a computationally indistinguishable way: either modify $T$ to reduce its rank to get low-entanglement states, or modify $T$ in a way that does \emph{not} reduce $\norm{T T^{\top}}_2$ too much so that the entanglement of the states remains high.
These modified $T$-matrices then correspond to modified phase states, and our goal is to hide which of the two procedures we performed, even when handing out a classical description for preparing the corresponding states.

In \cite{bouland2022quantum} the idea is to apply a private key cryptographic hash function to repeat rows of the matrix $T$ to reduce its rank.
That is, we pick a phase function $f(x)$ which yields high entanglement \cite{brakerski2019pseudo}, and then ``whittle down'' its entanglement by replacing $f(x)$ with $f(h(i)j)$  where $i,j\in\{0,1\}^{n/2}$ and $h:\{0,1\}^{n/2}\rightarrow \{0,1\}^{n/2}$ is a $2^k$-to-1 function. This reduces the number of distinct rows in $T$ (each of which is now repeated many times) and correspondingly decreases the rank of $T$.
As a result, this procedure lowers the entanglement of $\rho$ by $k$. 
On the other hand, if we pick $h$ to be a 1-to-1 function, we simply permute the rows of $T$ and do not change $\norm{T T^{\top}}_2$.
To make this public-key, we need a way of applying a $2^k$-to-1 or a 1-to-1 hash function to these phase states in such a way that it is hard to tell which one was applied, even when the state preparation circuit (and therefore the source code for $h$) is public.

At first sight, it may seem that we can use the same idea as above: take the lossy functions from~\cite{mahadev2022efficient} and use them to reduce the rank of $T$.
However, the phase state construction is much less flexible than states of the form $\sum_{x} \ket{x}\ket{h(x)}$: for the latter, the codomain of $h$ does not matter and we can use the functions $h$ from~\cite{mahadev2022efficient}, whose codomain are probability distributions over $\mathbb{Z}_q^m$ for $m \gg n$.
In contrast, for phase states we need lossy functions $h:\{0,1\}^{n/2}\rightarrow \{0,1\}^{n/2}$, i.e.~the codomain has to be the same as the domain.
This is a somewhat unusual requirement from a cryptographic perspective and forces us to use a custom construction of ``imperfect'' lossy functions based on LWE.
Concretely, we first show how to create lossy functions mapping $\{0,1\}^{n/2} \rightarrow \{0,1\}^{\text{poly}(n)}$ which are \emph{exactly} injective vs $2^k$-to-1 (for sufficiently large $k$) with high probability. This is not yet what we need, as the codomain is exponentially larger than the domain. To fix this, we compose these functions with pairwise independent hash functions which shrink the codomain size back to $2^{n/2}$. 
This can only make the $2^k$-to-1 functions more compressing, which is to our benefit.
However, this also introduces unwanted collisions for the injective functions which might break the high-entanglement case.
To deal with this, we show that the repetition pattern this produces in the matrix $T$ is sufficiently well-behaved that the corresponding states still have high entanglement.
Intuitively, this holds because even though the ``injective'' functions or no longer actually injective, they are still pairwise independent, which ensures enough independence in the row repetition pattern of $T$ to give a strong lower bound on the Frobenius norm $\norm{T T^{\top}}_2$. 

\paragraph{From single-cut to multi-cut pseudoentanglement.}
As we mentioned, the advantage of the above construction is that it can be extended to pseudoentanglement across all cuts simultaneously.
We now explain how this extension works at a high level.
A first observation is that to achieve pseudoentangled states with 1D near area law scaling, reducing the entanglement across all $n-1$ contiguous cuts of the line to $O(\text{polylog}(n))$ suffices by strong subadditivity. 
Therefore, a natural approach is to perform a 1D sweep of the line, reducing the entanglement of the contiguous cuts one at a time.
For reasons that will become apparent below, we perform this sweep right-to-left.
The final phase function is then a complicated composition of $n$ independent lossy functions and hash functions. 

To show that this construction indeed achieves pseudoentanglement across all cuts simultaneously, we need to worry about two issues: firstly, for the low entanglement states we need to ensure that performing the entanglement reduction operation for a cut towards the left of the line does not inadvertently increase the entanglement across earlier cuts to the right, so that the low entanglement across those earlier cuts is preserved.
Secondly, for the high entanglement states we need to make sure that the fact that we are using ``imperfect'' injective functions does not decrease the entanglement too much even after applying these functions across every cut.

The first concern is relatively easy to deal with due to the relationships between the ``$T$-matrices'' for different cuts.
To see this, consider a phase state $\ket{\psi} = \sum (-1)^{f(x)} \ket{x}$ on $n + m$ qubits on a line and let $T^{n|m} \in \{\pm 1\}^{2^n \times 2^m}$ be the $T$-matrix for the cut between the first $n$ and last $m$ qubits, i.e.~$T^{n|m}_{ij}=(-1)^{f(i \parallel j)}$ for $i \in \bits^n, j \in \bits^m$.
After performing the entanglement reduction operation, this matrix will only have some smaller number $R$ of distinct rows, each repeated many times.
In the next step of the sweep (recalling that we move right to left), we consider the cut between the first $n-1$ and last $m+1$ qubits.
We denote the corresponding $T$-matrix by $T^{n-1|m+1} \in \{\pm 1\}^{2^{n-1} \times 2^{m+1}}$.
From the definition of the $T$-matrix, one can see that the first row of $T^{n-1|m+1}$ simply consists of the first two rows of $T^{n|m}$ stacked side by side.
More generally, the $i$-th row of $T^{n-1|m+1}$ is simply the horizontal concatenation of rows $2i - 1$ and $2i$ from $T^{n|m}$.
Suppose we now reduce the rank of $T^{n-1|m+1}$ by removing some rows and duplicating others.
We then need to check that if we go back to the cut $n|m$, the resulting $T$-matrix (denoted $\tilde T^{n|m}$) still has rank at most $R$.
This is the case since the rows of $\tilde T^{n|m}$ consist of the first and second parts of the rows of $T^{n-1|m+1}$; since the rank reduction only repeats, but does not modify, rows in $T^{n-1|m+1}$, every row in $\tilde T^{n|m}$ must have already appeared in $T^{n|m}$.
As a result, the subsequent rank reduction for cut $n-1|m+1$ can only decrease, not increase, the rank of the $T$-matrix across $n|m$.
It is not too hard to see that this argument generalises to any future rank reduction operation, not just the immediately subsequent one.

The second concern is more difficult to deal with.
When applying this sweep with \emph{approximately} injective functions, the entanglement is reduced slightly each time. The rightmost (first) cut in particular has its entanglement reduced $n$ times, so even a tiny loss could kill the entire construction.
Perhaps surprisingly, we show that this is not the case, and the entanglement losses do not compound too badly. 
We show that different rows have different probabilities of being hashed together due to the structure of the sweep, and a careful accounting of this process reveals that not much entanglement is lost, even in the first cut.
The analysis is somewhat technical and we refer to \cref{sec:multicut} for details.
Finally, we note that while the construction we described here does not produce pseudorandom states ensembles (i.e.~the families of pseudoentangled states, without the public key, are not necessarily pseudorandom states~\cite{ji2018pseudorandom}), we can make a simple modification to our construction to ensure that this is the case.
Since our applications do not rely on this we only give a sketch of this in \cref{sec:smallcut}.

\subsection{Hardness of learning ground state entanglement structure}

Our second result is to show that this public-key, area vs volume-law pseudoentanglement construction enables new applications in quantum Hamiltonian complexity. 
Because our public-key pseudoentanglement construction can hide qualitative features of the entanglement structure, we can show natural results for the hardness of the Learning Ground State Entanglement Structure (LGSES) problem for broad differences in entanglement structure.
As we mentioned above, the main idea for turning pseudoentanglement constructions into hard instances of LGSES is to use a circuit-to-Hamiltonian construction on the state preparation circuit for the pseudoentangled state.
There are a variety of circuit-to-Hamiltonian constructions and using these on our pseudoentangled states yields a variety of hardness statements for LGSES.
In this paper, we consider three different constructions: a Kitaev clock construction with a binary clock, a Kitaev clock construction with a unary clock, and a customised version of a geoemtrically local 2D construction~\cite{aharonov2008adiabatic}.
As we discuss in \cref{sec:discussion}, an interesting open problem is whether a custom circuit-to-Hamiltonian construction that is focused purely on preserving entanglement structure (rather than $\mathsf{QMA}$-hardness) can produce hard instances of the LGSES problem for more physically natural Hamiltonians.

Using a Kitaev clock construction with a binary clock~\cite{kitaev}, we get the following result (see \cref{thm:formal-lgses-binary} for the formal statement).
\begin{theorem}[informal]
\label{thm:LGSES-binary}
Assuming subexponential-time hardness of LWE, LGSES is intractable when the input Hamiltonian is $O(\log n)$-local on $n$ qubits, and the goal is to decide whether the ground state is volume-law or near area-law entangled for the qubits arranged on a 1D line.
\end{theorem}

This result follows relatively straightforwardly from the standard Kitaev clock construction.
There are only two issues that need to be addressed: firstly, the ground state of the Hamiltonian in the Kitaev clock construction is the history state of the circuit, not the output state.
However, our pseudoentanglement construction only provides guarantees on the entanglement structure of the output state.
This problem can be addressed using a ``padding trick''~\cite{nirkhe2018approximate}: we can simply pad the pseudoentanglement circuit with a large (polynomial) number of identity gates at the end.
This will ensure that the history state has most weight on the output state.
Using continuity properties of the von Neumann entropy, this implies that the history state has the desired entanglement structure, too.
The second issue is that we have no control over the entanglement within the clock register of the Hamiltonian.
However, this does not matter for the coarse-grained entanglement structure of the state: since the clock register only has logarithmically many qubits, discarding it only changes the entanglement by $O(\log n)$, which is irrelevant for our $O(\poly\log n)$ vs $\Omega(n)$ entanglement gap.

The Hamiltonian in \cref{thm:LGSES-binary} does not achieve constant locality because the Hamiltonian terms acting on the binary clock register require locality $\log n$.
By using a unary clock instead of a binary clock, we can make the Hamiltonians in \cref{thm:LGSES-binary} have constant locality~\cite{kitaev}. This is also what was used in \cite{gheorghiu2020estimating} to study a Hamiltonian version of their entropy difference problem. However, the clock register now has $\Theta(n)$ qubits, and because it has so many qubits, the analysis from \cref{thm:LGSES-binary} no longer yields the desired entanglement gap. However, if we trace out the clock register and measure entanglement of the remaining mixed state by any operational mixed-state entanglement measure, we show that we still recover a maximal entanglement gap across any cut.
Intuitively, this is because due to the padding trick, after tracing out the clock register the remaining mixed state is close in trace distance to the (pure) output state of the pseudoentanglement circuit.
We refer to \cref{thm:formal-lgses-unary} for the formal statement.

The main downside of the Kitaev clock construction is that the resulting Hamiltonian is not \emph{geometrically} local, i.e.~even though we imagine the qubits as arranged on a 1D line in order to talk about area and volume law entanglement, the Hamiltonian itself has no inherent 1D geometrical structure.
In contrast, most physical Hamiltonians are geometrically local.
To obtain hard instances of the LGSES problem for geometrically local Hamiltonians we can use a more sophisticated 2D clock construction, where we account for time across one of the spatial dimensions instead of needing to add extra clock qubits to the circuit.
We first state the resulting hardness statement for LGSES informally and then briefly sketch the proof.
We refer to~\cref{thm:LGSES-hard-2d-formal} for the formal statement and \cref{sec:2dlearning} for details of the construction.

\begin{theorem} \label{thm:geom_local_intro}
Assuming subexponential-time hardness of LWE, LGSES is intractable when the input Hamiltonian is \emph{geometrically local} on a 2D grid of $n \times \mathrm{poly}(n)$ qudits with constant local dimension $d=O(1)$, and the goal is to decide whether the ground state has entanglement scaling $\mathrm{polylog}(n)$ or $n$ across horizontal cuts.
\end{theorem}

At a high level, the construction of the 2D geometrically local case is similar to before: we take padded versions of our pseudoentanglement circuits and convert them to local Hamiltonians using a 2D circuit-to-Hamiltonian construction~\cite{aharonov2008adiabatic}.
This circuit-to-Hamiltonian construction produces a geometrically local Hamiltonian by dispensing with an explicit clock register.
As a result, the ground state also does not have a clock register and is instead of the form $\sum_{t} V_t \ket{\psi_t}$ (with normalization), where $\ket{\psi_t}$ is the state of the circuit after time step $t$ and $V_t$ are isometries such that $V_t^\dagger V_{t'} = 0$ for any $t\neq t'$.
In other words, similarly to the Kitaev clock construction, the ground state of the Hamiltonian has the form of a history state, with different time steps encoded in mutually orthogonal states.
Since we padded the circuit with identities, we can approximate this ground state by $\sum V_t \ket{\psi_{\rm out}}$ with $\ket{\psi_{\rm out}}$ the output state of our pseudoentanglement circuit. 

The challenge in bounding the entropy of the reduced states of this ``history state'' is that the different time steps are encoded in different bases, specified by the isometries $V_t$.
This is in contrast to the Kitaev construction, where the intermediate states are all encoded in the same basis.
As a result, when we trace out part of the state $\sum V_t \ket{\psi_{\rm out}}$, we get a state that looks very different from just the reduced state of $\ket{\psi_{\rm out}}$.
With the construction of~\cite{aharonov2008adiabatic}, we do not know how to bound the entanglement of these reduced states.

We therefore need to modify the construction from~\cite{aharonov2008adiabatic} to gain better control over the entropy of these reduced states.
We do this by increasing the local dimension of the qudits in order to better keep track of different steps of the circuit execution.
With this modified construction, we can ensure that reduced states of different $V_t \ket{\psi_{\rm out}}$ corresponding to different phases of the circuit execution are, in a certain sense, ``cutwise'' orthogonal (see \cref{lem:cutwise-orthogonality} for details).
The overall reduced state is now a sum of different ``orthogonal''  reduced states, each corresponding to a different phase of the circuit execution.
We can compute the entropy of each of these individual reduced states relatively easily from the entanglement properties of our pseudoentangled states.
Using cutwise orthogonality allows us relate the entropy of the overall state to the entropies of the individual reduced states that we sum over (\cref{lem:CO-entanglement}).
As a result, we can compute the entropy of the overall reduced state even though all the different time steps are encoded in different bases.

\subsection{Related work}
We have already given a detailed discussion of the work of Gheorghiu and Hoban~\cite{gheorghiu2020estimating}, which introduced the idea of public key-pseudoentanglement and gave the first construction, and the work of Aaronson et al.~\cite{bouland2022quantum}, which coined the term pseudoentanglement and gave a private-key construction with maximal entanglement gap across any cut.

The main motivation in~\cite{gheorghiu2020estimating} was to provide a hardness result for the so-called (quantum) entropy difference problem: given two (quantum) circuits, decide whose output has more entropy when acting on a uniformly random input.
If the circuit depth is polynomial, these problems are known to be complete for the complexity classes $\mathsf{QSZK}$ and $\mathsf{SZK}$, respectively~\cite{goldreich1999comparing,ben2008quantum}.
Gheorghiu and Hoban showed that for constant-depth circuits with unbounded fan-out or logarithmic-depth circuits with constant fan-out, both the QED and ED problems are still at least as hard as breaking LWE.
In their proof, the entropy difference between the high- and low-entropy circuits was a single bit.
Our improved pseudoentanglement construction implies that both ED and QED remain LWE-hard with large entropy gaps.\footnote{This result only requires our single-cut pseudoentanglement construction, for which the depth can be made logarithmic as in~\cite{gheorghiu2020estimating}. In fact, as mentioned earlier, an independently updated version of~\cite{gheorghiu2020estimating} also achieves single-cut pseudoentanglement with a large gap, implying the same hardness result for the (Q)ED problem that we obtain from our construction, although their circuits have $\poly(n)$ output qubits for an entropy gap of $n^\delta$ vs $n$, whereas ours only have $O(n)$ output qubits, i.e.~achieve a larger relative gap.
We refer to~\cite{gheorghiu2020estimating} for a more detailed  analysis.}
This is similar in spirit to the classical result that SZK gaps can be amplified \cite{sahai2003complete}. 

Recently, independent and complementary work of Arnon-Friedman, Brakerski, and Vidick~\cite{arnonfriedman2023computational} gave a new definition of pseudoentanglement.
Their definition is private-key and is natural in the context of operational tasks in quantum Shannon theory. Consequently, they focus on operational mixed-state entanglement measures across a single cut and require their states to be efficiently preparable under LOCC. 
In contrast in our work we focus on creating \emph{public-key} pseudoentanglement with different large-scale geometrical structures, which is driven by our applications in Hamiltonian complexity.

Finally, we discuss the relationship between the LGSES problem and existing algorithms for properties of ground states.
While the physics literature on computing properties of ground states is too vast to survey here, we highlight two results closer to computer science.
First, in~\cite{landau2015polynomial} the authors provide a polynomial-time algorithm that, given a classical description of a one-dimensional geometrically local Hamiltonian with constant spectral gap, outputs an MPS description of the ground state.
Therefore 1D constant gapped geometrically local Hamiltonians cannot ``hide'' anything about their ground state entanglement structure.
We note that this algorithm cannot be used on the Hamiltonians we construct in this paper as they are neither 1D geometrically local nor have constant spectral gap.
Therefore our results limit potential further improvements to their algorithm.
Second, the recent result~\cite{huang2022provably} considers the problem of distinguishing phases of matter given labelled examples of states in different phases.
The authors show that if there is a constant spectral gap and the separation between the phases is sufficiently well-conditioned\footnote{In particular, there must be a well-behaved function of few-body observables that separates the phases.}, then a classical algorithm can efficiently learn to distinguish between the phases using information from only few-body measurements.
In condensed matter physics, different qualitative entanglement structures are often associated with different quantum phases of matter; therefore our result also limits potential further improvements to such algorithms, i.e.~it is not possible to relax some of their assumptions e.g.~to gapless phases.
An interesting direction for future work is to make our pseudoentanglement Hamiltonians ``more physical'' to be closer to the assumptions of theses algorithmic settings.
This would help to better delineate the boundary between tractability vs intractability of learning properties of ground states of local Hamiltonians.

\subsection{Discussion and open questions} \label{sec:discussion}

In this work, we have introduced  and studied the Learning Ground State Entanglement Structure (LGSES) problem: given a classical description of a local Hamiltonian, determine qualitative properties of the entanglement of its ground state, e.g.~whether it is area-law or volume-law entangled.
To prove hardness results for this problem, we have related it to a notion that we call public-key pseudoentanglement: low-entanglement states that are computationally indistinguishable from high-entanglement states even when given the state preparation circuit.
Our main technical contribution is to construct public-key pseudoentanglement with (near) area-law vs volume-law scaling assuming the hardness of LWE.

Psedoentanglement is a relatively new idea with many avenues for future work.
We suggest three main directions: (i) improving pseudoentanglement constructions themselves, (ii) strengthening the link between pseudoentanglement and condensed matter physics, and (iii) applications of pseudoentanglement beyond Hamiltonian complexity.
We briefly discuss each in turn.

\begin{enumerate}
\item Our public-key pseudoentanglement construction achieves essentially optimal parameters, but its construction uses a fairly involved iterated entanglement reduction procedure.
In contrast, the private-key construction from~\cite{bouland2022quantum} is very simple and based on subset states.
It would be desirable to have a similarly straightforward construction of public-key pseudoentanglement, too.
Furthermore, as we suggested in our discussion of~\cite{arnonfriedman2023computational}, one can extend our definition of public-key pseudoentanglement to include a trapdoor that allows for efficient distillation of the ``hidden'' entanglement.
It is not obvious how to extend our construction to include this feature.
\item Our hardness results for the LGSES problem use Hamiltonians that differ from the Hamiltonians typically studied in condensed matter physics.
For example, while we do prove a hardness result for the LGSES problem for 2D geometrically local Hamiltonians, this only holds for a certain set of cuts across the system.
We expect that these results can be improved to be closer to the settings studied in condensed matter physics, such as to geometrically local Hamiltonians with more natural entanglement structures and larger spectral gaps,\footnote{Of course, we cannot hope to construct hard instances of the LGSES problem where both the area and volume law Hamiltonians are geometrically local and have constant spectral gap. This is simply because the area law (proven in 1D~\cite{hastings2007area} and in 2D under extra conditions~\cite{anshu2022area}, but widely believed to hold generally) requires \emph{any} such Hamiltonian to have area law entanglement. However, this does not rule out computationally indistinguishable families of Hamiltonians where the area law Hamiltonian has constant gap and the volume law Hamiltonian has inverse polynomial gap, since determining the spectral gap itself is computationally infeasible~\cite{cubitt2015undecidability,bausch2020undecidability}.} which would have implications for the hardness of studying quantum phases of matter.
This may require developing new sorts of clock constructions where the only goal is to preserve entanglement structure of $\textsf{BQP}$ computations rather than to encode more general $\textsf{QMA}$-complete problems.
\item While we have focused on applications in Hamiltonian complexity in this work, pseudoentanglement might be a useful tool for proving hardness results in other domains, too.
For example, recent work~\cite{bostanci2023unitary,arnonfriedman2023computational} has analysed the computational resources required to execute certain tasks from quantum Shannon theory, e.g.~entanglement distillation.
As observed in~\cite{arnonfriedman2023computational}, proving hardness results for such problems is closely related to pseudoentanglement, and we hope that our construction of public-key pseudoentanglement will lead to additional and stronger hardness results in this direction.

Furthermore, public-key pseudoentanglement might also be interesting from a quantum cryptographic point of view, in particular its trapdoor-variant we suggested above.
For example, recent work has focused on finding minimal assumptions for quantum cryptography (see e.g.~\cite{zhandry_talk} and references therein for an overview), and it would be interesting to explore how pseudoentanglement is related to these assumptions.

Finally, it is natural to ask if public-key pseudoentanglement might have applications in quantum gravity. The AdS/CFT correspondence postulates that gravitational theories are dual to quantum mechanical systems, and that the entanglement structure of the quantum system is related to the geometry of the gravitational system \cite{ryu2006holographic}. Our results show that it is difficult to estimate the entanglement of quantum states. In contrast, geometry seems to be easy to determine, which might provide an argument that this duality is exponentially hard to compute, as first suggested in \cite{bouland2019computational}.
Indeed this was part of the motivation for prior works of pseudoentanglement \cite{gheorghiu2020estimating,bouland2022quantum,arnonfriedman2023computational}.
Our public-key extension might allow one to argue about hardness of different versions of the duality, e.g.~the duality remains hard to compute even if given a parent Hamiltonian for the quantum state. 
\end{enumerate}

\subsection*{Acknowledgments}
We thank Rotem Arnon-Friedman, Jordan Docter, Tudor Giurgica-Tiron, Andru Gheorghiu, Hsin-Yuan Huang, Vinod Vaikuntanathan, and Thomas Vidick for helpful discussions. 
B.F.~and S.G.~acknowledge support from AFOSR (FA9550-21-1-0008).
This material is based upon work partially
supported by the National Science Foundation under Grant CCF-2044923 (CAREER) and by the U.S. Department of Energy, Office of Science, National Quantum Information Science Research Centers (Q-NEXT).
This research was also supported in part by the National Science Foundation under Grant No. NSF PHY-1748958.  
A.B., B.F., C.Z., and Z.Z.~were supported in part by the DOE QuantISED grant DE-SC0020360.  
A.B.~and C.Z.~were supported in part by the U.S. DOE Office of Science under Award Number DE-SC0020266. 
A.B.~was supported in part by the AFOSR under grant FA9550-21-1-0392.
C.Z.~was supported in part by the Shoucheng Zhang graduate fellowship.
T.M.~acknowledges support from SNSF Grant No.~200021\_188541 and the ETH Zurich Quantum Center.
U.V.~was supported in part by DOE NQISRC QSA grant FP00010905, NSF QLCI Grant No. 2016245, and MURI Grant FA9550-18-1-0161.

\section{Preliminaries}

\subsection{Notation} \label{sec:notation}
We write $[n]$ for the set $\{1, \dots, n\}$.
For a bitstring $x \in \bits^n$, we denote the $m$ most and least signficant bits by ${\rm MSB}_m(x)$ and ${\rm LSB}_m(x)$, respectively.
We denote the concatenation of strings by $x \parallel y$.
For a set of indices $I \subset [n]$ and bitstrings $x \in \bits^{|I|}$, $y \in \bits^{n-|I|}$ we denote by $z = x \parallel_I y$ the string $z$ that equals $x$ in indices in $I$ and $y$ on indices in $[n]\setminus I$.
We will occasionally think of a bitstring as a $\Z_2$-vector, in which case we denote it as $\vec x$.

For a matrix $A \in \C^{m \times n}$, we denote by $\|A\|_p = \tr{(A^\dagger A)^{p/2}}^{1/p}$ its Schatten $p$-norm.
The 1-norm is also called the trace norm, the 2-norm the Frobenius norm (or Hilbert-Schmidt norm), and the $\infty$-norm the operator norm.

Quantum systems are denoted by capital letters $A, B$, etc.
For a pure state $\ket{\psi}_{AB}$ or a mixed state $\rho_{AB}$ on systems $A$ and $B$, we denote the reduced states on system $A$ by $\psi_A$ and $\rho_{A}$, respectively.
We write quantum circuits as $\mathsf{C} = U_T \cdot U_{T-1} \cdots U_1$, where $U_i$ are elementary gates.
This should be thought of as a list of gates, not simply a large unitary; in particular, inserting identity gates into the circuit does change the circuit (although of course it does not change the unitary implemented by the circuit).
We will use this property of circuits in \cref{section: kitaev clock}.

\subsection{Independent hash functions}
\begin{definition}[$r$-wise independent function family]
A function family $H = \{h_k: [N] \to [M]\}_{k \in \cK}$ indexed by some set of keys $\cK$ is $r$-wise independent if for all distinct $x_1, \dots, x_r \in [N]$, the random variables $h_k(x_1), \dots, h_k(x_{r})$ (for $k \in \cK$ chosen uniformly) are uniform i.i.d.
\end{definition}

The following is a standard result (see e.g.~\cite[Corollary 3.34]{vadhan2012pseudorandomness}):
\begin{lemma} \label{lem:r-wise_exist}
For any $n,m,r \in \N$, there exists an $r$-wise independent function family $H_n = \{h_{k}: \Z_q^n 
 \to \Z_q^m\}_{k \in \cK}$ such that each $k \in \cK$ has length $\poly(n,mr,\log q)$ and given $k \in \cK$, the function $h_k$ can be evaluated in time $\poly(n,m,r,\log q)$. 
\end{lemma}

\subsection{Entropies} \label{sec:entropies}

We recall the basic definitions of quantum entropies.
Throughout, we use the convention that $0 \log 0 = 0$.

\begin{definition}[von Neumann entropy]
The von Neumann entropy of a quantum state $\rho$ is defined as 
\begin{align*}
S(\rho) = - \tr{\rho \log \rho} \,.
\end{align*}
\end{definition}

\begin{definition}[Conditional von Neumann entropy]
The conditional von Neumann entropy of a quantum state $\rho_{AB}$ is defined as
\begin{equation*}
S(\rho_{A|B}) = S(\rho_{AB}) - S(\rho_B).
\end{equation*}
\end{definition}

\begin{definition}[Binary entropy function]
The binary entropy function is defined as
\begin{equation*}
h(x) = - x \log x - (1-x) \log (1-x),
\end{equation*}
for $x \in [0,1]$.
\end{definition}
\subsubsection{Continuity properties}

\begin{lemma}[Continuity of the von Neumann entropy \cite{Fannes1973, Audenaert2007}]
\label{von Neumann continuity}
Let $\rho_{AB}$ and $\sigma_{AB}$ be the density matrix of two $n$-qubit quantum states respectively, each partitioned into subsystems $A$ and $B$, and let
\begin{equation*}
\frac{1}{2} ||\rho_{AB} - \sigma_{AB}||_1 \leq \epsilon. 
\end{equation*}
Then, 
\begin{equation*}
|S(\rho_{AB}) - S(\sigma_{AB})| \leq \epsilon \cdot n + h(\epsilon),
\end{equation*}
where $h(\cdot)$ is the binary entropy function.
\end{lemma}

\begin{lemma}[Continuity of the conditional von Neumann entropy \cite{winter2016}]
\label{continuity_vonneumannentropy}
Let $\rho_{AB}$ and $\sigma_{AB}$ be the density matrix of two $n$-qubit quantum states respectively, each partitioned into subsystems $A$ and $B$, and let
\begin{equation*}
\frac{1}{2} ||\rho_{AB} - \sigma_{AB}||_1 \leq \epsilon. 
\end{equation*}
Then, 
\begin{equation*}
|S(\rho_{A|B}) - S(\sigma_{A|B})| \leq 2 \epsilon \cdot \log |A| + (1 + \epsilon)\cdot h\left(\frac{\epsilon}{1+\epsilon}\right),
\end{equation*}
where $|A|$ is the dimension of the Hilbert space for subsystem $A$ and $h(\cdot)$ is the binary entropy function.
\end{lemma}

\subsection{Entanglement measures}
\label{entanglement measures}

\subsubsection{Pure state entanglement measure}

For pure states on systems $AB$, the entanglement between $A$ and $B$ is quantified using the so-called entanglement entropy, which is simply the von Neumann entropy of the reduced state on either subsystem.
\begin{definition}[Entanglement entropy]
For a \emph{pure} state $\ket{\psi}_{AB}$, the entanglement entropy between systems $A$ and $B$ is defined as $S(\psi_A)$.
Note that this is invariant under swapping $A$ and $B$ since for a pure state $\ket{\psi}_{AB}$, $S(\psi_A) = S(\psi_B)$.
\end{definition}

\subsubsection{Entanglement entropy for phase states}
\begin{definition}[$T$-matrix associated with phase states] \label{def:t-matrix}
Let $s: \bits^n \to \bits$.
For an $n$-qubit phase state 
\begin{align*}
\ket{\psi} = \sum_{x} (-1)^{s(x)} \ket{x}
\end{align*} 
and a subset $X \subseteq [n]$, we define the ``$T$-matrix'' with respect to the cut $X$ as a $\{\pm 1\}^{2^{|I|} \times 2^{n - |I|}}$-matrix with entries
\begin{align*}
T_{ij} = (-1)^{s(i \parallel_X j)} \,,
\end{align*}
where $\parallel_X$ is the ``index string concatenation'' defined in \cref{sec:notation}.
\end{definition}

\begin{lemma} \label{lem:entanglement_from_matrix}
Let $s: \bits^n \to \bits$ and $\ket{\psi} = \sum_{x} (-1)^{s(x)} \ket{x}$.
Then for any cut $X \subseteq [n]$, the entanglement entropy of that cut is bounded by 
\begin{align*}
-\log \left( \norm{\frac{1}{2^{n}} T T^\top}_2 \right) \leq S(\psi_X) \leq \log \rank(T) \,,
\end{align*}
where $T$ is the $T$-matrix of $\ket{\psi}$ across cut $X$.
\end{lemma}
\begin{proof}
This follows from a direct computation, see~\cite[Equation (113)]{bouland2022quantum}.
\end{proof}

\subsubsection{Mixed state entanglement measures} \label{sec:mixedstate_measures}

Quantifying the entanglement of mixed states is much more difficult than for pure states, and several measures have been proposed~\cite{horodecki2009quantum}.
The two most operationally meaningful ones are the distillable entanglement, denoted $E_D(\rho_{AB})$, and the entanglement cost, denoted $E_C(\rho_{AB})$.
Informally, the distillable entanglement of a state $\rho_{AB}$ is the rate at which one can extract EPR pairs from many copies of $\rho_{AB}$ using local operations and classical communication (LOCC) (i.e.~it gives the number of extractable EPR pairs divided by the number of available copies of $\rho$ in the limit of infinitely many copies).
Conversely, the entanglement cost of a state $\rho_{AB}$ is the rate at which one needs to consume EPR pairs to create the state $\rho_{AB}$ using LOCC.
For our purposes the precise definition of these measures does not matter and we refer to \cite[Section XV.A]{horodecki2009quantum} for details.

More importantly for our purposes, distillable entanglement and entanglement cost are extremal measures in the sense that any entanglement measure $E$ that satisfies a number of nautral requirement is bounded as~\cite{Donald2002,christandl2006structure}
\begin{align}
E_D(\rho_{AB}) \leq E(\rho_{AB}) \leq E_C(\rho_{AB}) \,. \label{eqn:natural_entanglement_measure}
\end{align}

We will be interested in bounding these operational entanglement measures.
For this we will rely on non-operational quantities that are easier to deal with mathematically.
Concretely, we will use the entanglement of formation and the coherent information, both defined below.

\begin{definition}[Entanglement of formation]
    The entanglement of formation of a quantum state $\rho_{AB}$ is defined as
    \begin{equation*}
    E_F(\rho_{AB}) = \inf \bigg\{ \sum_{i} p_i ~S(\psi_{i, A}) \bigg\},
    \end{equation*}
    where the infimum is taken over all possible ways in which we can decompose $\rho_{AB}$ as
    \begin{equation*}
    \rho_{AB} = \sum_{i} p_i ~|\psi_i\rangle_{AB} \langle \psi_i|_{AB}.
    \end{equation*}
\end{definition}

\begin{definition}[Coherent information]
\label{coherent}
The coherent information of a quantum state $\rho_{AB}$ is defined as
\begin{equation*}
I(\rho_{AB}) = - S(\rho_{A|B}).
\end{equation*}
\end{definition}

We then get the following  useful bounds.
\begin{lemma} \label{lem:natural_ent_bounds}
Any \emph{natural entanglement measure} $E(\rho_{AB})$ (in the sense of \cref{eqn:natural_entanglement_measure}, see \cite{Donald2002} for details) satisfies 
\begin{align*}
I(\rho_{AB}) \leq E_D(\rho_{AB}) \leq E(\rho_{AB}) \leq E_C(\rho_{AB}) \leq E_C(\rho_{AB}) \,.
\end{align*}
\end{lemma}
\begin{proof}
The lower bound in terms of coherent information was shown in~\cite{winter2005} and the upper bound in terms of entanglement of formation follows from~\cite{hayden2001asymptotic}.
\end{proof}

For pure states, all of these measures coincide with the entanglement entropy:
\begin{lemma}[\cite{wooters}, \cite{horodecki_distillation}, \cite{nielsen_coherent}]
\label{entanglement for pure states}
Let $\rho_{AB}$ be a pure state. Then,
\begin{equation*}
    E_F(\rho_{AB}) = E_C(\rho_{AB}) = E_D(\rho_{AB}) = I(\rho_{AB}) = S(\rho_A).
\end{equation*}
\end{lemma}

\noindent Finally, we will require the following continuity property of entanglement of formation.
\begin{lemma}[Continuity of entanglement of formation \cite{winter2016}]
\label{continuity_formation}
Let $\rho_{AB}$ and $\sigma_{AB}$ be the density matrix of two $n$-qubit quantum states respectively, each partitioned into subsystems $A$ and $B$, and let
\begin{equation*}
\frac{1}{2} ||\rho_{AB} - \sigma_{AB}||_1 \leq \epsilon\,.
\end{equation*}
Additionally, let $\delta = \sqrt{\epsilon (2 - \epsilon)}$ and let $h(\cdot)$ be the binary entropy function. Then,
\begin{equation*}
|E_F(\rho_{AB}) - E_F(\sigma_{AB})| \leq \delta \cdot n + (1 + \delta) \cdot h\left(\frac{\delta}{1+\delta}\right)\,. 
\end{equation*}
\end{lemma}

\subsection{Learning with Errors} \label{sec:lwe}

\begin{definition}[Uniform, Gaussian, and lossy matrix distributions] \label{def:matrix_distributions}
We denote the uniform distribution over $\Z_q^{m \times n}$-matrices by $U_q^{m \times n}$.
We denote by $D_{q,\sigma}^{m \times n}$ the distribution over $\Z_q^{m \times n}$ where each element is an i.i.d.~sample from the discretised Gaussian distribution with width $\sigma$ (see \cite[Section 2.2]{peikert2016decade} for a definition).

We further define the following distribution $L_{q, \ell, \sigma}^{m \times m}$ over $\Z_q^{m \times m}$-matrices:
\begin{enumerate}
\item Sample $B, C \leftarrow U_q^{\ell \times m}$.
\item Sample $E \leftarrow D_{q,\sigma}^{m \times m}$.
\item Output $B^\top \cdot C + E$.
\end{enumerate}
\end{definition}

\begin{definition}
The $\lwe_{n,m,q,\sigma}$ problem is the computational problem of distinguishing the following two distributions:
\begin{enumerate}
\item $(A, \vec u) \leftarrow U_q^{m \times n} \times U_q^{m}$.
\item $(A, A \cdot \vec s + \vec e)$, where $A \leftarrow U_q^{m \times n}$, $\vec s \leftarrow U_q^n$, and $\vec e \leftarrow D_{q,\sigma}^m$.
\end{enumerate}
\end{definition}

We will rely on the following standard assumption in post-quantum cryptography~\cite{regev2009lattices,peikert2016decade}.

\begin{assumption}[Standard LWE assumption] \label{lwe_assumption}
For every constant $\beta > 0$, every $\poly(m)$-time quantum algorithm has $\negl(m)$-advantage in solving the $\lwe_{n,m,q,\sigma}$ problem for $n = m^\beta$, $q = 2^{\poly(n)}$, and $\sigma \geq \frac{q}{\poly(m)} \geq 2 \sqrt{m}$.
\end{assumption}

We will also make use of the following stronger assumption to achieve larger entanglement gaps in our pseudoentanglement construction~\cite{lindner2011better,peikert2016decade}.
\begin{assumption}[Subexponential-time LWE assumption] \label{lwe_assumption_subexp}
There exist constants $\beta_1 > \beta_2 > 0$ such that every $\poly(m)$-time quantum algorithm has advantage at most $2^{-\Omega(\ell^{\beta_1})}$ in solving the $\lwe_{\ell,m,q,\sigma}$ problem for $m = 2^{O(\ell^{\beta_2})}$, $q = 2^{\poly(m)}$, and $\sigma \geq \frac{q}{\poly(m)} \geq 2 \sqrt{m}$.
\end{assumption}

We will use the following immediate consequence of \cref{lwe_assumption}~\cite{goldwasser2010robustness}.

\begin{lemma} \label{lem:lossy_matrix_indist}
~
\begin{enumerate}
\item Under \cref{lwe_assumption} matrices sampled from $U_q^{m \times m}$ and $L_{q, \ell, \sigma}^{m \times m}$ are computationally indistinguishable for $m = \poly(\ell)$, $q = 2^{\poly(m)}$, and $\sigma \geq \frac{q}{\poly(m)} \geq 2 \sqrt{m}$.
\item Under \cref{lwe_assumption_subexp} there exists a constant $\beta > 0$ such that matrices sampled from $U_q^{m \times m}$ and $L_{q, \ell, \sigma}^{m \times m}$ are computationally indistinguishable for $m = 2^{\ell^\beta}$, $q = 2^{\poly(m)}$, and $\sigma \geq \frac{q}{\poly(m)} \geq 2 \sqrt{m}$.
\end{enumerate}
\end{lemma}

\begin{proof}
In $A = B^\top \cdot C + E$, we can view every column $c_i \in \Z_q^\ell$ of $C$ as an LWE secret of length $\ell$.
Then distinguishing $U_q^{m \times m}$ and $L_{q, \ell, \sigma}^{m \times m}$ is equivalent to solving at least one of $m$ instances of $\lwe_{\ell,m,q,\sigma}$ decision problems.
This can be bounded using a union bound for both parts of the lemma:
\begin{enumerate}
\item By \cref{lwe_assumption}, the distinguishing advantage on a single instance of $\lwe_{\ell,m,q,\sigma}$ is $\negl(m)$, so the distinguishing advantage for $m$ copies is at most $m \negl(m)$, which is still negligible in $m$.
\item By \cref{lwe_assumption_subexp}, there are constants $\beta \deq \beta_2 < \beta_1$ such that the distinguishing advantage on a single instance of $\lwe_{\ell,m,q,\sigma}$ is $2^{-\Omega(\ell^{\beta_2})}$, so the distinguishing advantage for $m$ copies is at most $m 2^{-\Omega(\ell^{\beta_2})} = 2^{O(\ell^\beta) - \Omega(\ell^{\beta_2})} = \negl(m)$ since $\beta_2 > \beta$.
\end{enumerate}
\end{proof}

\begin{lemma}\label{lem:random_invertible}
If $q \geq 2^{\ell/2}$ and $\ell \leq m$, then
\begin{align*}
\prs{A \leftarrow U_q^{\ell \times m}}{\rank A = \ell} \geq 1 - O(2^{-\ell/4}) \,.
\end{align*}
\end{lemma}
\begin{proof}
We can imagine $A$ being sampled row-by-row.
Conditioned on the first $r$ rows being linearly independent, the probability of the first $r+1$ rows being linearly independent is $1 - q^r/q^m$, so the probability of all rows being linearly independent is 
\begin{align*}
\prod_{r = 0}^{\ell-1} \left( 1 - \frac{q^r}{q^m} \right) \geq 1 - \frac{\ell}{q} \geq 1 - O(2^{-\ell/4}) \,.
\end{align*}
\end{proof}

We will need the following simple fact about the number of binary vectors that get mapped to $0$ under a random matrix multiplication.
\begin{lemma} \label{lem:random_binary_kernel}
If $q \geq 2^{m/\ell}$, then
\begin{align*}
\prs{A \leftarrow U_q^{\ell \times m}}{ |\{\vec x \in \bits^m \sth A \cdot \vec x = \vec 0\}| \leq 2^\ell } \geq 1 - O(2^{-\ell}) \,.
\end{align*}
\end{lemma}
\begin{proof}
For $\vec x \in \bits^m$ define $\ind_{\vec x}$ as the indicator random variable (over the random choice of $A$) that is 1 if $A \cdot \vec x = \vec 0$.
Then by linearity of expectation, 
\begin{align*}
\E_{A} \left[ |\{\vec x \in \bits^m \sth A \cdot \vec x = \vec 0\}| \right] = \E_{A} \left[ \sum_{\vec x} \ind_{\vec x} \right] = \sum_{\vec x} \E_{A} \left[ \ind_{\vec x} \right] = 1 + \frac{2^m - 1}{q^\ell} \,,
\end{align*}
where the last equality holds because $1_{\vec 0} = 1$ always, and for the remaining $2^m - 1$ choices of $\vec x \neq \vec 0$, the vector $A \cdot \vec x$ is uniform in $\Z_q^\ell$.
Therefore by Markov's inequality, 
\begin{align*}
\prs{A \leftarrow U_q^{\ell \times m}}{ |\{\vec x \in \bits^m \sth A \cdot \vec x = \vec 0\}| > 2^\ell } \leq 2^{-\ell} + \frac{2^m - 1}{(2q)^\ell} = O(2^{-\ell})
\end{align*}
for $q \geq 2^{m/\ell}$.
\end{proof}
\subsection{Hamiltonians and Kitaev clock construction}
\label{section: kitaev clock}
\noindent In this section, we will introduce some notations having to do with some variants of the standard circuit-to-Hamiltonian construction by \cite{kitaev}.

\begin{definition}[History state] \label{def:historystates}
For a circuit $\mathsf{C} = U_K \cdot U_{T-1} \cdots U_1$, the history state $\ket{\Psi_{\mathsf C}}$ is defined as 
\begin{align*}
\ket{\Psi_{\mathsf C}} = \frac{1}{\sqrt{K+1}} \sum_{t = 0}^{T} U_{t} \cdot U_{t-1} \cdots U_{1} |0^n\rangle \otimes |{\rm clock}(t)\rangle.
\end{align*}
Here, ${\rm clock}(t)$ can be any encoding of the integer $t$, typically in binary or unary.
\end{definition}

\begin{definition}[Padded history state Hamiltonian] \label{def:historyham}
For a circuit $\mathsf{C} = U_K \cdot U_{K-1} \cdots U_1$ on $n$ qubits with $K = \poly(n)$, define the $M$-padded circuit as $\mathsf{C}'_M = U_{K+M} \cdots U_{K+1} \cdot U_K \cdots U_1$, where $U_{K+M}, \dots, U_{K+1}$ are identity gates.
A $M$-padded history state Hamiltonian $H_{\mathsf C, M}$ for circuit $C$ is any Hamiltonian with spectral gap at least $1/\poly(n)$ whose unique ground state is $\ket{\Psi_{\mathsf{C}'_M}}$.
\end{definition}

Applying Kitaev's well-known circuit-to-Hamiltonian construction~\cite{kitaev} to a padded circuit immediately shows the existence of padded history state Hamiltonians.
The requirement on the spectral gap follows e.g.~from~\cite{aharonov2008adiabatic}.
\begin{lemma}[\cite{kitaev,aharonov2008adiabatic}] \label{lem:kitaev_ham_exists}
If we choose ${\rm clock}(t)$ in \cref{def:historystates} to be a binary encoding, then for any $K = \poly(n)$, a $\poly(n)$-length circuit $\mathsf C$ has a $O(\log n)$-local $M$-padded history state Hamiltonian $H_{\mathsf{C}, M}$.

If we choose ${\rm clock}(t)$ in \cref{def:historystates} to be a unary encoding, then for any $K = \poly(n)$, a $\poly(n)$-length circuit $\mathsf C$ has a $O(1)$-local $M$-padded history state Hamiltonian $H_{\mathsf{C}, M}$.
\end{lemma}

The useful property of the padded history state Hamiltonian is that its ground state is close to the output state of the circuit $\mathsf C$.
This has already been shown and used in~\cite{nirkhe2018approximate,gheorghiu2020estimating}, so we only sketch the proof.
\begin{lemma}
\label{clock: trace distance closeness}
Let $\mathsf C$ be a $K = \poly(n)$-gate circuit on $n$ qubits.
For any $\eps \geq 1/\poly(n)$ there exists a $M = \poly(n)$ such that the ground state $\ket{\psi_{\mathrm{ground}}}$ of the $M$-padded history state Hamiltonian
$H_{\mathsf{C}, M}$ satisfies 
\begin{align*}
\norm{\proj{\psi_f} - \proj{\psi_{\mathrm{ground}}}}_1 \leq \eps \,,
\end{align*}
where $\ket{\psi_f} = U_K \cdots U_1 \ket{0^n} \ot \frac{1}{\sqrt{M+K+1}} \sum_{t = 0}^{M+K} \ket{{\rm clock}(t)}$
\end{lemma}
\begin{proof}
By construction, 
\begin{align*}
\ket{\psi_{\mathrm{ground}}} = \frac{1}{\sqrt{M+ K + 1}}\sum_{t = 0}^{K} U_{t} \cdot U_{t-1} \cdots U_{1} |0^n\rangle \otimes |{\rm clock}(t)\rangle + \frac{1}{\sqrt{M+ K + 1}} \ket{\psi_f} \ot \sum_{t = K+1}^M \ket{{\rm clock}(t)} \,.
\end{align*}
The lemma then follows by a direct calculation choosing $M = \poly(n)$ sufficiently large such that $\frac{K+1}{M+ K + 1} = O(\eps)$ is sufficiently small.
\end{proof}

\section{Public-key pseudoentanglement: definition and construction}

In this section, we define and construct public-key pseudoentanglement.
Our construction uses a similar idea as in~\cite[Appendix A]{bouland2022quantum}, which is to consider phase states whose phases have been manipulated in a particular way to create high or low entanglement.
The ``manipulation'' of these phases is by means of applying a one-to-one or  many-to-one function (see \cref{sec:single-cut} for details).
We therefore need to construct such \emph{lossy functions} with the appropriate parameters, which we do in \cref{sec:lossy_fct} based on the LWE assumption.

In \cref{sec:single-cut}, we then use these lossy functions to construct indistinguishable families of quantum states where states in one family have high entanglement and states in the other family have low entanglement across a \emph{single fixed bipartition} of the qubits.
In \cref{sec:multicut}, we extend this construction to states that have high or low entanglement \emph{for (almost) every cut on a 1Dimensional line}, i.e.~we imagine the qubits of the state being arranged on a line and consider all bipartitions into left and right qubits.
We show that under the subexponential-time LWE assumption, it is possible to construct pseudoentangled states of this form where either all cuts a have a linear or a polylogarithmic amount of entanglement, which is the largest possible separation, as discussed in \cref{rem:max_min_entanglement}.
In this sense our public-key pseudoentanglement construction is optimal.
We will use this multi-cut construction in \cref{sec:lgses} to show that learning the ground state entanglement structure of (classically described) local Hamiltonians is computationally hard (under the LWE assumption).

\subsection{Construction of lossy functions} \label{sec:lossy_fct}
We define the following rounding function for $\Z_q$ elements.
\begin{definition}
For $q = c p$ with $c \in \N$, divide $\Z_q$ into $p$ consecutive bins.
We define $\roundp{x} \in \Z_p$ as the index of the bin in which $x$ lies.
For a vector $x \in \Z_p^m$, $\roundp{\vec x} \in \Z_p^m$ is defined as the element-wise application of $\roundp{\cdot}$.
\end{definition}

\begin{definition}[Lossy function construction] \label{def:function_construction}
Choose parameters $\ell(m), r(m) \leq \poly(m)$.
Let $p = 2^4$, $q = 2^{m}$, and $\sigma = q/m^3$.
Let $H_m = \{h_k: \Z_p^m \to \Z_2^m\}_{k \in \keyh_m}$ be the $r(m)$-wise independent function family from \cref{lem:r-wise_exist}. 
We define two families of functions $f: \bits^m \to \bits^m$ indexed by key sets $\keyi_m, \keyl_m \subset \Z_q^{m \times m} \times \keyh_m$ as follows:
\begin{itemize}
    \item To sample a key from $\keyi_m$, denoted $k \leftarrow \keyi_m$, sample $A \leftarrow U_q^{m \times m}$ and $k^{\rm hash} \in \keyh_m$ uniformly. Set $k = (A, k^{\rm hash})$.
    \item To sample a key from $\keyl_m$, denoted $k \leftarrow \keyl_m$, sample $A \leftarrow L_{q, \ell, \sigma}^{m \times m}$ and $k^{\rm hash} \in \keyh_m$ uniformly. Set $k = (A, k^{\rm hash})$.
    \item For a key $k = (A, k^{\rm hash}) \in \keyi_m \cup \keyl_m$, the function $f_k: \bits^m \to \bits^m$ is defined as 
    \begin{align*}
    f_k(\vec x) = h_{k^{\rm hash}} \left( \roundp{A \cdot \vec x} \right) \,.
    \end{align*}
\end{itemize}
\end{definition}

The following theorem summarises the relevant properties of our lossy function construction.
\begin{theorem} \label{thm:lossyfunction}
The following properties hold for the function family from \cref{def:function_construction} instantiated with any $r(m) \leq \poly(m)$ and $\poly\log(m) \leq \ell(m) \leq \poly(m)$.
\begin{enumerate}
\item \label{item:almost_inj}(Almost injective functions)
With overwhelming probability over the choice of $A \leftarrow U_q^{m \times m}$, the function family $\{f_{(A, k^{\rm hash})}\}_{k^{\rm hash} \in \keyh_m}$ is $r(m)$-wise independent: 
\begin{align*}
\prs{A \leftarrow U_q^{m \times m}}{\{f_{(A, k^{\rm hash})}\}_{k^{\rm hash} \in \keyh_m} \text{ is $r(m)$-wise independent}} \geq 1 - O(2^{-m}) \,.
\end{align*}
\item \label{item:compressing}(Highly compressing functions)
For keys in $\keyl$, the image of $f_k$ has size at most $2^{\ell^2}$ with overwhelming probability:
\begin{align*}
\prs{k \leftarrow \keyl_m}{|\img f_k| \leq 2^{\ell^2}} \geq 1 - 2^{-\ell/4}
\end{align*}
\item \label{item:indist}(Computational indistinguishability)
Under \cref{lwe_assumption}, if $r(m) \leq \poly(m)$ and $\ell(m) \geq \Omega(m^c)$ for any constant $c>0$, then all efficient quantum algorithms have advantage $\negl(m)$ in distinguishing keys sampled according to $\keyi_m$ and $\keyl_m$.

Similarly, under \cref{lwe_assumption_subexp}, there exists a constant $\delta > 0$ such that if $r(m) \leq \poly(m)$ and $\ell(m) \geq \Omega(\log(m)^{1+\delta})$, then all efficient quantum algorithms have advantage $\negl(m)$ in distinguishing keys sampled according to $\keyi_m$ and $\keyl_m$.
\end{enumerate}
\end{theorem}
\begin{proof}
We prove \cref{item:almost_inj} in \cref{lem:inj_indepedent}.
\cref{item:compressing} is shown in \cref{lem:lossy_img_size}.
\cref{item:indist} follows immediately from \cref{lem:lossy_matrix_indist}.
\end{proof}

\subsubsection{Bounds for almost injective functions}
\begin{lemma} \label{lem:pre-rounding-injective}
For $m \in N$, $p = 2^4$, and $q = c p$ for some $c \in \N$,
\begin{align*}
\prs{A \leftarrow U_q^{m \times m}}{\vec x \mapsto \roundp{A \vec x} \text{ is injective}} \geq 1 - O(2^{-m}) \,.
\end{align*}
\end{lemma}
\begin{proof}
We bound the probability that there exists a pair $\vec x \neq \vec x' \in \bits^m$ for which $\roundp{A \vec x} = \roundp{A \vec x'}$.
Suppose first that $\vec x = 0$. Then $\vec x' \neq 0$, so $\roundp{A \vec x'}$ is a uniform $\Z_p^m$-vector over the random choice of $A$.
Hence using $p = 2^4$ and a union bound,
\begin{align*}
\prs{A \leftarrow U_q^{m \times m}}{\exists \vec x' \sth \roundp{A \vec x} = \roundp{A \vec x'}} \leq 2^{m} 2^{- 4 m} \leq 2^{-m} \,.
\end{align*}
If both $\vec x, \vec x'$ are non-zero, then over the random choice of $A$, $A \vec x$ and $A \vec x'$ are independent uniform $\Z_q^m$-vectors, so $\roundp{A \vec x}$ and $\roundp{A \vec x'}$ are independent uniform $Z_p^m$-vectors.
Hence, for fixed non-zero $\vec x, \vec x'$,
\begin{align*}
\prs{A \leftarrow U_q^{m \times m}}{\roundp{A \vec x} = \roundp{A \vec x'}} \leq 2^{-4 m}.
\end{align*}
The lemma follows by a union bound over the (less than) $2^{2m}$ possible pairs $(\vec x, \vec x')$ and using $p = 2^4$.
\end{proof}

\begin{corollary} \label{lem:inj_indepedent}
\begin{align*}
\prs{A \leftarrow U_q^{m \times m}}{\{f_{(A, k^{\rm hash})}\}_{k^{\rm hash} \in \keyh_m} \text{ is $r$-wise independent}} \geq 1 - O(2^{-m}) \,.
\end{align*}
\end{corollary}
\begin{proof}
Composing an injective function with an $r$-wise independent hash family yields another $r$-wise independent hash family. Therefore, the corollary follows directly from \cref{lem:pre-rounding-injective} and the construction of $f_{(A, k^{\rm hash})}$ in \cref{def:function_construction}.
\end{proof}

\subsubsection{Bounds for compressing functions}

\begin{definition}
For any $M \in \Z_q^{\ell \times m}$, let $f_{M}: \bits^m \to \Z_q^\ell$ be the associated linear map.
We define $\vec v_j(M)$ as the $j$-th element of $\img(f_{M})$ (ordered according to the usual ordering on $\Z_q^\ell$).
For $j \in \{1, \dots, q^{\rank M}\}$ we then define the following collection of subsets: 
\begin{align*}
C_j(M) = f_{M}^{-1}(\vec v_j(M)) \subseteq \bits^m \,.
\end{align*}
If $j > |\img(f_M)|$, then $C_j(M) = \varnothing$.
\end{definition}

\begin{remark} \label{rem:clouds}
Note that in the above definition, $j > |\img(f_M)|$ can indeed occur: this is because $\rank M$ is defined over $\Z_q$, but we have restricted the domain of $f_M$ to bitstrings, so $\img f_M$ only contains $\Z_q$ elements that have a \emph{binary} preimage under $f_M$.
Still, $q^{\rank M}$ is an upper bound on $|\img f_M|$, which is all that will matter for our purposes.
Further, observe that the family $\{C_j(M)\}_j$ defines a partition on $\bits^m$.
\end{remark}

\begin{lemma} \label{lem:lossy_img_size}
\begin{align*}
\prs{k \leftarrow \keyl_m}{|\img f_k| \leq 2^{\ell^2}} \geq 1 - 2^{-\ell}
\end{align*}
\end{lemma}
\begin{proof}
First observe that for any choice of key $k = (A, k^{\rm hash})$, 
\begin{align}
|\img f_k| \leq |\img g_A| \,, \label{eqn:bound_fg}
\end{align}
where $g_A(\vec x) = \roundp{A \cdot \vec x}$.
We will therefore bound the latter.
For this, for $A = B^\top C + E$ chosen as in \cref{def:function_construction} we can bound 
\begin{align}
|\img g_A| \leq \sum_{j=1}^{q^\ell} |\{g_A(\vec x) \;|\; \vec x \in C_j(C) \}| \,. \label{eqn:gcloud_bound}
\end{align}
This is because $\rank C \leq \ell$, so by \cref{rem:clouds} $\{C_j(C)\}_{j = 1, \dots, q^\ell}$ forms a partition of $\bits^m$.

We now fix a $C$ and a $j$ such that $\vec 0 \notin C_j(C)$ and bound the expectation of $|\{g_A(x) \;|\; x \in C_j(C) \}|$ over $A \leftarrow L_{q, \ell}^{m \times m}$.
Recall that $g_A(x)$ ``rounds'' each entry of $A\cdot \vec x$ to a $\Z_p$ element by dividing $\Z_q$ into $p$ equally sized contiguous regions.
We can therefore define $B_{q, p, \beta}$ as the set of $\Z_q$-elements that are within $\beta < \frac{q}{2p}$ of such a cut (where distance is in $\Z_q$, i.e.~mod $q$).
Since there are exactly $p$ cuts, $|B_{q, p, \beta}| = 2 p \beta$.
(This is an equality because for $\beta < \frac{q}{2p}$, the $\beta$-thickenings around the cuts cannot overlap.)

For $i = 1, \dots, m$ we define the random variable (where the randomness is only over $B$ as we fixed a choice of $C$ and $j$ and $C \cdot \vec x$ is the same for all $\vec x \in C_j(C)$)
\begin{align*}
\ind_i = \begin{cases}
1 & \text{if} \quad \exists \vec x \in C_j(C) \sth (B^\top C \cdot \vec x)_i \in B_{q, p, \beta} \,,\\
0 & \text{else}\,.
\end{cases}
\end{align*}
Here, $(B^\top C \cdot \vec x)_i \in \Z_q$ denotes the $i$-th vector entry.
Importantly, the different $\ind_i$ are independent.
To see this, observe that $C \cdot \vec x$ is a fixed non-zero vector for all $\vec x \in C_j(C)$.
Therefore, $(B^\top C \cdot \vec x)_i$ is the inner product between the $i$-th row of $B$ and this fixed vector.
Since the different rows of $B$ are sampled independently, $(B^\top C \cdot \vec x)_i$ are independent, too, and consequently so are $\ind_i$.
By the same reasoning, $(B^\top C \cdot \vec x)_i$ is uniform, so 
\begin{align}
\prs{B}{\ind_i = 1} = \frac{2 p \beta}{q} \,. \label{eqn:ind_prob}
\end{align}

To relate these random variables to $|\{g_A(x) \;|\; x \in C_j(C) \}|$, we set $\beta = 2 m \sigma$ and define a set of ``large'' error matrices $E$: 
\begin{align*}
L = \{E \in \Z_q^{m \times m} \;|\; E \text{ has a row $e_i$ with 1-norm $\norm{e_i}_1 \geq \beta$}\} \,.
\end{align*}
\begin{enumerate}
\item Case $E \in L$.
We will bound the probability of this happening and treat this event separately.
Since each entry to $E$ is a Gaussian with standard deviation $\sigma$, $\norm{e_i}_1$ is a sum of absolute values of i.i.d.~Gaussians, so by a union bound over the $m$ rows and the $2^m$ possible signs for the absolute values within each row,
\begin{align}
\pr{E \in L} \leq m 2^m e^{-\frac{\beta^2}{2 m \sigma^2}} \leq 2^{-m}. \label{eqn:E_in_L}
\end{align}
\item Case $E \notin L$.
In this case, in each direction $i$ the difference between $(B^\top C \cdot \vec x)_i$ and $(A \cdot \vec x)_i$ is at most $\beta$.
Hence, since $B^\top C \cdot \vec x$ is the same for all $x \in C_j(C)$, the only way that the ``cloud'' of points $\{A \cdot \vec x\}_{x \in C_j(C)}$ be split into two distinct sets along dimension $i$ by the rounding procedure is if $\ind_i = 1$.
Therefore, for any $E \notin L$ we find that 
\begin{align*}
|\{g_A(x) \;|\; x \in C_j(C) \}| \leq 2^{\sum_{i} \ind_i} \,.
\end{align*}
We can bound the expectation of this quantity over the random choice of $B$.
Because the different $\ind_i$ are independent:
\begin{align*}
\E_{B \leftarrow U_{q}^{\ell\times m}} \left[ 2^{\sum_{i} \ind_i} \right]
&= \Pi_{i = 1}^m \E_{B \leftarrow U_{q}^{\ell\times m}} \left[ 2^{\ind_i} \right] \\
\intertext{By \cref{eqn:ind_prob}:}
&= \left( 1 + \frac{2 p \beta}{q} \right)^m
\intertext{Since $q \geq 2 p \beta m$ for sufficiently large $m$:}
&\leq e \,.
\end{align*}
\end{enumerate}

We can now combine the pieces of the proof as follows.
We write $E \leftarrow L^c$ for an error matrix $E$ sampled as in \cref{def:function_construction} conditioned on $E \notin L$ and $C \leftarrow U_q^{\ell \times m}|_{\rm good}$ for a uniform $C$ conditioned on 
\begin{align*}
|\{\vec x \in \bits^m \sth C \cdot \vec x = \vec 0\}| \leq 2^\ell \,.
\end{align*}
By \cref{eqn:bound_fg,eqn:E_in_L,lem:random_binary_kernel} and with $A = B^\top C + E$ as before:
\begin{align*}
\prs{k \leftarrow \keyl}{|\img f_k| > 2^{\ell^2}} 
&\leq O(2^{-\ell}) + \prs{B \leftarrow U_q^{\ell \times m}, C \leftarrow U_q^{\ell \times m}|_{\rm full}, E \leftarrow L^c}{|\img g_{A}| > 2^{\ell^2}} \\
\intertext{By Markov's inequality:}
&\leq O(2^{-\ell}) + \frac{1}{2^{\ell^2}} \E_{B \leftarrow U_q^{\ell \times m}, C \leftarrow U_q^{\ell \times m}|_{\rm good}, E \leftarrow L^c} \left[|\img g_{A}| \right] \\
\intertext{By \cref{eqn:gcloud_bound}}
&\leq O(2^{-\ell}) + \frac{1}{2^{\ell^2}} \sum_{j = 1}^{q^\ell} \E_{C \leftarrow U_q^{\ell \times m}|_{\rm good}, E \leftarrow L^c} \left( \E_{B \leftarrow U_q^{\ell \times m}} \left[|\{g_A(\vec x) \;|\; \vec x \in C_j(C) \}| \right]  \right) \\
\intertext{By our analysis of the case $E \notin L$, we know that the expression in parentheses is at most $e$ in case $j$ is such that $\vec 0 \notin C_j(C)$. 
On the other hand, if $\vec 0 \in C_j(C)$ then for that value of $j$ the expression in parentheses is trivially bounded by $2^\ell$ by our definition of $U_q^{\ell \times m}|_{\rm good}$.
Therefore we can bound the whole expression by}
&\leq O(2^{-\ell}) + \frac{e (q^\ell - 1) + 2^\ell}{2^{\ell^2}} \leq O(2^{-\ell}) + O(2^{- \ell (\ell - \log q)}) \leq O(2^{-\ell})
\end{align*}
since $m \geq \ell$ and $\ell - \log q \geq 1$.
\end{proof}

\subsection{Public-key pseudoentanglement across a single cut} \label{sec:single-cut}
We will now use our lossy function construction from \cref{sec:lossy_fct} to construct public-key pseudoentangled states for the middle cut that separates the left and right half of qubits.
In \cref{sec:multicut} we will use the ideas from this section in an iterated way to construct pseudoentangled states for qubits arranged on a line that are pseudoentangled across every cut on the line.
Strictly speaking, all the results in this section follow from the more general analysis in \cref{sec:multicut}.
We spell them out nonetheless because it may be easier for readers to first understand the single-cut construction in detail before moving on to \cref{sec:multicut}.

We begin by defining single-cut public-key pseudoentanglement formally.

\begin{definition}[Public-key pseudoentanglement across a single cut] \label{def:pk_entanglement_singlecut}
A public-key pseudoentangled state ensemble with entanglement gap $(f(n), g(n))$ across cuts $X_n \subset [n]$ consists of two sequences of families of quantum states $\Psi^\lo_n = \{ \ket{\psi_{k}}\}_{k \in \cK^\lo_n}$ and $\Psi^\hi_n = \{ \ket{\psi_{k}}\}_{k \in \cK^\hi_n}$ indexed by key sets $\cK^\lo_n$ and $\cK^\hi_n$ respectively with the following properties: 
\begin{enumerate}
\item \label{item:def:n_qubits}
Every $\ket{\psi_{k}} \in \Psi^\lo_n \cup \Psi^\hi_n$ is an $n$-qubit state.
\item \label{item:def:key_sampling}
Every key $k \in \cK^\lo_n \cup \cK^\hi_n$ has length $\poly(n)$, and there exists an efficient sampling procedure that, given as input $n$ and a label ``high'' or ``low'', outputs a key $k \in \cK^\lo_n$ or $k \in \cK^\hi_n$, respectively.
We write $k \leftarrow \cK^\lo_n$ and $k \leftarrow \cK^\hi_n$ for keys sampled according to this procedure.
\item \label{item:def:eff_prep}
Given $k \in \cK^\lo_n \cup \cK^\hi_n$, the corresponding state $\ket{\psi_{k}}$ is efficiently preparable (without knowing whether $k \in \cK^\lo$ or $k \in \cK^\hi$).
Formally, there exists a uniform polynomial-time circuit family $\{C_n\}$ such that $C_n$ takes as input a key $k \in \cK^\lo_n \cup \cK^\hi_n$ and outputs a state negligibly close to $\ket{\psi_{k}}$.
\item \label{item:def:com_indist}
The keys from $\cK^\lo_n$ and $\cK^\hi_n$ are computationally indistinguishable.
Formally, for all $\poly(n)$-time quantum adversaries $\cA$ that take as input a key $\cK^\lo_n \cup \cK^\hi_n$ and output a single bit: 
\begin{align*}
\Big| \prs{k \leftarrow \cK_n^\lo}{\cA(k) = 0} - \prs{k \leftarrow \cK_n^\hi}{\cA(k) = 0} \Big| = \negl(n) \,.
\end{align*}
\item \label{item:def:entanglement_gap} With overwhelming probability, across the cut $X_n$ states in $\Psi^\lo_n$ have entanglement entropy $\Theta(f(n))$ and states in $\Psi^\hi_n$ have entanglement entropy $\Theta(g(n))$.
Formally, there exist constants $0 < C_1 < C_2$ such that for all sufficiently large $n$,
\begin{align*}
\prs{k \leftarrow \cK_n^\lo}{S((\psi_k)_{X_n}) \in [C_1 f(n), C_2 f(n)]} &\geq 1 - \negl(n)\,,\\
\prs{k \leftarrow \cK_n^\hi}{S((\psi_k)_{X_n}) \in [C_1 g(n), C_2 g(n)]} &\geq 1 - \negl(n) \,.
\end{align*}
Here, $S((\psi_k)_{X_n})$ is the von Neumann entropy of the reduced state of $\ket{\psi_k}$ on the qubits in the set $X_n \subset [n]$.
\end{enumerate}
\end{definition}

\begin{remark} \label{rem:one-sided-notation}
We will frequently abuse notation and use entanglement gaps of the form $(O(f(n)), \Omega(g(n)))$.
By this, we mean that (across a specified cut $X_n$) $\Psi^\lo_n$ has entanglement at most $O(f(n))$ and $\Psi^\hi_n$ has entanglement at least $\Omega(g(n))$.
Formally, this means that for this case \cref{item:def:entanglement_gap} of \cref{def:pk_entanglement_singlecut} has to be modified as follows:
\begin{align*}
\prs{k \leftarrow \cK_n^\lo}{S((\psi_k)_{X_n}) \leq O(f(n))} &\geq 1 - \negl(n)\,,\\
\prs{k \leftarrow \cK_n^\hi}{S((\psi_k)_{X_n}) \geq \Omega(g(n))} &\geq 1 - \negl(n) \,.
\end{align*}
\end{remark}

\begin{remark} \label{rem:max_min_entanglement}
A natural question is what the optimal entanglement gap for pseudoentangled states is.
Clearly, the high-entanglement states can have entanglement at most $g(n) = O(n)$ across any cut, since the entanglement entropy is upper bounded by the number of qubits.
For the low-entanglement states, one can show that the entanglement entropy needs to scale faster than $\log n$, i.e.~$f(n) = \omega(\log n)$.
Otherwise, one could distinguish the low-entanglement states from the high-entanglement states using a variant of the SWAP test.
This was proven in \cite{ji2018pseudorandom} and \cite[Appendix F]{bouland2022quantum} for private-key pseudoentangled states and the same proof applies to the public-key setting, too. 
\end{remark}

We now give a construction of single-cut pseudoentangled states based on our lossy function construction from \cref{sec:lossy_fct}.
As we will show in \cref{thm:single_cut}, these states do indeed form pseudoentangled ensembles in the sense of \cref{def:pk_entanglement_singlecut}.
Under the standard LWE assumption (\cref{lwe_assumption}), we can achieve an entanglement gap of $(O(n^\delta), \Omega(n))$ for any $\delta > 0$, where $n$ is the number of qubits and the cut divides the qubits into two equal halves;
under the stronger subexponential-time LWE assumption (\cref{lwe_assumption_subexp}) we can achieve an entanglement gap of $(O(\poly\log n), \Omega(n))$, which is essentially  optimal as noted in \cref{rem:max_min_entanglement}. 

For simplicity, for the rest of this subsection we always assume that $n$ is even and write $m = n/2$. We will consider the cut that divides the qubits into two sets of size $n$; we could also consider more general cuts and treat them with the same technique, which we do in \cref{sec:multicut}.

\paragraph{Overview.}
The intuition of the construction below is as follows: our high entanglement and low entanglement states will be phase states of the form $\ket{\psi_k} = \sum_{x \in \bits^n} (-1)^{s_k(x)} \ket{x}$, where $s_k: \bits^n \to \bits$ is a function described by some public key $k$.
This is why we call the construction public-key: the states we are considering are fully described by the public key $k$, and given $k$ one can efficiently generate $\ket{\psi_k}$.
Our task then is to construct two families of functions $\{s_k\}_{k \in \cK^\hi}$ and $\{s_k\}_{k \in \cK^\lo}$ indexed by key sets $\cK^\hi$ and $\cK^\lo$, respectively, such that
\begin{enumerate}
\item for $k \in \cK^\lo$, the state $\ket{\psi_k}$ has low entanglement entropy and for $k \in \cK^\hi$, the state $\ket{\psi_k}$ has high entanglement entropy, but
\item it is computationally hard to distinguish keys in $\cK^\hi$ from those in $\cK^\lo$.
\end{enumerate}

To construct these functions, recall from \cref{lem:entanglement_from_matrix} that if we define a matrix $T \in \{\pm 1\}^{2^{n/2} \times 2^{n/2}}$ by $T_{ij} = (-1)^{s_k(i \parallel j)}$, then 
\begin{align*}
-\log \left( \norm{\frac{1}{2^{n}} T T^\top}_2 \right) \leq S(\psi_X) \leq \log \rank(T) \,.
\end{align*}
Therefore, for the high-entanglement states we need to show that $\norm{\frac{1}{2^{n}} T T^\top}_2$ is small, and for the low-entanglement case we need to show that $\rank(T)$ is small.

If one samples $s_k$ from a 4-wise independent function family, one can show that $\norm{\frac{1}{2^{n}} T T^\top}_2$ is indeed exponentially small.
We take this as our starting point and now want to modify this $s_k$ to ``remove'' entanglement to get low-entanglement states, i.e.~we need to reduce the rank of the corresponding matrix $T$.
One natural way to do this is to repeat rows of $T$: if we consider some many-to-one function $h$, then the matrix $\tilde T_{ij} \deq T_{h(i),j}$ is the same as $T$, except instead of all of the rows of $T$, the matrix $\tilde T$ only contains (many copies) of those rows in $T$ whose index is in $\img h$.
Hence, the rank of $\tilde T$ is at most $\img h$.
We can incorporate this row repetition procedure into our definition of our function $s_k$ to get an ``augmented phase function'' whose states $\ket{\psi_k}$ have low entanglement.
The technical term we use to describe removing entanglement by means of row repetitions is \emph{smashing down} the entanglement.

Of course if for the high-entanglement case our $s_k$ was simply a 4-wise independent function, but for the low-entanglement case it was augmented with this row reduction procedure, the two cases would be easily distinguishable given a description of either function.
In other words, we need to smash down the entanglement in a computationally undetectable way.
To solve this problem, in the high-entanglement case we will also augment $s_k$ by a ``fake'' row reduction procedure that is computationally indistinguishable from the procedure in the low-entanglement case, but that does not actually meaningfully reduce the rank of the matrix $T$.
For this, we use the lossy function construction from \cref{sec:lossy_fct}:
in the low-entanglement case, the row repetition function $h$ described above will be a compressing function and therefore reduce the rank of the matrix $T$; for the high-entanglement case, the function $h$ is an almost injective function (\cref{thm:lossyfunction}), which we show does not reduce the Frobenius norm of $TT^\intercal$ too much.
Since the compressing and almost injective functions are computationally indistinguishable (even given a public key description of them), the descriptions of the corresponding augmented phase functions $s_k$ are also computationally indistinguishable.

We now give a formal description of this construction.

\begin{definition} \label{def:single_cut_construction}
Fix a function $f(n)$. (This will be treated as a parameter of the construction.)
Let $H_n = \{h_k: \bits^n \to \bits\}_{k \in \cK_n^4}$ be a 4-wise independent family as given in \cref{lem:r-wise_exist}.
Instantiate the lossy functions from \cref{sec:lossy_fct} with parameters $\ell(m) = \sqrt{f(2m)}$ and $r(m) = 2$. (Recall that $m \deq n/2$.)

We first describe the sampling procedure for the keys $\cK^\lo_n$ and  $\cK^\hi_n$.
\begin{enumerate}
\item To sample $k \in \cK^\lo_n$, sample $k^{\rm rep} \leftarrow \keyl_m$ and $k^{\rm fin} \in \cK^{4}_n$ uniformly.
Set $k = (k^{\rm rep}, k^{\rm fin})$.
\item To sample $k \in \cK^\hi_n$, sample $k^{\rm rep} \leftarrow \keyi_m$ and $k^{\rm fin} \in \cK^{4}_n$ uniformly.
Set $k = (k^{\rm rep}, k^{\rm fin})$.
\end{enumerate}
For $k = (k^{\rm rep}, k^{\rm fin})$, define the \emph{labelling function} $r_k: \bits^n \to \bits^n$ by 
\begin{align*}
r_k(x) = (f_{k^{\rm rep}}(i) \parallel j) \qquad \text{with } i = {\rm MSB}_m(x), j = {\rm LSB}_{m}(x) \,.
\end{align*}
We next define the function $s_k: \bits^n \to \bits$ as 
\begin{align*}
s_k(x) = h_{k^{\rm fin}}(r_k(x)) \,.
\end{align*}
The states $\ket{\psi_k}$ are then given by
\begin{align*}
\ket{\psi_k} = \sum_{x \in \bits^n} (-1)^{s_k(x)} \ket{x} \,.
\end{align*}
\end{definition}

\begin{theorem} \label{thm:single_cut}
~
\begin{enumerate}
\item Under the standard LWE assumption (\cref{lwe_assumption}), for any function $f(n) = n^\delta$ for $\delta > 0$, the state families $\Psi^\lo_n = \{ \ket{\psi_{k}}\}_{k \in \cK^\lo_n}$ and $\Psi^\hi_n = \{ \ket{\psi_{k}}\}_{k \in \cK^\hi_n}$ from \cref{def:single_cut_construction} form a pseudoentangled state ensemble with entanglement gap $(O(f(n)), \Omega(n))$ across the cuts $X_n = [n/2]$.
\item Under the subexpoential-time LWE assumption (\cref{lwe_assumption_subexp}), there exists a function $f(n) = \poly \log n$ such that the state families $\Psi^\lo_n = \{ \ket{\psi_{k}}\}_{k \in \cK^\lo_n}$ and $\Psi^\hi_n = \{ \ket{\psi_{k}}\}_{k \in \cK^\hi_n}$ from \cref{def:single_cut_construction} form a pseudoentangled state ensemble with entanglement gap $(O(f(n)), \Omega(n))$ across the cuts $X_n = [n/2]$.
\end{enumerate}
\end{theorem}

\begin{proof}
We need to check the properties required in \cref{def:pk_entanglement_singlecut}.
\cref{item:def:n_qubits,item:def:key_sampling} are obvious.
\cref{item:def:eff_prep} is also immediate because the function $s_k(x)$ is an efficiently computable classical function, so it follows from the standard phase oracle construction that $\ket{\psi_k}$ can be constructed efficiently, too.
It \cref{thm:lossyfunction} \cref{item:indist} that the key families $\cK^\hi$ and $\cK^\lo$ are computationally indistinguishable for all choices $f(n) = n^\delta$ for $\delta > 0$ under the standard LWE assumption, and that they are indistinguishable for some (sufficiently large) $f(n) = O(\poly\log n)$. 
The proof of \cref{item:def:entanglement_gap} is slightly more involved, so we show it separately as \cref{lem:single_cut_entanglement_gap} below.
\end{proof}

\begin{lemma}\label{lem:single_cut_entanglement_gap}
Using the same setup as in \cref{thm:single_cut}, we have that
\begin{align*}
\prs{k \leftarrow \cK_n^\lo}{S((\psi_k)_{X_n}) \leq f(n)} &\geq 1 - \negl(n) \,,\\
\prs{k \leftarrow \cK_n^\hi}{S((\psi_k)_{X_n}) \geq n/8} &\geq 1 - \negl(n) \,.
\end{align*}
\end{lemma}

\begin{proof}
The states $\ket{\psi_k}$ in \cref{def:pk_entanglement_singlecut} are phase states, so viewing $s_k$ as a function $s_k: \bits^m \to \bits^m \to \bits$, we can define the matrix $T_k \in \{\pm1\}^{2^m \times 2^m}$ as in \cref{def:t-matrix}.
From \cref{lem:entanglement_from_matrix} we then have that 
\begin{align*}
-\log \left( \norm{\frac{1}{2^n} T_k T_k^\top}_2 \right) \leq S((\psi_k)_{X_n}) \leq \log \rank(T_k) \,.
\end{align*}

To prove the first statement in the lemma, we therefore need to bound the probability that $T_k$ has large rank if $k = (k^{\rm rep}, k^{\rm fin}) \in \cK^\lo_n$.
For this, observe that by construction $T_k$ has at most $|\img f_{k^{\rm rep}}|$ distinct rows.
Therefore, recalling that we chose $\ell(m) = \sqrt{f(n)}$ (with $n = 2m$ as before) we get from \cref{thm:lossyfunction} \cref{item:compressing} that
\begin{align*}
\prs{k \leftarrow \cK_n^\lo}{S((\psi_k)_{X_n}) > f(n)} 
&\leq \prs{k \leftarrow \cK_n^\lo}{\log \rank T_k > f(n)} \\
&\leq \prs{k \leftarrow \cK_n^\lo}{|\img f_{k^{\rm rep}}| > 2^{\ell(m)^2}} \leq 2^{-\ell(m)} = \negl(n) \,.
\end{align*}
Here we used that $\ell(m) = \log(m)^{1+\delta}$ for some $\delta > 0$ by \cref{thm:lossyfunction} \cref{item:indist}, so $2^{-\ell(m)} = \negl(n)$.

To prove the second statement in the lemma, we need to bound the probability that $T_k T_k^\top$ has large Frobenius norm.
We will do this by computing the expectation of the Frobenius norm and then applying Markov's inequality.
We begin by expanding the Frobenius norm and using that $T_{ij} \in \{\pm 1\}$:
\begin{align*}
        \E_{k \leftarrow \cK_n^\hi} \left [\left \| T T^{\intercal } \right \|_2^2 \right] & = \sum_{i, j \in \bits^m} \E \left [ \left (  \sum_{l \in \bits^{m}} T_{il}T_{jl}  \right )^2 \right] \\
        & = \sum_{i, j \in \bits^m} \E \left [  \sum_{a, b \in \bits^{m}} T_{ia}T_{ja}T_{ib}T_{jb}  \right] \\
        & = \sum_{i \in \bits^m} \E \left [  \sum_{a, b \in \bits^{m}} T_{ia}T_{ia}T_{ib}T_{ib}  \right]  + \sum_{i, j \in \bits^m} \E \left [  \sum_{a \in \bits^{m}} T_{ia}T_{ja}T_{ia}T_{ja}   \right]  \\
        & \quad \quad + \sum_{\substack{i, j \in \bits^m \\ i \ne j}} \E \left [  \sum_{\substack{a, b \in \bits^{m} \\ a \ne b}} T_{ia}T_{ja}T_{ib}T_{jb}  \right] \\
        & = 2^{3 m} + 2^{3 m} + \sum_{\substack{i, j \in \bits^m \\ i \ne j}} \E \left [  \sum_{\substack{a, b \in \bits^{m} \\ a \ne b}} T_{ia}T_{ja}T_{ib}T_{jb}  \right] \\
        & = 2^{3m + 1} + \sum_{\substack{i, j \in \bits^m \\ i \ne j}} \sum_{\substack{a, b \in \bits^{m} \\ a \ne b}}  \E \left [   (-1)^{s_k(i \parallel a) + s_k(i \parallel b) + s_k(j \parallel a) + s_k(j \parallel b)  }  \right]\,.\numberthis \label{eqn:singleexpsum}
    \end{align*}

Since $s_k(x) = h_{k^{\text{fin}}}(r_k(x))$, and $h_{k^{\text{fin}}}$ is drawn from a $4$-wise independent function family, it is easy to see that the term $\E \left [   (-1)^{s_k(i \parallel a) + s_k(i \parallel b) + s_k(j \parallel a) + s_k(j \parallel b)  }  \right]$ is zero unless the 4 terms in the exponents pair up to cancel.
In particular, this means that a necessary condition for $\E \left [   (-1)^{s_k(i \parallel a) + s_k(i \parallel b) + s_k(j \parallel a) + s_k(j \parallel b)  }  \right] \neq 0$ is that 
\begin{align*}
r_k(i \parallel a) = r_k(i \parallel b) \qquad \text{or} \qquad r_k(i \parallel a) = r_k(j \parallel a) \qquad \text{or} \qquad r_k(i \parallel a) = r_k(j \parallel b)\,.
\end{align*}
Since $a \neq b$, from the definition of $r_k$ we see that the first and last case never occur.
The second case is equivalent to $f_{k^{\rm rep}}(i) = f_{k^{\rm rep}}(j)$.
We can bound this using the properties of $f_{k^{\rm rep}}$.
For this, recall that $k^{\rm rep} = (A, k^{\rm hash})$.
We call a choice of $A$ ``good'' if $\{f_{(A, k^{\rm hash})}\}_{k^{\rm hash}}$ forms a 2-wise independent family.
\begin{align*}
\prs{k \leftarrow \cK_n^\hi}{f_{k^{\rm rep}}(i) = f_{k^{\rm rep}}(j)}
&= \prs{k^{\rm rep} \leftarrow \keyi_m}{f_{k^{\rm rep}}(i) = f_{k^{\rm rep}} (j)}\\
&\leq \prs{k^{\rm rep} = (A, k^{\rm hash}) \leftarrow \keyi_m}{f_{k^{\rm rep}}(i) = f_{k^{\rm rep}} (j) \; | \; A \text{ good}} + \prs{k^{\rm rep} = (A, k^{\rm hash}) \leftarrow \keyi_m}{A \text{ not good}}\\
&\leq O(2^{-m}) \,.
\end{align*}
The last line follows because for a good $A$, the probability that $f_{k^{\rm rep}}(i) = f_{k^{\rm rep}} (j)$ (for $i \neq j$) is $2^{-m}$ by 2-wise independence, and the probability that $A$ is not good is $O(2^{-m})$ by \cref{thm:lossyfunction} \cref{item:almost_inj}.
Therefore, 
\begin{align*}
\sum_{\substack{i, j \in \bits^m \\ i \ne j}} \sum_{\substack{a, b \in \bits^{m} \\ a \ne b}}  \E \left [   (-1)^{s_k(i \parallel a) + s_k(i \parallel b) + s_k(j \parallel a) + s_k(j \parallel b)  }  \right] \leq 2^{4m} \cdot O(2^{-m}) = O(2^{3m}) \,.
\end{align*}
Inserting this into \cref{eqn:singleexpsum}, we get that 
\begin{align*}
\E_{k \leftarrow \cK_n^\hi} \left [\left \| T T^{\intercal } \right \|_2^2 \right] \leq O(2^{3m}) \,.
\end{align*}
Therefore we get from Markov's inequality:
\begin{align*}
\prs{k \leftarrow \cK_n^\hi}{-\log \norm{\frac{1}{2^n} T_{X,k} T_{X,k}^\intercal}_2 < \frac{m}{4}}
&= \prs{k \leftarrow \cK_n^\hi}{\norm{T_{X,k} T_{X,k}^\intercal}_2^2 > 2^{2n - m/2}} \\
&\leq \frac{1}{2^{2n - m/2}} \E_{k \leftarrow \cK_n^\hi} \left[ \norm{T_{X,k} T_{X,k}^\intercal}_2^2 \right] \\
&\leq O(2^{3m - 4m + m/2}) = O(2^{-m/2})\,,
\end{align*}
where we inserted $n = 2m$.
\end{proof}

\subsubsection{Application: quantum entropy difference problem with maximal gap}
One immediate application of our single-cut construction is that we get a maximal stretching of the entropy gap for the quantum entropy difference problem studied in \cite{gheorghiu2020estimating}. To recap, the problem is as follows:
\begin{definition}[Quantum entropy difference problem] Given two circuits $C_1$ and $C_2$ each acting on $n$ qubit, decide whether or not $S(\rho_A) > S(\sigma_A)$, where $\rho_A$ and $\sigma_A$ are the reduced density matrices of the states $|\psi_1\rangle = C_1 |0^n\rangle$ and $|\psi_2\rangle = C_2 |0^n\rangle$ with the last $n/2$ qubits traced out.
\end{definition}

Under the subexponential-time LWE assumption (\cref{lwe_assumption_subexp}), our construction in \cref{sec:single-cut} yields $C_1$ and $C_2$ such that $S(\rho_A) = \Omega(n)$ and $S(\sigma_A) = \text{poly}\log n$, yet the quantum entropy difference problem still remains hard for any polynomially bounded quantum algorithm. From \cref{rem:max_min_entanglement}, it can be argued that this entropy gap is maximal. Previously, the entropy gap in \cite{gheorghiu2020estimating} was only $O(1)$. 
Note that by the same arguments as \cite{gheorghiu2020estimating}, our circuits can also be made shallow depth---that is, they are in $\mathsf{NC}^1$.

A different version of the entropy difference problem, called the Hamiltonian quantum entropy difference problem, also appears in \cite{gheorghiu2020estimating}, where, instead of circuits, the input is the description of two Hamiltonians, and the task is to distinguish between different entanglement entropies across a fixed cut of the ground states of these Hamiltonians. \cite{gheorghiu2020estimating} showed that this problem is LWE-hard for $O(1)$ entropy gaps. Our construction achieves hardness for a maximal entropy gap of $\text{poly}\log n$ vs $n$ for that problem, too, by running our pseudoentanglement construction through a history state Hamiltonian construction in the same manner as in~\cite{gheorghiu2020estimating}.

In \cref{sec:lgses} we will consider a much stronger version of this problem: instead of distinguishing Hamiltonians with ground states whose entanglement differs across a single cut, we will construct indistinguishable Hamiltonians whose ground states have qualitatively different entanglement structures, e.g.~near area-law vs volume-law entanglement.
This will require states that are pseudoentangled across many cuts simultaneously, which we construct in the next section.

\subsection{Area-law public-key pseudoentangled states on a 1D line} \label{sec:multicut}

In this section we will give a (nearly) area-law public-key pseudoentangled states construction on a line based on the row repetition technique we introduced in \cref{sec:single-cut}. 
This means that we will construct public-key pseudoentangled states such that if we imagine the qubits arranged on a line, the entanglement gap is $\poly\log n$ vs $n$ across all cuts that separate the qubits into left and right qubits on the line.

We begin by formally defining this multi-cut version of pseudoentanglement.
The definition is almost identical to \cref{def:pk_entanglement_singlecut}, except that we now need to require an entanglement gap across all cuts on a line simultaneously.
One slight subtlety is that if we consider cuts close to the end of the line, the entanglement will be low simply by virtue of the fact that there are only very few qubits on one side of the cut.
Therefore, in the high-entanglement case, we need to require the entanglement to be at least $g(\text{distance from end of line})$ rather than simply $g(n)$.
Furthermore, very close to the boundary (namely, $O(\log n)$ close), certain properties of our construction break down.
Therefore, we only consider cuts that are at least $\omega(\log n)$ far from the boundary.
Since we are primarily interested in large entanglement gaps of the form $(O(\poly\log n), \Omega(n))$, this small boundary region is of no particular interest to us.
Nonetheless, it is possible to modify our construction to work for such small boundary regions too; we briefly sketch this in \cref{sec:smallcut}.

As in the single-cut case, we use entanglement gaps of the form $(O(f(n)), \Omega(g(n)))$ for pseudoentangled states where we only have an upper bound on the entanglement in the low-entanglement case and a lower bound in the high-entanglement case (see \cref{rem:one-sided-notation} for details).
To simplify the definition slightly, below we state the definition directly for this case; it is straightforward to adapt it to the case where one wants the exact scaling rather than one-sided bounds, but we will not need this for our results.

\begin{definition}[Public-key pseudoentanglement across geometrically local cuts in 1D] \label{def:pk_entanglement_multicut}
A public-key pseudoentangled state ensemble with entanglement gap $(O(f(n)), \Omega(g(n)))$ across geometrically local cuts on a 1D line consists of two sequences of families of quantum states $\Psi^\lo_n = \{ \ket{\psi_{k}}\}_{k \in \cK^\lo_n}$ and $\Psi^\hi_n = \{ \ket{\psi_{k}}\}_{k \in \cK^\hi_n}$ indexed by key sets $\cK^\lo_n$ and $\cK^\hi_n$ respectively that satisfy \cref{item:def:n_qubits,item:def:key_sampling,item:def:eff_prep,item:def:com_indist} from \cref{def:pk_entanglement_singlecut} and the following modified version of \cref{item:def:entanglement_gap} from \cref{def:pk_entanglement_singlecut}:
\begin{enumerate}[label=(v')]
\item \label{item:def:entanglement_gap_multi} For any function $b(n) = \omega(\log n)$, with overwhelming probability, states in $\Psi^\lo_n$ have entanglement entropy $O(f(n))$ and states in $\Psi^\hi_n$ have entanglement entropy $\Omega(g(\text{distance from end of line}))$ for all geometrically local cuts that are at least $b(n)$ far from the end of the line.
Formally,
\begin{align*}
\prs{k \leftarrow \cK_n^\lo}{\forall c \in \{b(n), \dots, n - b(n)\}: \; S((\psi_k)_{[c]}) \leq O(f(n))} &\geq 1 - \negl(n)\,,\\
\prs{k \leftarrow \cK_n^\hi}{\forall c \in \{b(n), \dots, n - b(n)\}: \; S((\psi_k)_{[c]}) \geq \Omega\big(\min(g(c), g(n-c))\big)} &\geq 1 - \negl(n) \,.
\end{align*}
Here, $(\psi_k)_{[c]}$ is the reduced states of $\ket{\psi_k}$ on qubits $(1, \dots, c)$.
\end{enumerate}
\end{definition}

\paragraph{Overview.} We first give an informal overview of our construction and proof.
Recall the high-level idea for single-cut pseudoentanglement: we started from a phase-state whose phases where chosen in a 4-wise independent way and then ``smashed down'' the entanglement across a particular cut by repeating rows in the matrix $T$ associated with the phase state (see \cref{def:t-matrix}) associated with that cut. 
To generalise this to all geometrically local cuts on the $1$-D line, we can sequentially apply this single cut technique. 
Specifically, we again begin with a phase state and consider a cut on right end of the 1D line.
We can then reduce the $T$-matrix across that cut as in the single cut construction.
The modified $T$-matrix corresponds to a modified phase state which is now guaranteed to have high or low entanglement across the cut we chose.
We can now consider the next cut, moving right to left.
Starting from the modified phase state from the first iteration of the procedure, we can again consider the associated $T$-matrix, but now across the second cut, and again use the row-repetition procedure as in the single cut case.
Performing this procedure for all cuts, moving right to left cut-by-cut, we obtain our final states.

The key difficulty in the proof is to show that subsequent row repetition operations do not tamper with the entanglement structure across previous cuts.
For the low entanglement ensemble, one might worry that future row repetition might increase the rank of the $T$-matrix across previous cuts.
We prove that if the cuts are treated in the correct order, this does not happen, later treatments never jeopardize the low entanglement status we established for previous cuts. 
Combining this with the single cut analysis, we conclude that the low entanglement ensemble enjoys the same low entanglement guarantee as in the single cut analysis for each cut.

For the high entanglement ensemble, one might worry about the opposite direction, namely that treating future cuts might \emph{reduce} the rank of the $T$-matrix for previous cuts.
Indeed, this might happen because the functions we use for the row repetition operations in the high entanglement case are not perfectly injective, only ``almost injective'' as shown in \cref{thm:lossyfunction}.
To show that nonetheless the entanglement across all cuts remains high, we prove a similar upper bound on the expectation of the Frobenius norm as the single cut version. 
The analysis is somewhat complicated and relies on the statistical properties of the ``almost injective functions'' from \cref{thm:lossyfunction} and the structure of the iterated row repetition operations.

We begin by describing our construction for the multi-cut case.
This is analogous to the single-cut case (\cref{def:single_cut_construction}), except that we now need to iterate the row repetition operations, so the function $r_k$ from \cref{def:single_cut_construction} needs to be modified and defined recursively.

\begin{definition} \label{def:multi-cut-construction}
Fix a function $f(n)$ (this will be treated as a parameter of the construction). Let $H_n = \{h_k: \bits^n \to \bits\}_{k \in \cK_n^4}$ be a 4-wise independent family as given in \cref{lem:r-wise_exist}. For $m \in \{f(n), f(n)+1, \ldots, n\}$, instantiate the $m$-bit lossy function from~\cref{sec:lossy_fct} with parameters $\ell(m) = \sqrt{f(n)}$ and $r(m) = 2$.

We first describe the sampling procedure for the keys $\cK^\lo_n$ and  $\cK^\hi_n$.
\begin{enumerate}
\item To sample $k \in \cK^\lo_n$, for $m \in \{f(n), f(n)+1, \ldots, n\}$, independently sample $k^{\rm rep}_m \leftarrow \keyl_m$ (see~\cref{def:function_construction}) and $k^{\rm fin} \in \cK^{4}_n$ uniformly.
Set $k = ( k^{\rm rep}_{f(n)},  k^{\rm rep}_{f(n) + 1}, \ldots, k^{\rm rep}_n, k^{\rm fin})$.

\item To sample $k \in \cK^\hi_n$, for $m \in \{f(n), f(n)+1, \ldots, n\}$,  independently sample $k^{\rm rep}_m \leftarrow \keyi_m$ and $k^{\rm fin} \in \cK^{4}_n$ uniformly.
Set $k = (k^{\rm rep}_{f(n)},  k^{\rm rep}_{f(n) + 1}, \ldots, k^{\rm rep}_n, k^{\rm fin})$.
\end{enumerate}
For  $k = (k^{\rm rep}_{f(n)},  k^{\rm rep}_{f(n) + 1}, \ldots, k^{\rm rep}_n, k^{\rm fin})$, define the \emph{labelling functions} $r^{f(n)}_k,  \dots, r^n_k: \{0, 1\}^n \to \{0, 1\}^n$ recursively by 
\[
 r^m_k(x) = \begin{cases}
      x \quad \quad  & m = n+1, \\
      r_k^{m+1}(f_{k^{\rm rep}_m}(i)\parallel j) \quad \quad & f(n) \le m \le n, \mathrm{MSB}_m(x) = i, \mathrm{LSB}_{n - m}(x) = j,
 \end{cases}
\]
where $\mathrm{MSB}_m(x)$ is the first $m$ bits of $x$, $\mathrm{LSB}_m(x)$ is the last $m$ bits of $x$, and $i \parallel j$ is the concatenation of bit strings $i$ and $j$. For simplicity, we define $r_k(x) = r_k^{f(n)}(x)$.

With this notation, for $k = (k^{\rm rep}_{f(n)},  k^{\rm rep}_{f(n) + 1}, \ldots, k^{\rm rep}_n, k^{\rm fin})$, we next define the function $s_k: \{0, 1\}^n \to \bits$ as 
\[
s_k(x) = h_{k^{\rm fin}} \left ( r_k(x) \right ) \,,
\]
The states $\ket{\psi_k}$ are then given by
\begin{align*}
\ket{\psi_k} = \sum_{x \in \{0, 1\}^n } (-1)^{s_k(x)} \ket{x} \,.
\end{align*}
\end{definition}

Our main result is that this construction satisfies the requirements from \cref{def:pk_entanglement_multicut} as summarised by the following theorem.
In particular, we show that under the subexponential-time LWE assumption, our construction achieves an entanglement gap of $\poly\log n$ vs $n$, which is essentially optimal by \cref{rem:max_min_entanglement}.
On a 1D line, this entanglement scaling corresponds to area-law (up to $\poly\log$ factors) vs volume law entanglement, which is why we call this result area-law pseudoentangled states.

\begin{theorem} \label{thm:multi_cut}
~
\begin{enumerate}
\item Under the standard LWE assumption (\cref{lwe_assumption}), for any function $f(n) = n^\delta$ for $\delta > 0$, the state families $\Psi^\lo_n = \{ \ket{\psi_{k}}\}_{k \in \cK^\lo_n}$ and $\Psi^\hi_n = \{ \ket{\psi_{k}}\}_{k \in \cK^\hi_n}$ from \cref{def:multi-cut-construction} form a pseudoentangled state ensemble with entanglement gap $(O(f(n)), \Omega(n))$ across geometrically local cuts in 1D (\cref{def:pk_entanglement_multicut}).
\item Under the subexponential-time LWE assumption (\cref{lwe_assumption_subexp}), there exists a function $f(n) = \poly \log n$ such that the state families $\Psi^\lo_n = \{ \ket{\psi_{k}}\}_{k \in \cK^\lo_n}$ and $\Psi^\hi_n = \{ \ket{\psi_{k}}\}_{k \in \cK^\hi_n}$ from \cref{def:multi-cut-construction} form a pseudoentangled state ensemble with entanglement gap $(O(f(n)), \Omega(n))$ across geometrically local cuts in 1D (\cref{def:pk_entanglement_multicut}).
\end{enumerate}
\end{theorem}

\begin{proof}
We need to check the properties required in \cref{def:pk_entanglement_multicut}.
\cref{item:def:n_qubits,item:def:key_sampling,item:def:eff_prep,item:def:com_indist} hold by the same argument as in \cref{thm:single_cut}.
The proof of \cref{item:def:entanglement_gap_multi} is more involved.
We show the entanglement scaling for the low-entanglement states in \cref{lem:low-multicut} and for the high-entanglement states in \cref{lem:high-multicut}.
\end{proof}

\subsubsection{Low-entanglement states}

\begin{proposition} \label{lem:low-multicut}
Using the same setup as in \cref{thm:multi_cut}, for any $c \in [n]$ we have that 
\begin{align*}
\prs{k \leftarrow \cK_n^\lo}{\forall j = 1, \dots, n: \; S((\psi_k)_{[c]}) \leq O(f(n))} &\geq 1 - \negl(n)\,.
\end{align*}
\end{proposition}

\begin{proof}
If $c < f(n)$ the statement holds trivially because then $(\psi_k)_{[c]}$ has at most $f(n)$ qubits, and therefore at most entropy $f(n)$.
Hence fix any $c \in \{f(n), \dots, n\}$ and consider the cut between the first $c$ qubits and the last $n-c$ qubits.
As in \cref{def:t-matrix}, we can define the matrix $T \in \{\pm 1\}^{2^c \times 2^{n-c}}$ by
\begin{align*}
T_{ij} = (-1)^{s_k(i \parallel j)} = (-1)^{h_{k^{\rm fin}}(r_{k}(i \parallel j))} \,.
\end{align*}
Then by \cref{lem:entanglement_from_matrix}, it suffices to show that the rank of this matrix is bounded by $2^{O(f(n))}$ with overwhelming probability.

Define the functions $\tilde r_k^m: \bits^c \to \bits^{2^{n-c}}$ by 
\begin{align*}
\tilde r_k^m(i) = (r_k^m(i \parallel j))_{j \in \bits^{n-c}} \,.
\end{align*}
In particular, $\tilde r_k(i) \deq \tilde r_k^{f(n)}(i)$ outputs the $i$-th row of $T$ prior to applying the hash function $h_{k^{\rm fin}}$.

From the recursive definition of $r_k$, it is easy to see that for $i \in \bits^c, j \in \bits^{n-c}$ we can write $r_k(i \parallel j) = \xi(r_k^c(\zeta(i)) \parallel j)$ for some functions $\xi, \zeta$ whose details do not matter here.
This implies that $|\img \tilde r_k| \leq |\img \tilde r_k^c|$.
We then have that 
\begin{align*}
\rank(T) \leq |\img \tilde r_k| \leq |\img \tilde r_k^c| \leq |\img f_{k_c^{\rm rep}} |\,,
\end{align*}
where the first inequality holds because the number of distinct rows in $T$ is at most $|\img \tilde r_k|$, the second inequality follows from the argument above, and the third inequality follow from the definition of $r_k^c$ in \cref{def:multi-cut-construction}.
Having reduced the problem to bounding the image size of $f_{k_c^{\rm rep}}$, the lemma now follows exactly like \cref{lem:single_cut_entanglement_gap}.
\end{proof}

\subsubsection{High-entanglement states}
We will now show that despite using the imperfect injective functions from \cref{thm:lossyfunction}, the states in $\Psi_n^\hi$ from our construction in \cref{def:multi-cut-construction} have high entanglement across every geometrically local cut on a 1D line (sufficiently far from the end of the line).
The keys in $\cK^\hi_n$ are of the form $k = (k^{\rm rep}_{f(n)},  k^{\rm rep}_{f(n) + 1}, \ldots, k^{\rm rep}_n, k^{\rm fin})$, and each $k_m^{\rm rep}$ has the form $k_m^{\rm rep} = (A_m, k^{\rm hash}_m)$.
By \cref{item:almost_inj} of \cref{thm:lossyfunction}, we know that with overwhelming probability over the choice of the matrix $A_m$ (as long as $m = \omega(\log n)$), the family $\{f_{(A_m, k^{\rm hash}_m)}\}_{k^{\rm hash}_m}$ is pairwise independent.
We call a matrix $A_m$ ``good'' if this is the case.
We then define a ``good'' key $k \in \cK^\hi_n$ as one for which all the $A_m$-matrices are good: 
\begin{align*}
\cK^\hi_n\big|_{\rm good} = \left\{ (k^{\rm rep}_{f(n)},  k^{\rm rep}_{f(n) + 1}, \ldots, k^{\rm rep}_n, k^{\rm fin}) \in \cK^\hi_n \sth \forall m: k^{\rm rep}_m = (A_m, k^{\rm hash}_m) \text{ has a good $A_m$}\right\} \,.
\end{align*}
We write $k \leftarrow \cK^\hi_n\big|_{\rm good}$ to denote a key sampled as before, but conditioned on getting a good key.

We being by showing a few useful properties of the $r_k^m$-functions.
\begin{lemma}\label{lem:multicut-collision-equivalence}
Instantiate the construction from \cref{def:multi-cut-construction} with any choice of function $f(n)$.
For any $x, y, x', y' \in \bits^{n}$ and $m \in \{f(n), f(n)+1, \dots, n\}$ such that $x \ne y$, $x' \ne y'$ , $\mathrm{LSB}_{n-m}(x) = \mathrm{LSB}_{n-m}(x')$, and $\mathrm{LSB}_{n-m}(y) = \mathrm{LSB}_{n-m}(y')$, the following two properties hold:
\begin{align*}
\prs{k \leftarrow \cK^\hi_n\big|_{\rm good}}{r^m_k(x) = r^m_k(y)} = \begin{cases}
    \prs{k \leftarrow \cK^\hi_n\big|_{\rm good}}{r^{m + 1}_k(x) = r^{m+1}_k(y)} & \text{ if } {\rm LSB}_{n -m}(x) \ne {\rm LSB}_{n -m}(y)  \\
    \frac{1}{2^m} + \frac{2^m - 1}{2^m} \prs{k \leftarrow \cK^\hi_n\big|_{\rm good}}{r^{m + 1}_k(x) = r^{m+1}_k(y)} & \text{ if } {\rm LSB}_{n -m}(x) = {\rm LSB}_{n -m}(y) \, ,
\end{cases} 
\end{align*}
and
\[
 \prs{k \leftarrow \cK^\hi_n\big|_{\rm good}}{r^m_k(x) = r^m_k(y)} = \prs{k \leftarrow \cK^\hi_n\big|_{\rm good}}{r^m_k(x') = r^m_k(y')} \, .
\]

\end{lemma}

\begin{proof}
    For ease of exposition, let $p^m_{x, y} = \prs{k \leftarrow \cK^\hi_n\big|_{\rm good}}{r^m_k(x) = r^m_k(y)}$. Throughout the proof all probabilities are over $k \leftarrow \cK^\hi_n\big|_{\rm good}$. We prove the lemma by induction. 
    
   From the definition of $r^{n+1}(x)$ we have
    \begin{align*}
        p^n_{x, y} & = \sum_{x', y' \in \bits^n} \Pr[f_{k^{\rm rep}_n}(x) = x' \wedge f_{k^{\rm rep}_n}(y) = y'] \cdot p^{n + 1}_{x', y'} \\
        &= \frac{1}{2^n} \\
        &= \frac{1}{2^n} + \frac{2^n - 1}{2^n} \prs{k \leftarrow \cK^\hi_n\big|_{\rm good}}{r^{n + 1}_k(x) = r^{n+1}_k(y)} \,.
    \end{align*}
     The second line holds because $f_{k^{\rm rep}_n}$ is pairwise independent conditioned on good $k$, and $p^{n + 1}_{x', y'} = 0$ for any $x', y' \in \bits^n$ such that $x \ne y$.
    This shows the lemma for $m = n$ since trivially ${\rm LSB}_0(x) = {\rm LSB}_0(y)$ (so we are in the second case for the first equation in the lemma) and we have shown that $p^n_{x, y}$ is independent of $x,y$ (so the second equation in the lemma holds).
     
    Now assuming the claim holds for $m + 1$, we will argue that it also holds for $m$. Let $i = \mathrm{MSB}_{m}(x) , j = \mathrm{MSB}_{m}(y), a = \mathrm{LSB}_{n - m}(x), b = \mathrm{MSB}_{m}(y)$. In particular, we have the following useful observation: if $\mathrm{LSB}_{n - m}(x) = \mathrm{LSB}_{n - m}(y)$ then $\mathrm{LSB}_{n - m - 1}(x) = \mathrm{LSB}_{n - m - 1}(y)$.

\paragraph{Case 1: $a \ne b$ and $i \ne j$. }
    Since $f_{k^{\rm rep}_m}$ is pairwise independent conditioned good $k$, and $\{ k^{\rm rep}_m \}$ are sampled independently for different $m$, we have 

\begin{align*}
    p^m_{x, y} & = \sum_{i', j' \in \bits^m} \Pr[f_{k^{\rm rep}_m}(i) = i' \wedge f_{k^{\rm rep}_m}(j) = j'] \cdot p^{m+1}_{i' \parallel a, j' \parallel b} \\
    & = \frac{1}{2^{2m}} \sum_{i', j' \in \bits^m} p^{m+1}_{i' \parallel a, j' \parallel b} \\
    & = \frac{1}{2^{2m}} \sum_{i', j' \in \bits^m} p^{m+1}_{x, y} \\
    & = p^{m+1}_{x, y}\,,
\end{align*}
where the third line uses the induction hypothesis.
The second property is obvious, by induction hypothesis, $p^{m}_{x, y} = p^{m+1}_{x, y} =  p^{m+1}_{x', y'} = p^{m+1}_{x', y'}$. 

\paragraph{Case 2: $a \ne b$ and $i = j$. }  Similarly, by the fact that $f_{k^{\rm rep}_m}$ is pairwise independent conditioned good $k$, and $\{ k^{\rm rep}_m \}$ are sampled independently for different $m$, we have 

\begin{align*}
    p^m_{x, y} & = \sum_{i' \in \bits^m} \Pr[f_{k^{\rm rep}_m}(i) = i'] \cdot p^{m+1}_{i' \parallel a, i' \parallel b} \\
    & = \frac{1}{2^{m}} \sum_{i'\in \bits^m} p^{m+1}_{i' \parallel a, i' \parallel b} \\
    & = \frac{1}{2^{m}} \sum_{i' \in \bits^m} p^{m+1}_{x, y} \\
    & = p^{m+1}_{x, y}\,,
\end{align*}
where the third line uses the induction hypothesis.
Similarly, the second property holds for this case.

\paragraph{Case 3: $a = b$ and $i \ne j$. }  Similarly, by the fact that $f_{k^{\rm rep}_m}$ is pairwise independent conditioned good $k$, and $\{ k^{\rm rep}_m \}$ are sampled independently for different $m$, we have 

\begin{align*}
    p^m_{x, y} & = \sum_{i', j' \in \bits^m} \Pr[f_{k^{\rm rep}_m}(i) = i' \wedge f_{k^{\rm rep}_m}(j) = j'] \cdot p^{m+1}_{i' \parallel a, j' \parallel a} \\
    & = \frac{1}{2^{2m}} \sum_{i' \in \bits^n} p^{m+1}_{i' \parallel a, i' \parallel a}  + \frac{1}{2^{2m}} \sum_{\substack{i', j' \in \bits^m \\ i' \ne j'}} p^{m+1}_{i' \parallel a, j' \parallel a} \\
    & = \frac{1}{2^m}  + \frac{1}{2^{2m}} \sum_{\substack{i', j' \in \bits^m \\ i' \ne j'}} p^{m+1}_{x, y} \\
    & = \frac{1}{2^m} + \frac{2^m - 1}{2^m} p^{m+1}_{x, y}\,,
\end{align*}
where the third line uses the induction hypothesis.
Similarly, the second property holds for this case.
\end{proof}

\begin{lemma}\label{lem:multicut-collision-probability}
Instantiate the construction from \cref{def:multi-cut-construction} with any choice of function $f(n)$.
For any $d \in [n]$ and $x, y \in \bits^{n}$ such that $\mathrm{LSB}_d(x) = \mathrm{LSB}_d(y)$ and $\mathrm{LSB}_{d + 1}(x) \ne \mathrm{LSB}_{d+1}(y)$, the following holds:
\begin{align*}
\prs{k \leftarrow \cK^\hi_n\big|_{\rm good}}{r_k(x) = r_k(y)} \le \frac{2^{d+1}}{2^n} \,.
\end{align*}
\end{lemma}

\begin{proof}
    For ease of exposition, we again abbreviate $p^m_{x,y} = \prs{k \leftarrow \cK^\hi_n\big|_{\rm good}}{r_k(x) = r_k(y)}$. By \cref{lem:multicut-collision-equivalence} we have  the following recursive relationship:
    \[     
p^{m}_{x, y} =
\begin{cases}
\frac{1}{2^m} + \frac{2^m - 1}{2^m}  p^{m + 1}_{x, y} \quad \quad &  \text{if } \mathrm{LSB}_{n-m}(x) = \mathrm{LSB}_{n-m}(y) \\
p^{m + 1}_{x, y} \quad \quad &  \text{otherwise} \\

\end{cases},
\]
with the boundary condition $p^{n+1}_{x, y} = 0$. 

For $x, y \in \bits^n$ such that $\mathrm{LSB}_d(x) = \mathrm{LSB}_d(y)$ and $\mathrm{LSB}_{d + 1}(x) \ne \mathrm{LSB}_{d+1}(y)$, we have 
\begin{align*}
    p^{f(n)}_{x, y} & \le \sum_{m = f(n)}^n \frac{\mathbb{1}[{\rm LSB}_{n-m}(x) = {\rm LSB}_{n-m}(x) ]}{2^m} \\
    & \le \sum_{m = n - d}^n \frac{1}{2^m} \\
    & \le \frac{2^{d + 1}}{2^n}\,,
\end{align*}
as desired.
\end{proof}

\begin{corollary}\label{cor:multicut-collision-probability-sum}
Instantiate the construction from \cref{def:multi-cut-construction} with any choice of function $f(n)$.
For any $x \in \bits^n$, the following holds:
\[\sum_{\substack{y \in \bits^n  \\ x \ne y}} \Pr[r_k(x) = r_k(y)] \le n.\]
\end{corollary}
\begin{proof}
Observe that for any given $x$, the number of $y \in \bits^n$ such that $\mathrm{LSB}_d(x) = \mathrm{LSB}_d(y)$ but $\mathrm{LSB}_{d+1}(x) \ne \mathrm{LSB}_{d + 1}(y)$ is $\frac{2^n}{2^{d + 1}}$. Therefore, splitting up the sum over $y$ depending on the first bit on which $x$ and $y$ differ, we get from \cref{lem:multicut-collision-probability} that
\begin{align*}
    \sum_{\substack{y \in \bits^n \\  x \ne y}} \Pr[r_k(x) = r_k(y)] 
    &= \sum_{d = 0}^{n-1} \sum_{\substack{y \in \bits^n \\  {\rm LSB}_d(x) = {\rm LSB}_d(y) \\ {\rm LSB}_{d+1}(x) \neq {\rm LSB}_{d+1}(y)}} \Pr[r_k(x) = r_k(y)]\\
    &\le \sum_{d = 0}^{n - 1} \frac{2^n}{2^{d + 1}} \frac{2^{d+1}}{2^n} \\
    &= n\,. \qedhere
\end{align*}
\end{proof}

Now we can use this lemma to upper bound the expectation of the Frobenius norm of the $T$-matrix (see \cref{def:t-matrix}) associated \emph{any} cut $X \subset [n]$. 
While \cref{thm:multi_cut} only requires us to consider geometrically local cuts, the proof works for arbitrary cuts, which we will make use of in \cref{sec:multicut_2d}.
Specifically, we have the following lemma.

\begin{lemma} \label{lem:multi_frob_exp}
Instantiate the construction from \cref{def:multi-cut-construction} with any choice of function $f(n)$ and consider any cut $X \subseteq [n]$ of size $m \deq |X| \leq n/2$.
Let $T_{X,k}$ be the $T$-matrix (\cref{def:t-matrix}) for $\ket{\psi_{k}}$ associated with cut $X$.
Then we have the following bound on the expected squared Frobenius norm:
\begin{align*}
\E_{k \leftarrow \cK^\hi_n\big|_{\rm good}} \left [\left \| T_{X,k} T_{X,k}^{\intercal} \right \|_2^2 \right] \leq 5n \cdot 2^{2n - m} \,.
\end{align*}
\end{lemma}

\begin{proof}
To simplify the notation, we drop the explicit dependence on $X$ and $k$ and simply write $T = T_{X,k}$.
All expectations are taken over $k \leftarrow \cK^\hi_n\big|_{\rm good}$.

The first part of the proof is almost identical to the proof of \cref{thm:single_cut}, except that now we consider cuts that do not necessarily divide the qubits into two equally sized sets.
We spell out the details for completeness.
Expanding the Frobenius norm and using that $T_{ij} \in \{\pm 1\}$:
\begin{align*}
        \E \left [\left \| T T^{\intercal } \right \|_2^2 \right] & = \sum_{i, j \in \bits^m} \E \left [ \left (  \sum_{l \in \bits^{n - m}} T_{il}T_{jl}  \right )^2 \right] \\
        & = \sum_{i, j \in \bits^m} \E \left [  \sum_{a, b \in \bits^{n - m}} T_{ia}T_{ja}T_{ib}T_{jb}  \right] \\
        & = \sum_{i \in \bits^m} \E \left [  \sum_{a, b \in \bits^{n - m}} T_{ia}T_{ia}T_{ib}T_{ib}  \right]  + \sum_{i, j \in \bits^m} \E \left [  \sum_{a \in \bits^{n - m}} T_{ia}T_{ja}T_{ia}T_{ja}   \right]  \\
        & \quad \quad + \sum_{\substack{i, j \in \bits^m \\ i \ne j}} \E \left [  \sum_{\substack{a, b \in \bits^{n - m} \\ a \ne b}} T_{ia}T_{ja}T_{ib}T_{jb}  \right] \\
        & = 2^{2n - m} + 2^{n + m} + \sum_{\substack{i, j \in \bits^m \\ i \ne j}} \E \left [  \sum_{\substack{a, b \in \bits^{n - m} \\ a \ne b}} T_{ia}T_{ja}T_{ib}T_{jb}  \right] \\
        & = 2^{2n - m} + 2^{n + m} + \sum_{\substack{i, j \in \bits^m \\ i \ne j}} \sum_{\substack{a, b \in \bits^{n - m} \\ a \ne b}}  \E \left [   (-1)^{s_k(i \parallel a) + s_k(i \parallel b) + s_k(j \parallel a) + s_k(j \parallel b)  }  \right]\,. \numberthis \label{eqn:expsum}
    \end{align*}

Since $s_k(x) = h_{k^{\text{fin}}}(r_k(x))$, and $h_{k^{\text{fin}}}$ is drawn from a $4$-wise independent function family, it is easy to see that the term $\E \left [   (-1)^{s_k(i \parallel a) + s_k(i \parallel b) + s_k(j \parallel a) + s_k(j \parallel b)  }  \right]$ is zero unless the 4 terms in the exponents pair up to cancel.
In particular, this means that a necessary condition for $\E \left [   (-1)^{s_k(i \parallel a) + s_k(i \parallel b) + s_k(j \parallel a) + s_k(j \parallel b)  }  \right] \neq 0$ is that 
\begin{align*}
r_k(i \parallel a) = r_k(i \parallel b) \qquad \text{or} \qquad r_k(i \parallel a) = r_k(j \parallel a) \qquad \text{or} \qquad r_k(i \parallel a) = r_k(j \parallel b)\,.
\end{align*}
In \cref{thm:single_cut}, the probability of this event was straightforward to bound using pairwise independence.
Here, this direct argument is replaced by using \cref{cor:multicut-collision-probability-sum}.
For the first case, we can bound 
\begin{align*}
\sum_{\substack{i, j \in \bits^m \\ i \ne j}} \sum_{\substack{a, b \in \bits^{n - m} \\ a \ne b}}  \Pr[r_k(i \parallel a) = r_k(j \parallel a)] 
&= \sum_{i \in \bits^m} \sum_{\substack{a, b \in \bits^{n - m} \\ a \ne b}}  \left ( \sum_{\substack{j \in \bits^m \\ i \ne j}}  \Pr[r_k(i \parallel a) = r_k(j \parallel a)] \right) \\
&\leq \sum_{i \in \bits^m} \sum_{\substack{a, b \in \bits^{n - m} \\ a \ne b}}  \left ( \sum_{\substack{y \in \bits^n \\ y \ne (i \parallel a)}}  \Pr[r_k(i \parallel a) = r_k(y)] \right) \\
&\leq 2^{2n - m} \cdot n \,.
\end{align*}
For the second line we have added terms to the sum in parentheses, which can only increase its value.
For the third line we have used \cref{cor:multicut-collision-probability-sum} to deduce that the term in parentheses is at most $n$ and then counted the number of terms in the outer sum.

By analogous derivations, we also obtain 
\begin{align*}
\sum_{\substack{i, j \in \bits^m \\ i \ne j}} \sum_{\substack{a, b \in \bits^{n - m} \\ a \ne b}}  \Pr[r_k(i \parallel a) = r_k(i \parallel b)] &\leq 2^{n+m} \cdot n \,,\\  
\sum_{\substack{i, j \in \bits^m \\ i \ne j}} \sum_{\substack{a, b \in \bits^{n - m} \\ a \ne b}} \Pr[r_k(i \parallel a) = r_k(j \parallel b)] &\leq 2^n \cdot n \,.
\end{align*}

Using the trivial bound $\E \left [   (-1)^{s_k(i \parallel a) + s_k(i \parallel b) + s_k(j \parallel a) + s_k(j \parallel b)}\right] \leq 1$, we can therefore bound the remaining sum in \cref{eqn:expsum}
\begin{align*}
&\sum_{\substack{i, j \in \bits^m \\ i \ne j}} \sum_{\substack{a, b \in \bits^{n - m} \\ a \ne b}}  \E \left [   (-1)^{s_k(i \parallel a) + s_k(i \parallel b) + s_k(j \parallel a) + s_k(j \parallel b)  }  \right] \\
&\qquad \le \sum_{\substack{i, j \in \bits^m \\ i \ne j}} \sum_{\substack{a, b \in \bits^{n - m} \\ a \ne b}}  \Bigg(\Pr[r_k(i \parallel a) = r_k(j \parallel a)]  +  \Pr[r_k(i \parallel a) = r_k(i \parallel b)]  +  \Pr[r_k(i \parallel a) = r_k(j \parallel b)] \Bigg) \\
&\qquad \leq n \cdot 2^{2n - m} + n \cdot 2^{n+m} + n \cdot 2^n \,.
\end{align*}
Inserting this into \cref{eqn:expsum} we get that
\[
    \E \left [\left \| T T^{\intercal } \right \|_2^2 \right] \le 2^{2n - m} + 2^{n + m} + n \cdot 2^{2n - m} + n \cdot 2^{n+m} + n \cdot 2^n \leq 5n \cdot 2^{2n - m}
\]
since we assumed that $m \leq n/2$.
\end{proof}

\begin{proposition} \label{lem:high-multicut}
Using the same setup as in \cref{thm:multi_cut}, for any cut $X \subseteq [n]$ of size $m \deq |X| = \omega(\log n)$ we have that 
\begin{align*}
\prs{k \leftarrow \cK_n^\hi}{\forall j = 1, \dots, n: \; S((\psi_k)_{X}) \ge \frac{1}{4} \min(m, n - m)} &\geq 1 - \negl(n)\,.
\end{align*}
\end{proposition}
\begin{proof}
Without loss of generality, we can assume that $m = |X| \leq n/2$; otherwise we can simply consider the complementary cut $[n] \setminus X$ and use the fact that $S((\psi_k)_{X}) = S((\psi_k)_{[n] \setminus X})$.
Let $T_{X,k}$ be the $T$-matrix (\cref{def:t-matrix}) for $\ket{\psi_{k}}$ associated with cut $X$.
By \cref{lem:entanglement_from_matrix} it suffices to show that 
\begin{align*}
\prs{k \leftarrow \cK_n^\hi}{-\log \norm{\frac{1}{2^n} T_{X,k} T_{X,k}^\intercal}_2 < \frac{m}{4}} \leq \negl(n) \,.
\end{align*}
The event that $-\log \norm{\frac{1}{2^n} T_{X,k} T_{X,k}^\intercal}_2 < \frac{m}{4}$ is equivalent to the event that $\norm{T_{X,k} T_{X,k}^\intercal}_2^2 > 2^{2n - m/2}$.
We can bound the probability of the latter as follows.
From \cref{item:almost_inj} in \cref{thm:lossyfunction} and a union bound we know that if we sample a key $k \leftarrow \cK_n^\hi$, the probability that this key is \emph{not} contained in $\cK_n^\hi\big|_{\rm good}$ at most $n \cdot 2^{-f(n)}$, which is negligible in $n$ since we chose $f(n) = \omega(\log n)$ in \cref{thm:multi_cut}.
Therefore we get that
\begin{align*}
\prs{k \leftarrow \cK_n^\hi}{-\log \norm{\frac{1}{2^n} T_{X,k} T_{X,k}^\intercal}_2 < \frac{m}{4}}
&\leq \prs{k \leftarrow \cK_n^\hi\big|_{\rm good}}{-\log \norm{\frac{1}{2^n} T_{X,k} T_{X,k}^\intercal}_2 < \frac{m}{4}} + \negl(n) \\
&= \prs{k \leftarrow \cK_n^\hi\big|_{\rm good}}{\norm{T_{X,k} T_{X,k}^\intercal}_2^2 > 2^{2n - m/2}} + \negl(n) \\
&\leq \frac{1}{2^{2n - m/2}} \E_{k \leftarrow \cK_n^\hi\big|_{\rm good}} \left[ \norm{T_{X,k} T_{X,k}^\intercal}_2^2 \right] + \negl(n) \\
&\leq 5n \cdot 2^{2n - m - 2n + m/2} + \negl(n) \\
&= 5n \cdot 2^{-m/2} + \negl(n) \\
&= \negl(n) \,.
\end{align*}
Here, the third line uses Markov's inequality, the fourth line uses \cref{lem:multi_frob_exp}, and the last line holds because $m = \omega(\log n)$.
\end{proof}

\subsubsection{Pseudorandom states and entropy lower bounds for cuts with sub-logarithmic sizes}
\label{sec:smallcut}
Note that the current high entanglement construction may not always satisfy volume law entanglement for partitions where one side is extremely small (formally speaking, of size $O(\log n)$). Although this limitation does not impede the applications presented in this paper, we briefly discuss adjustments needed in our construction guaranteeing that volume law entanglement even holds for these small cuts. 

Intuitively, when the size of the subsystem is at most $O(\log n)$, adversaries can approximate entanglement entropy up to an inverse polynomial error. On the other hand, a logarithmically-sized subsystem of Haar random states concentrates very strongly around the maximally mixed state. Therefore, any pseudorandom state must have volume law entanglement for cut $X$ of size $|X| = O(\log n)$. 

While the construction we have described so far is not pseudorandom,\footnote{By ``pseudorandom'' we mean that the state families \emph{without} the public key are indistinguishable from Haar random states. Naturally, if the circuit to prepare the state is made public, it is easy to distinguish the state from Haar random states since generic Haar random states have no short state preparation circuit.} a simple modification can achieve this: according to \cite[Appendix A]{bouland2022quantum}, it suffices to modify $k^{\textrm{rep}}_i$ and $k^\textrm{fin}$ such that they are pseudorandom functions (PRFs). The construction utilizes two types of functions: $k$-wise independent functions and LWE functions, with the latter being PRFs based on the standard LWE assumption. Consequently, the transformation involves converting $k$-wise independent families into families that are both $k$-wise independent and pseudorandom (for a specific construction, refer to \cite[Appendix A]{bouland2022quantum}). Notably, since both our low-entanglement and high-entanglement constructions employ the same $k$-wise independent families, these modifications do not jeopardize the public-key indistinguishability of two constructions.

\subsection{Area-law public-key pseudoentangled states on a 2D grid} \label{sec:multicut_2d}
We can easily generalise the 1D construction from \cref{sec:multicut} to a system of qubits arranged on a 2D grid.
This is the same construction as in \cite[Appendix D.5]{bouland2022quantum}, so we only provide a short sketch.
Let $\ket{\psi_k}$ be an $n$-qubit state from a pseudoentangled state ensemble.
We can arrange the qubits of this state on an $\sqrt{n} \times \sqrt{n}$ grid as shown in \cref{fig:snake}.
\begin{figure*}[ht!]
\centering
\includegraphics[width=0.2\textwidth]{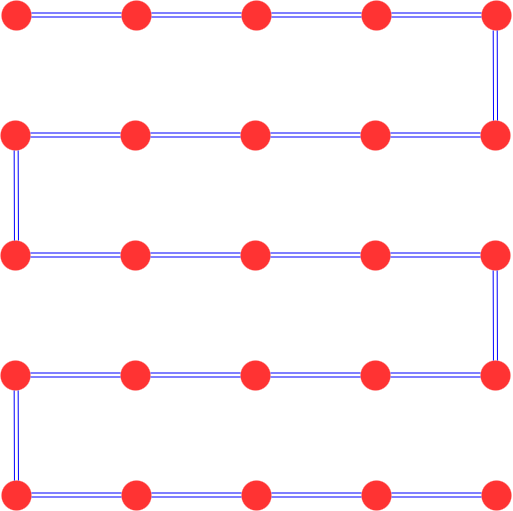}
\caption{Arranging an $n$-qubit state on a $\sqrt{n}\times \sqrt{n}$ grid.}
\label{fig:snake}
\end{figure*}
Now consider a contiguous 2D subregion $R$ of this $\sqrt{n} \times \sqrt{n}$ grid.
Let $|R|$ be the size of $R$ (i.e.~the number of qubits in $R$) and $|\partial R|$ the size of the boundary of $R$.
Unfolding the ``snake'', this region $R$ corresponds to a (not necessarily geometrically local) cut in 1D.

For the pseudoentangled state ensembles we constructed in \cref{thm:multi_cut}, a high-entanglement state $\ket{\psi_k}$ has entanglement entropy $\Omega(|R|)$ for any (sufficiently large) cut $R$, even if the cut is not 1D geometrically local (\cref{lem:high-multicut}).
This means that arranged on a 2D grid, the high-entanglement states from \cref{thm:multi_cut} exhibit volume law entanglement scaling.

Conversely, consider a low-entanglement state $\ket{\psi_k}$ from the construction in \cref{thm:multi_cut}.
From the geometry of \cref{fig:snake} it is easy to see that a region $R$ corresponds to a 1D cut that divides the qubits into at most $O(|\partial R|)$ contiguous regions; this is because the boundary of the region $R$ can cut the ``snake'' at most $O(|\partial R|)$ times.
For each of these $O(|\partial R|)$ cuts in 1D, we know from \cref{thm:multi_cut} that $\ket{\psi_k}$ has entanglement entropy at most $O(\poly \log n)$ across that cut.
Using subadditivity of entanglement entropy, it then follows that the entanglement entropy of the region $R$ is at most $O(|\partial R| \cdot \poly \log n)$, which corresponds to area-law scaling in two dimensions (up to polylogarithmic factors).

\section{Computational hardness of learning ground state entanglement structure} \label{sec:lgses}
Our public-key pseudoentanglement constructions can be leveraged to construct Hamiltonians such that it is hard to learn the entanglement structure of their ground state. This is what we will discuss in the next sections. Specifically, we will study variants of the following problem, which we define somewhat informally.

\begin{definition}[Learning Ground State Entanglement Structure (LGSES) problem]
\label{entanglement structure}
Given a classical description of a $k$-local Hamiltonian $H$ on $n$ qubits with spectral gap at least $\frac{1}{\text{poly}(n)}$, decide whether the ground state of $H$ has entanglement structure $\mathsf{A}$ or $\mathsf{B}$?
Here, $\mathsf{A}$ and $\mathsf{B}$ should be two qualitatively different, pre-specified entanglement structures, e.g.~near are-law and volume law entanglement.
\end{definition}

We will see three different variants of this problem for three different types of local Hamiltonians and correspondingly three slightly different entanglement structures. 
We will progressively make our constructions more local---in some sense, more local corresponds to more physical Hamiltonians---but we will pay a slight price in terms of how straightforwardly the entanglement structures can be described.

\begin{enumerate}
\item In \cref{first construction}, we will study the hardness of LGSES for $O(\log n)$-local Hamiltonians on $n$ qubits arranged in a 1D line. The two entanglement structures to distinguish between will be $\text{poly}\log n$ vs $\Omega(n)$ entanglement across geometrically local cuts in 1D. In other words, we are asked to distinguish 1D near area-law vs volume-law entanglement.
\item In \cref{sec:hamiltonian_kitaev_mixed}, we will improve upon the locality of the Hamiltonian and study the hardness of LGSES for $O(1)$-local Hamiltonians on $n$ qubits arranged in a 1D line. 
However, the entanglement structure will be slightly more complicated: we will consider the reduced states of ground states on a specific subsystem and ask whether this has 1D near area-law or volume-law entanglement structure for a mixed state measure of entanglement.
\item In \cref{sec:2dlearning}, we will study the hardness of LGSES for 2-local Hamiltonians on a 2D grid of size $n \times \poly(n)$ and with constant local dimension, where all Hamiltonian terms are \emph{geometrically local} (i.e.~only nearest neighbors on the grid can interact). The two entanglement structures to distinguish will be entanglement entropy $O(\poly\log n)$ vs $\Omega(n)$ across horizontal cuts through the grid.
\end{enumerate}

\subsection{1D Hamiltonians with $\log n$-locality and pure states}
\label{first construction}
In this section, we will show how to obtain two families of $\log n$-local Hamiltonians, one whose ground state has $\text{poly}\log n$ entanglement scaling and the other whose ground state has $\Omega(n)$ entanglement scaling across geometrically local cuts in 1D, such that given the description of one of these Hamiltonians it is computationally hard to decide which family it belongs to.

We will start with the public-key pseudoentangled state constructions in \cref{sec:multicut}, use the padded circuit-to-Hamiltonain construction of \cref{section: kitaev clock}, and then use the trace distance closeness property of \cref{clock: trace distance closeness} to show that the entanglement structures of the ground states of these Hamiltonians resemble the entanglement structure of the public-key pseudoentangled states.

\begin{theorem}
\label{thm:formal-lgses-binary}
For every $n \in \N$, there exist two families $\mathcal{H}^\lo_n$ and $\mathcal{H}^\hi_n$ of $O(\log n)$-local Hamiltonians on $(n+O(\log n))$ qubits arranged on a 1D line with spectral gap $\Omega(1/\poly(n))$ and efficient procedures that sample (classical descriptions of) Hamiltonians from either family (denoted $H \leftarrow \mathcal{H}^\lo_n$ and $H \leftarrow \mathcal{H}^\hi_n$) with the following properties:
\begin{enumerate}
\item Hamiltonians sampled according to $H \leftarrow \mathcal{H}^\lo_n$ and $H \leftarrow \mathcal{H}^\hi_n$ are computationally indistinguishable under \cref{lwe_assumption_subexp}. \label{item:hambin_indist}
\item  With overwhelming probability, the ground states of Hamiltonians $H \leftarrow \mathcal{H}^\lo_n$ have 1D near area-law entanglement and Hamiltonians $H \leftarrow \mathcal{H}^\lo_n$ have 1D volume-law entanglement.
Formally, this means that for geometrically local cuts in 1D of size $r = \omega(\log n)$, the ground states of the Hamiltonians have entanglement entropy $O(\poly\log n)$ or $\Omega(\min(r, n-r))$, respectively.
\label{item:hambin_ent}
\end{enumerate}
\end{theorem}

\begin{proof}
 Let $\Psi^\lo_n = \{ \ket{\psi_{k}}\}_{k \in \cK^\lo_n}$ and $\Psi^\hi_n = \{ \ket{\psi_{k}}\}_{k \in \cK^\hi_n}$ be two ensembles of public key pseudoentangled states across geometrically local cuts from \cref{sec:multicut}. 
 For a key $k \in \cK^\lo_n \cup \cK^\hi_n$, we denote by $H_k$ the $O(\log n)$-local $K$-padded history state Hamiltonian from \cref{lem:kitaev_ham_exists} with a binary clock, for $K = \poly(n)$ to be chosen later.
 We then choose $\mathcal{H}^\lo_n = \{H_k\}_{k \in \cK^\lo_n}$ and $\mathcal{H}^\hi_n = \{H_k\}_{k \in \cK^\hi_n}$.
 It follows from our pseudoentanglement construction and the padded history state Hamiltonian that these are efficiently sampleable families of $O(\log n)$-local Hamiltonians, noting that each term has a classical description (as a list of matrix entries) of size $\poly(n)$ and there are $\poly(n)$ terms, so the full Hamiltonians has an explicit classical description of size $\poly(n)$, too.
 The statement about the spectral gap follows from \cref{lem:kitaev_ham_exists}.
 \cref{item:hambin_indist} follows directly from the indistinguishability of $\cK^\lo_n$ and $\cK^\hi_n$ (\cref{item:def:com_indist} in \cref{def:pk_entanglement_singlecut}).
 
To prove \cref{item:hambin_ent}, first consider a Hamiltonian $H \leftarrow \cH^\lo_n$.
Let $|\psi_{\mathrm{ground}}\rangle$ be its ground state and define
\begin{equation*}
|\psi_{f} \rangle = \ket{\psi} \otimes \frac{1}{\sqrt{T+1}} \sum_{t=0}^{T} \ket{{\rm clock}(t)}\,,
\end{equation*}
where $\ket{\psi}$ is the pseudoentangled state for which $H$ is the history state Hamiltonian and $T$ is the number of gates in the (padded) circuit (see \cref{section: kitaev clock}).
Note that from \cref{clock: trace distance closeness},
\begin{equation}
\label{Trace_distance_1D_logn}
\norm{\proj{\psi_f} - \proj{\psi_{\mathrm{ground}}}}_1 = {O}\left( \frac{1}{\text{poly}(n)}\right)\,,
\end{equation}
for a sufficiently large choice of padding $K = \poly(n)$.

Let $(A, B)$ be a geometrically local bipartition of the $n + O(\log n)$ with $|A|, |B| = \omega(\log n)$.
Since we are considering geometrically local cuts with $|B| = \omega(\log n)$ and there are only $O(\log n)$ qubits, for a suitable choice of implicit constants in $|B| = \omega(\log n)$ all of the clock qubits are contained in $B$ and therefore do not contribute to the entanglement across the cut. 
Hence, 
\begin{align*}
S((\psi_f)_A) = O(\poly\log n) \,.
\end{align*}
Using the continuity property of the von Neumann entropy (\cref{continuity_vonneumannentropy}) and \cref{Trace_distance_1D_logn}, if we choose the padding $K = \poly(n)$ sufficiently large we get that 
\begin{align*}
S((\psi_{\mathrm{ground}})_A) \leq S((\psi_f)_A) + O\left(\frac{1}{\text{poly}(n)}\right) = O\left(\text{poly}\log n\right)
\end{align*}
as desired.

By applying a similar argument, when we consider $H \leftarrow \cH_n^\hi$ and a bipartition $(r, n-r)$,
\begin{equation*}
\label{low}
S((\psi_{\mathrm{ground}})_A) = \Omega(\text{min}(r, n-r)).
\end{equation*}
This completes the proof.
\end{proof}

\begin{remark}
\label{first_remark}
Under the standard LWE assumption (\cref{lwe_assumption}) instead of \cref{lwe_assumption_subexp}, \cref{thm:formal-lgses-binary} still holds, but with the smaller entanglement gap $O(n^\delta)$ vs $\Omega(n)$ for any $\delta > 0$.
This mirrors directly the statement in \cref{thm:multi_cut}.
\end{remark}

\subsection{1D Hamiltonians with constant locality and mixed states}
\label{sec:hamiltonian_kitaev_mixed}

In this section, we will modify the construction in \cref{first construction} with a unary clock to get constant locality. However, because the clock register will now have $\text{poly}(n)$ qubits, we can no longer simply remove the clock qubits as we did in \cref{thm:formal-lgses-binary}.
Therefore, we will consider the entanglement structure of the reduced density matrices of the ground state with the clock register traced out. Using mixed state entanglement measures (\cref{sec:mixedstate_measures}), we will show that one such density matrix will have high entanglement, and the other will have low entanglement.

As discussed in \cref{entanglement measures}, there are many mixed state measures of entanglement.
We will show that for any ``natural'' mixed state entanglement measure (in the sense of \cref{lem:natural_ent_bounds}), the ground state of the our Hamiltonian (with the clock register traced out) has either high or low entanglement.
We achieve this by giving an upper bound on the entanglement of formation of our low entanglement construction and a lower bound on the distillable entanglement of our high entanglement construction. 
Combined with \cref{lem:natural_ent_bounds}, this gives an entanglement gap for any natural entanglement measure. In fact, Hamiltonians constructed from our ensembles of pseudoentangled states achieve a maximally large gap.

\begin{theorem}
\label{thm:formal-lgses-unary}
For every $n \in \N$, there exist two families $\mathcal{H}^\lo_n$ and $\mathcal{H}^\hi_n$ of $O(1)$-local Hamiltonians on $(n+\text{poly}(n))$ qubits arranged on a 1D line with spectral gap $\Omega(1/\poly(n))$ and efficient procedures that sample (classical descriptions of) Hamiltonians from either family (denoted $H \leftarrow \mathcal{H}^\lo_n$ and $H \leftarrow \mathcal{H}^\hi_n$) with the following properties:
\begin{enumerate}
\item Hamiltonians sampled according to $H \leftarrow \mathcal{H}^\lo_n$ and $H \leftarrow \mathcal{H}^\hi_n$ are computationally indistinguishable under \cref{lwe_assumption_subexp}. \label{item:hamun_indist}
\item If we trace out $\text{poly}(n)$ many qubits from the ground state of each Hamiltonian, the entanglement gap between the resultant quantum states in the high and low families is $\Omega(\text{min}(r, n^{}-r))$ versus $\mathcal{O}(\text{poly}\log n)$, for a cut of size $(r, n-r)$, for any natural entanglement measure.
With overwhelming probability, the reduced states on the first $n$ qubits of the ground states of Hamiltonians $H \leftarrow \mathcal{H}^\lo_n$ have 1D near area-law entanglement and Hamiltonians $H \leftarrow \mathcal{H}^\lo_n$ have 1D volume-law entanglement with respect to any natural mixed state entanglement measure (in the sense of \cref{lem:natural_ent_bounds}).
\label{item:hamun_ent}
\end{enumerate}
\end{theorem}

\begin{proof}
We consider the same construction of $\cH_n^\lo$ and $\cH_n^\hi$ as in \cref{thm:formal-lgses-binary}, except that we now use the unary clock construction from \cref{lem:kitaev_ham_exists} instead of the binary one as in \cref{thm:formal-lgses-binary}.
With this, we only need to show \cref{item:hamun_ent} as the other properties follow in exactly the same way as in the proof of \cref{thm:formal-lgses-binary}.

First consider a Hamiltonian $H \leftarrow \cH^\lo_n$.
Let $|\psi_{\mathrm{ground}}\rangle$ and $\ket{\psi_f}$ be as in \cref{thm:formal-lgses-binary}, except with a unary clock.
Let $\rho_{\mathrm{data}}$ be the density matrix on $n$ qubits left when we trace out the clock qubits from $|\psi_{\mathrm{ground}}\rangle$ and recall that the reduced state of $\ket{\psi_f}$ on the first $n$ qubits is the pseudoentangled state $\proj{\psi}$ (using the notation from \cref{thm:formal-lgses-binary}).
Then, using \cref{clock: trace distance closeness} and the fact that tracing out qubits cannot increase the trace distance between two states,
\begin{equation}
\label{eq_imp1}
||~\rho_{\mathrm{data}} - \proj{\psi} ~||_1 = O\left(\frac{1}{\poly(n)}\right)
\end{equation}
for a sufficiently large padding $K = \poly(n)$.
Then, using continuity properties of the conditional von Neumann entropy (\cref{eq_imp1}, \cref{continuity_vonneumannentropy}), standard properties of the binary entropy function, the fact that all natural mixed state entanglement measures collapse to the von Neumann entanglement entropy for pure states (\cref{entanglement for pure states}), and the fact that coherent information is just the negative of conditional von Neumann entropy (\cref{coherent}) and lower bounds distilalble entanglement (\cref{lem:natural_ent_bounds}),
\begin{equation*}
E_D(\rho_{\mathrm{data}, ~AB}) \geq I(\rho_{\mathrm{data}, ~AB}) = \Omega\left(n^{} - \frac{1}{\poly(n)}\right).
\end{equation*}
By applying a similar argument, when we consider $H \leftarrow \cH_n^\hi$ and a bipartition $(r, n-r)$, using continuity properties of the entanglement of formation (\cref{continuity_formation}) we get that
\begin{equation*}
E_F(\rho_{\mathrm{data}, ~AB}) = O\left(\text{poly}\log n + \frac{1}{\poly(n)} \right).
\end{equation*}
Then, with the observation that any natural entanglement measure is upper bounded by $E_F$ and lower bounded by $E_D$ (\cref{lem:natural_ent_bounds}),  the result follows.
\end{proof}

\begin{remark}
Just as in \cref{first_remark}, under the standard LWE assumption (\cref{lwe_assumption}) \cref{thm:formal-lgses-unary} still holds, but with the smaller entanglement gap $O(n^\delta)$ vs $\Omega(n)$ for any $\delta > 0$.
\end{remark}

\subsection{2D Hamiltonians with geometric locality and pure states}
\label{sec:2dlearning}

\noindent In this section, we will show how to obtain two families of $2$D Hamiltonians on a $2$D grid of $\poly (n)$ qudits, one whose ground state has entanglement entropy of order $n$ and the other whose ground state has entanglement entropy of order $\poly\log n$, with respect to most horizontal cuts across the $2$D grid. Thus, arguably, this gives us a relatively more complicated entanglement structure than the constructions in \cref{thm:formal-lgses-binary} and \cref{thm:formal-lgses-unary}. However, we gain in geometric locality: the Hamiltonian only has nearest neighbor interactions on a $2$D grid. Formally, 2D Hamiltonians are defined as follows.

\begin{definition}[$2$D (local) Hamiltonian,~\cite{aharonov2009power}]
Let $H$ be a Hermitian operator (interpreted as a Hamiltonian,
giving the energy of some system). We say that $H$ is an $r$-state
Hamiltonian if it acts on $r$-state qudits (i.e.~$d=r$). When $r=2$, namely, when
the qudits are qubits. We say that $H$ is $k$-local
if it can be written as $$H = \sum_i H_i,$$ where each $H_i$ acts
non-trivially on at most $k$ qudits. Note that this term does
not assume anything about the physical location of the qudits.
We say that $H$
is a 2D Hamiltonian
if the qudits are arranged on a 2D grid and the terms
$H_i$ interact only pairs of nearest neighbor qudits.
In particular, a 2D
Hamiltonian is $2$-local.
\end{definition}

\paragraph{Overview.}
The main steps of our construction are as follows:
\begin{itemize}
\item First, we start with two $n$-qubit public key pseudoentangled states across multiple cuts, according to the construction in \cref{sec:multicut}, and consider the circuits for preparing them.
Suppose these circuits have $K=\poly(n)$ gates. Without loss of generality, we assume the circuit can be decomposed into $R = \poly(n)$ ``rounds'', each made up of exactly $n$ nearest-neighbor interactions on qubits 1, (1,2), (2,3), etc.
Any circuit can be transformed into this form by inserting a polynomial number of identity and swap gates.
Hence, the circuit contains $n R$ gates in total.
As in \cref{first construction}, we pad the circuits with $n M$ identity gates at the end for a sufficiently large $M = \poly(n)$.
\item Then, we use a modified version of the 2D clock construction~\cite{aharonov2008adiabatic} to construct two families of 2D Hamiltonians such that the padded circuit is embedded into its ground state. That is, if $\mathsf{C} = {U}_{nT} \cdot {U}_{nT-1} \cdots {U}_1$ is the circuit with padding, where $T = M + R$ is the total number of rounds after padding, we construct a Hamiltonian $H$ such that the ground state $\psiground$, on an $n\times T$ grid of qudits, encodes the time evolution of the padded circuit.
\item We show how, because of the padding, the entanglement structure of the 2D ground state across any horizontal cut resembles the entanglement structure of the state
\begin{align}
\ket{\psi_{\mathrm{output}}}=U_{nR}\cdot U_{nR-1}\cdots U_1\ket{0^n},\label{output-state}
\end{align}
across the same cut.

\item By virtue of our pseudoentanglement construction, the state in \cref{output-state} either has high or low entanglement, whenever the cut has distance $\omega(\log n)$ to the boundary of the grid. Then, by a continuity argument, we show that the ground state $\psiground$ also inherits the high or low entanglement property.
\end{itemize}

\subsubsection{Geometrically local 2D history state Hamiltonians\label{sec:2D-Hamiltonian}}

Based on the framework of~\cite{aharonov2008adiabatic}, we begin by showing how to construct a 2D Hamiltonian whose ground state encodes the state obtained by taking our padded circuit, as defined in~\cref{output-state}. Without loss of generality, we assume that before padding, the circuit has initial state $\ket{0}$ and consists of $R$ rounds, where in each rounds there are $n$ nearest-neighbor gates, acting on qubits $1$, $(1,2)$, $(2,3)$, $(3,4)$ and so on.

\paragraph{Legal shapes of the clock.}
To start, we will introduce some notation to reason about the valid clock states.
Suppose we have in total $n(T+1)$ 9-state qudits arranged on a two-dimensional square lattice with $n$ rows and $T+1$ columns. We refer to the quantum state defined on this lattice as a \textit{clock state}. From the top to the bottom, we number the rows from $1$ to $n$. From left to right, we number the columns from $0$ to $T$. Intuitively, the vertical axis of the 2D grid aligns with the spatial dimension of the original circuit, while the horizontal axis of the 2D grid aligns with the temporal dimension of the original circuit. Additionally, we have the following features.
\begin{itemize}
\item \textbf{Possible states}: For each qudit, following the notation of~\cite{aharonov2008adiabatic}, we denote its 9 possible states to be $$\ket{\statea},
\ket{\stateba},\ket{\statebb},
\ket{\stateca},\ket{\statecb}, \ket{\stated}, \ket{\stateea}, \ket{\stateeb}, \ket{\statem}.$$ The former 6 states were introduced in~\cite{aharonov2008adiabatic} whereas the latter three states are newly introduced in this work. 
\item \textbf{Phases}: These states can be categorized into six distinct {\em phases}: the \emph{unborn phase} represented by $\ket{\statea}$, the {\em first phase} containing $\ket{\stateba},\ket{\statebb}$, the {\em second phase} containing $\ket{\stateca},\ket{\statecb}$, the {\em dead phase} represented by $\ket{\stated}$, the {\em flag phase} containing $\ket{\stateea}$, $\ket{\stateeb}$, and the {\em marker phase} represented by $\ket{\statem}$. 
\item \textbf{Encoding the value}: The orientation of the state in the first phase, the second phase, and the flag phase encodes the value, i.e., $\ket{0}$ and $\ket{1}$ states of a qubit in the circuit. For example, we can map $\ket{0}$ to $\ket{\stateba}$ ($\ket{\stateca}$, $\ket{\stateea}$) and $\ket{1}$ to $\ket{\statebb}$ ($\ket{\statecb}$, $\ket{\stateeb}$). 
\item \textbf{Symbols for spanned subspaces}: Denote $\ket{\stateb}$ as any state in the subspace spanned by $\ket{\stateba}$ and $\ket{\statebb}$. Similarly, denote $\ket{\statec}$ as any state in the space spanned by $\ket{\stateca}$ and $\ket{\statecb}$, and denote $\ket{\statee}$ as any state in the space spanned by $\ket{\stateea}$ and $\ket{\stateeb}$. 
\item \textbf{Shapes}: Based on the concept of phases, we can further define the \textit{shape} of the quantum state on this two-dimensional square lattice. A shape is defined as an allocation of one of five phases to each qudit, ignoring the orientations of the states in the first, second, or the flag phase. Intuitively, the flag phase serves as a marker for the location where the current change occurs, and plays a similar role as the flag in~\cite{aharonov2009power} concerning the ground state in the one-dimensional setting. 
\end{itemize}
Similar to~\cite{aharonov2008adiabatic}, we show that we can represent the time evolution of the padded circuit using a series of $(3n-1)T+1$ clock states that are \textit{legal shapes}, where each shape corresponds to one time step in the circuit. 
\begin{remark}
Formally defined in~\cref{def:legal-shapes}, the notion of legal shapes we consider is slightly different from the original ones in~\cite{aharonov2008adiabatic} because we have the extra flag and marker states.
\end{remark}

\begin{definition}[Legal shapes,~\cite{aharonov2008adiabatic}]\label{def:legal-shapes}
Denote $L=(3n-1)T$. For any $t\in[0,L]$, we say a clock state is in the $t$-th legal shape if one of the following conditions is satisfied
\begin{enumerate}
\item There exists $r$ and $k\in[0,n]$ such that $t=(3n-1)r+k$. Moreover, the $r$ leftmost columns of the state are in the dead phase, the top $k-1$ qudits in the $(r+1)$-th column are in their second phase, the bottom $n-k$ are in the first phase, the qudit in between is in the flag phase, and qudits in the remaining $R-r$ columns are all in the unborn phase
\item There exists $r$ and $k\in[0,2n-1]$ such that $t=(3n-1)r+n+k$. Moreover, the $r$ leftmost columns of the state are in the dead phase, the $(r+1)$-th column has its $n-k$ topmost qudits in the second phase while its remaining $k$ qudits in the dead phase, the $(r+2)$-th column has its $n-k$ topmost qudits in the unborn phase, the bottom $k-1$ qudits in the first phase, the qudit in between is in the flag phase, and all qudits in the remaining columns are in the unborn phase.
\item There exists $r$ and $k\in[0,2n-1]$ such that $t=(3n-1)r+n+k$. Moreover, the $r$ leftmost columns of the state are in the dead phase, the $(r+1)$-th column has its $n-k$ topmost qudits in the second phase while its remaining $k$ qudits in the dead phase, the $(r+2)$-th column has its $n-k-1$ topmost qudits in the unborn phase, the bottom $k$ qudits in the first phase, the qudit in between is in the marker phase, and all qudits in the remaining columns are in the unborn phase.
\end{enumerate}
\end{definition}

Denote $\mathcal{S}$ to be the $(L+1)2^n$-dimensional space spanned by all states in legal shapes. Moreover, for each $0\leq t\leq L$ and $0\leq j\leq 2^n-1$, define $\ket{\gamma_{t}^j}$ to be the state in the $t$-th legal shape, where the orientations of its $n$ active qudits (i.e., qudits in the first phase, second phase, or the flag phase) correspond to the state of the circuit after applying the first $t$ gates to an initial state that encodes the binary representation of $j$, or quantitatively,
\begin{align}
U_{t}\cdot U_{t-1}\cdots U_1\ket{j}.\label{eqn:legal-shape-state}
\end{align}
Moreover, we denote $\ket{\gamma(t)}\equiv\ket{\gamma_t^0}$ for brevity.

The following is immediate from the construction.
\begin{lemma}
For any $0\leq j\leq 2^n-1$ and any $t_1\neq t_2\in[0,L]$, $\ket{\gamma_{t_1}^j}$ and $\ket{\gamma_{t_2}^j}$ are orthogonal to each other.
\end{lemma}

\paragraph{Transition rules.}
Based on the definition of $\ket{\gamma_t^j}$, we can further divide $\mathcal{S}$ into $2^n$ subspaces $\mathcal{S}_1,\ldots,\mathcal{S}_{2^n}$, where each subspace $\mathcal{S}_j$ is spanned by $\{\ket{\gamma_0^j},\ldots,\ket{\gamma_L^j}\}$. As shown in~\cite{aharonov2008adiabatic}, intuitively, the transition rules between the states $\ket{\gamma_0^j},\ldots,\ket{\gamma_L^j}$ can also be understood as follows. 
\begin{itemize}
\item \textbf{Downward stage}: Consider a state $\ket{\gamma^j_t}$ for some $t=(3n-1)r$. In this stage, the $n$ qudits located in the $r$-th column are in their initial phase, and the orientations of these active qudits (i.e., those in the first phase, second phase, and the flag phase) encode the initial state of the circuit at the beginning of the $r$-th round. To the left of this column, the qudits are in the dead phase, while to the right, they are in the unborn phase. The state $\ket{\gamma^j_{t+1}}$ can be obtained from $\ket{\gamma^j_{t}}$ by converting the uppermost qudit in the $r$-th column into active-phase and then applying the first gate of the $r$-th round (a one-qubit gate) to the state encoded in these active qudits. Subsequently, the state $\ket{\gamma^j_{t+2}}$ is derived from $\ket{\gamma^j_{t+1}}$ by converting the second qudit from the top in the $r$-th column into active phase, converting the frist qudit from the top in the $r$-th column into first phase, and applying the second gate of the $r$-th round (a two-qubit gate) to both this qudit and the one above it. This process continues in a similar manner until we reach the state $\ket{\gamma^j_{t+n}}$, where the entire $r$-th column consists of qudits in the second phase except the one at the bottom. These sequential steps are referred to as the downward stage. 

\item \textbf{Upward stage}: Correspondingly, there is an upward stage. In this stage, the state $\ket{\gamma^j_{t+n+1}}$ is transformed from $\ket{\gamma^j_{t+n}}$ by  shifting the bottommost qudit in the $r$-th column one position to the right. To be more precise, the bottommost qudit shifts into the dead phase, while the qudit to the right of it shifts into the active phase. Then, $\ket{\gamma^j_{t+n+2}}$ is transformed from $\ket{\gamma^j_{t+n+1}}$ by shifting the bottommost qudit into the first phase without changing its encoding and shifting the qudit above it into the marker phase. It is worth noticing that the encoded state is not changed during this process, which corresponds to the identity gate applied in the circuit. Continuing this iterative process to go up the column, we can ``copy'' the encoding of the $r$-th column into the $(r+1)$-th column.
We observe that the upward stage ends in the state $\ket{\gamma^j_{t+n+n}}=\ket{\gamma^j_{2(r+1)n}}$ which aligns with the previously described initial state of a new round.
\end{itemize}

An example of legal shapes with $n=6$ and $T=2$ is given in \cref{fig:legal-shapes}, where the shapes correspond to $t=0,1,\ldots,19$. It is worth noting that each shape has exactly $n$ qudits in the first, second, or the flag phase. Consequently, each shape contains a subspace with a dimension of $2^n$ spanned by the possible states of qudits within the first or second phase.

\begin{figure}[h!]
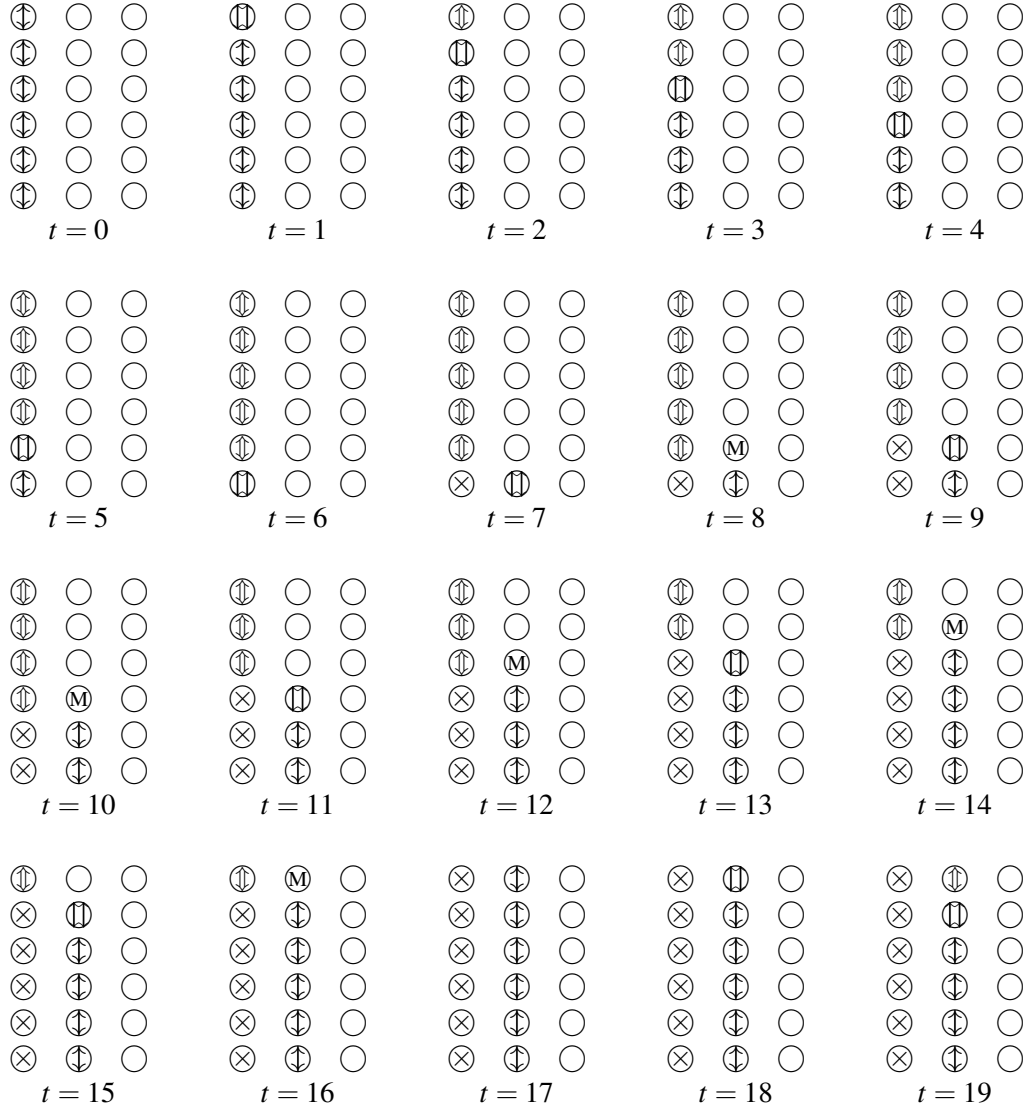

\center{{
\begin{equation*}
\begin{array}{ccccccccc}
\begin{array}{ccc}
\stateb & \statea & \statea \\
\stateb & \statea & \statea \\
\stateb & \statea & \statea \\
\stateb & \statea & \statea \\
\stateb & \statea & \statea \\
\stateb & \statea & \statea \\
\multicolumn{3}{c}{t=0}
\end{array}
& &
\begin{array}{ccc}
\statee & \statea & \statea \\
\stateb & \statea & \statea \\
\stateb & \statea & \statea \\
\stateb & \statea & \statea \\
\stateb & \statea & \statea \\
\stateb & \statea & \statea \\
\multicolumn{3}{c}{t=1}
\end{array}
& & 
\begin{array}{ccc}
\statec & \statea & \statea \\
\statee & \statea & \statea \\
\stateb & \statea & \statea \\
\stateb & \statea & \statea \\
\stateb & \statea & \statea \\
\stateb & \statea & \statea \\
\multicolumn{3}{c}{t=2}
\end{array} & &
\begin{array}{ccc}
\statec & \statea & \statea \\
\statec & \statea & \statea \\
\statee & \statea & \statea \\
\stateb & \statea & \statea \\
\stateb & \statea & \statea \\
\stateb & \statea & \statea \\
\multicolumn{3}{c}{t=3}
\end{array}
& &
\begin{array}{ccc}
\statec & \statea & \statea \\
\statec & \statea & \statea \\
\statec & \statea & \statea \\
\statee & \statea & \statea \\
\stateb & \statea & \statea \\
\stateb & \statea & \statea \\
\multicolumn{3}{c}{t=4}
\end{array}\\\nonumber
&&&\\\nonumber

\begin{array}{ccc}
\statec & \statea & \statea \\
\statec & \statea & \statea \\
\statec & \statea & \statea \\
\statec & \statea & \statea \\
\statee & \statea & \statea \\
\stateb & \statea & \statea \\
\multicolumn{3}{c}{t=5}
\end{array}& & 

\begin{array}{ccc}
\statec & \statea & \statea \\
\statec & \statea & \statea \\
\statec & \statea & \statea \\
\statec & \statea & \statea \\
\statec & \statea & \statea \\
\statee & \statea & \statea \\
\multicolumn{3}{c}{t=6}
\end{array} & &

\begin{array}{ccc}
\statec & \statea & \statea \\
\statec & \statea & \statea \\
\statec & \statea & \statea \\
\statec & \statea & \statea \\
\statec & \statea & \statea \\
\stated & \statee & \statea \\
\multicolumn{3}{c}{t=7}
\end{array}& &
\begin{array}{ccc}
\statec & \statea & \statea \\
\statec & \statea & \statea \\
\statec & \statea & \statea \\
\statec & \statea & \statea \\
\statec & \statem & \statea \\
\stated & \stateb & \statea \\
\multicolumn{3}{c}{t=8}
\end{array}
& &
\begin{array}{ccc}
\statec & \statea & \statea \\
\statec & \statea & \statea \\
\statec & \statea & \statea \\
\statec & \statea & \statea \\
\stated & \statee & \statea \\
\stated & \stateb & \statea \\
\multicolumn{3}{c}{t=9}
\end{array}\\\nonumber
&&&\\\nonumber

\begin{array}{ccc}
\statec & \statea & \statea \\
\statec & \statea & \statea \\
\statec & \statea & \statea \\
\statec & \statem & \statea \\
\stated & \stateb & \statea \\
\stated & \stateb & \statea \\
\multicolumn{3}{c}{t=10}
\end{array}
 & &
 \begin{array}{ccc}
\statec & \statea & \statea \\
\statec & \statea & \statea \\
\statec & \statea & \statea \\
\stated & \statee & \statea \\
\stated & \stateb & \statea \\
\stated & \stateb & \statea \\
\multicolumn{3}{c}{t=11}
\end{array}
& &
\begin{array}{ccc}
\statec & \statea & \statea \\
\statec & \statea & \statea \\
\statec & \statem & \statea \\
\stated & \stateb & \statea \\
\stated & \stateb & \statea \\
\stated & \stateb & \statea \\
\multicolumn{3}{c}{t=12}
\end{array}& &

\begin{array}{ccc}
\statec & \statea & \statea \\
\statec & \statea & \statea \\
\stated & \statee & \statea \\
\stated & \stateb & \statea \\
\stated & \stateb & \statea \\
\stated & \stateb & \statea \\
\multicolumn{3}{c}{t=13}
\end{array}& & 
\begin{array}{ccc}
\statec & \statea & \statea \\
\statec & \statem & \statea \\
\stated & \stateb & \statea \\
\stated & \stateb & \statea \\
\stated & \stateb & \statea \\
\stated & \stateb & \statea \\
\multicolumn{3}{c}{t=14}
\end{array}
\\\nonumber
&&&\\\nonumber

\begin{array}{ccc}
\statec & \statea & \statea \\
\stated & \statee & \statea \\
\stated & \stateb & \statea \\
\stated & \stateb & \statea \\
\stated & \stateb & \statea \\
\stated & \stateb & \statea \\
\multicolumn{3}{c}{t=15}
\end{array} & &
\begin{array}{ccc}
\statec & \statem & \statea \\
\stated & \stateb & \statea \\
\stated & \stateb & \statea \\
\stated & \stateb & \statea \\
\stated & \stateb & \statea \\
\stated & \stateb & \statea \\
\multicolumn{3}{c}{t=16}
\end{array}
& &
\begin{array}{ccc}
\stated & \stateb & \statea\\
\stated & \stateb & \statea\\
\stated & \stateb & \statea\\
\stated & \stateb & \statea\\
\stated & \stateb & \statea\\
\stated & \stateb & \statea \\
\multicolumn{3}{c}{t=17}
\end{array}& & 
\begin{array}{ccc}
\stated & \statee & \statea\\
\stated & \stateb & \statea\\
\stated & \stateb & \statea\\
\stated & \stateb & \statea\\
\stated & \stateb & \statea\\
\stated & \stateb & \statea \\
\multicolumn{3}{c}{t=18}
\end{array}& & 
\begin{array}{ccc}
\stated & \statec & \statea \\
\stated & \statee & \statea \\
\stated & \stateb & \statea \\
\stated & \stateb & \statea \\
\stated & \stateb & \statea \\
\stated & \stateb & \statea \\
\multicolumn{3}{c}{t=19}
\end{array}
\\\nonumber
&&&\\\nonumber
\end{array}\nonumber
\end{equation*}
\vspace{-13mm}
}}
\caption{An example of clock states in legal shapes with $n=6$, $T=2$.}
\label{fig:legal-shapes}
\end{figure}

\paragraph{Constructing the corresponding Hamiltonian.}
\label{embedding}
Based on the construction in~\cite{aharonov2008adiabatic}, in the remaining part of this section, we present a two dimensional Hamiltonian on the grids whose ground state is the history state of the evolution
\begin{align}
\ket{\psi_{\mathrm{ground}}}=\frac{1}{\sqrt{L+1}}\sum_{t=0}^L\ket{\gamma_t^0}.\label{eqn:2d-groundstate}
\end{align}
Following~\cite{aharonov2008adiabatic}, our 2D Hamiltonian $H_{2D}$ contains three parts:
\begin{align}
H_{2D}\equiv\frac{1}{2}\sum_{\ell=1}^L{H_\ell}
+ H_{\mathrm{input}} +
J \cdot H_{\mathrm{clock}},\label{eqn:2D-Hamiltonian-sum}
\end{align}
where $J=\eps^{-2}\cdot L^6$ with $\epsilon=\frac{1}{\poly n}$. Conceptually, each Hamiltonian term \cref{eqn:2D-Hamiltonian-sum} imparts an ``energy penalty'' upon any state that fails to meet the required properties of the ground state. In particular, $H_{\mathrm{clock}}$ guarantees that the ground state is a superposition over a set of clock states in legal shapes. Then, $H_{\mathrm{input}}$ verifies that the component in the first legal shapes encodes the input of the circuit, i.e., $\ket{0}$, in its first column of qudits. Finally, the propogation $H_{\ell}$ ensures the correctness of transitions between consecutive legal shapes, meaning that the encoded quantum state evolves in accordance with the corresponding gate applied in the circuit during that time step. Formally, as in~\cite{aharonov2008adiabatic}, we let
\begin{align*}
H_{\mathrm{input}} 
\deq \sum_{i=1}^{n}{(\ket{\statebb}\bra{\statebb})_{i,1}},
\end{align*}
where the indices indicate that the corresponding term acts on the $i$-th row and the first column of the 2D grid. As for $H_{\mathrm{clock}}$, similar to the case in~\cite{aharonov2008adiabatic}, it is worth noting that the legal shapes defined in \cref{def:legal-shapes} can be verified in a {\em $2$-local way}, which allows us to define a $2$-local Hamiltonian such that its ground space contains only superpositions of clock states in legal shapes.

\begin{table}[h!]
\center{
\begin{tabular}{|l|l|}
\hline
Forbidden  & Guarantees that \\ \hline \hline
   $\statea\stateb, \statea\statec, \statea\stated, \statea\statee, \statea\statem$& $\statea$ is to the right of all other qubits\\ \hline
   $\statea\stated, \stateb\stated, \statec\stated, \statee\stated,\statem\stated$& $\stated$ is to the left of all other qubits\\ \hline
   $\statea\stated, \stated\statea$& $\statea$ and $\stated$ are not horizontally adjacent\\ \hline
   $\stateb\stateb$, $\stateb\statec$, $\stateb\statee$&\\
 $\statec\stateb$, $\statec\statec$, $\statec\statee$ & only one of $\stateb$, $\statec$, or $\statee$ per row \\
    $\statee\stateb$, $\statee\statec$, $\statee\statee$ &  \\  \hline
   $\ontop{\statea}{\statec}, \ontop{\stateb}{\statec}, \ontop{\stated}{\statec}, \ontop{\statee}{\statec}$& only $\statec$ above $\statec$\\ \hline
   $\ontop{\stateb}{\statea}, \ontop{\stateb}{\statec}, \ontop{\stateb}\stated, \ontop{\stateb}{\statee}$& only $\stateb$ below $\stateb$\\ \hline
   $\ontop{\statea}{\stated}, \ontop{\stated}{\statea}$& $\statea$ and $\stated$ are not vertically adjacent\\ \hline
   $\ontop{\statec}{\statea}, \ontop{\statee}{\statea}, \ontop{\stated}{\stateb}, \ontop{\stated}{\statee}$& no $\statea$ below $\statec$ and $\statee$; no $\stateb$ and $\statee$ below $\stated$\\ \hline
   $\ontop{\statea}{\stateb}$& no $\statea$ above $\stateb$\\ \hline
   $\ontop{\stateb}{\statem},\ontop{\statec}{\statem},\ontop{\stated}{\statem},\ontop{\statee}{\statem}$& only $\statea$ above $\statem$\\ \hline
    $\ontop{\statem}{\statea},\ontop{\statem}{\statec},\ontop{\statem}{\stated},\ontop{\statem}{\statee}$& only $\stateb$ below $\statem$\\ \hline
    $\stateb\statem,\stated\statem,\statee\statem,\statem\statem$& only $\statec$ to the left of $\statem$\\ \hline
    $\statem\stateb,\statem\statec,\statem\statee$& only $\statea$ to the right of $\statem$\\ \hline
\end{tabular}
}
\caption{Local rules for basis state to be in legal shape.}
\label{tab:rules}
\end{table}

Similarly to~\cite[Claim 4.2]{aharonov2008adiabatic}, we prove the following result.
\begin{lemma}\label{lem:rules}
A clock state is in legal shape defined in \cref{def:legal-shapes} if and only if it contains none of the forbidden configurations in \cref{tab:rules}.
\end{lemma}
\begin{proof}
The proof is similar to the proof of~\cite[Claim 4.2]{aharonov2008adiabatic}. First, we can check that any clock state in legal shape contains none of the forbidden configurations. Conversely, for any clock state that contains none of these configurations, each of its rows admits to the pattern $\stated^*[\stateb,\statec,\statee,\statec\statem]\statea^*$. In other words, it starts with a sequence of zero or more $\stated$ from the left, then incorporates either $\stateb$, $\statec$, $\statee$, or $\statec\statem$, and ends with a sequence of zero or more $\statea$ on the right. As for the vertical direction, the columns are in one of three distinct formats when read from top to bottom: $\statec^*[\statee]\stateb^*$, $\statec^*\stated^*$, $\statea^*[\statee]\stateb^*$, or $\statea^*[\statem]\stateb^*$, and such a shape must satisfy~\cref{def:legal-shapes}.
\end{proof}
Based on \cref{lem:rules}, we can define a two-local
nearest-neighbor Hamiltonian that checks whether a clock state is in legal shape or not. For example, prohibits a qudit at location $(i,j)$ in state $\statea$ to be positioned to the left of a qudit at location $(i, j+1)$ in state $\stated$. Then, the corresponding Hamiltonian term would be $(\ket{\statea,\stated}\bra{\statea,\stated})_{(i,j),(i,j+1)}$.
Summing over all the forbidden configurations of \cref{tab:rules} and over all relevant pairs of particles, we have
\begin{align}
H_{\mathrm{clock}}\equiv \sum_{r\in \mathrm{rules}} H_{r}.\label{eqn:H_clock}
\end{align}
Next, we turn to the propagation term $H_{\ell}$ that guarantees the correctness of transitions between clock states in consecutive legal shapes. Observe that according to \cref{def:legal-shapes}, consecutive legal shapes only differ in at most two adjacent locations, which makes it possible to be verified using two-local nearest-neighbour Hamiltonians. In particular, the structure of $H_{\ell}$ takes on different forms depending on whether $\ell$ is in the downward phase (i.e., is of the form $(3n-1)r+k$ for $1\le k\le n$) or the upward phase (i.e., is of the form $(3n-1)r+n+k$ for $1\le k< 2n-1$). For the downward stage, $H_\ell$ checks that a gate is
applied correctly. For $\ell=(3n-1)r+k$ and $2<k\leq n$, as the same in~\cite{aharonov2008adiabatic}, we define 
\begin{align*}
H_\ell\equiv \left(\begin{array}{cc} 0& -U_\ell\\ -U_\ell^{\dagger}&0
\end{array}\right) &+
\left(
\stackketbra{\stateea}{\stateba} +
\stackketbra{\stateea}{\statebb} +
\stackketbra{\stateeb}{\stateba} +
\stackketbra{\stateeb}{\statebb}
\right)\ontop{\rm_{k-1,r}}{\rm_{k,r}} \\
&+ \left(
\stackketbra{\stateea}{\stateba} +
\stackketbra{\stateea}{\statebb} +
\stackketbra{\stateeb}{\stateba} +
\stackketbra{\stateeb}{\statebb}
\right)\ontop{\rm_{k,r}}{\rm_{k+1,r}},\numberthis\label{eqn:H_ell}
\end{align*}
where the stacked notation 
$\stackketbra{\cdot}{
\cdot}$
denotes a 2-local Hamiltonian term acting on nearest neighbors as indicated by the indices outside the parentheses. Moreover, the first term in \cref{eqn:H_ell} represents a Hamiltonian that acts on the two particles in positions $(k,r)$ and $(k+1,r)$ and is restricted to the subspace spanned by
\begin{align*}
\stackket{\stateea}{\stateba} ~~ \stackket{\stateea}{\statebb} ~~ \stackket{\stateeb}{\stateba} ~~ \stackket{\stateeb}{\statebb}~~~~
   \stackket{\stateca}{\stateea} ~~ \stackket{\stateca}{\stateeb} ~~ \stackket{\statecb}{\stateea} ~~ \stackket{\statecb}{\stateeb}
\end{align*}
For the special cases where $k=1$ or $n$, we set
\begin{align*}
H_{(3n-1)r+1}\equiv \left(\begin{array}{cc} 0& -U_{(3n-1)r+1}\\ -U_{(3n-1)r+1}^{\dagger}&0
\end{array}\right) &+
\left(
\ketbra{\stateba}{\stateba} +
\ketbra{\statebb}{\statebb}
\right)_{1,r} \\
&+ \left(
\stackketbra{\stateea}{\stateba} +
\stackketbra{\stateea}{\statebb} +
\stackketbra{\stateeb}{\stateba} +
\stackketbra{\stateeb}{\statebb}
\right)\ontop{\rm_{1,r}}{\rm_{2,r}},
\end{align*}
and
\begin{align*}
H_{(3n-1)r+n}\equiv \left(\begin{array}{cc} 0& -U_{(3n-1)r+n}\\ -U_{(3n-1)r+n}^{\dagger}&0
\end{array}\right) &+
\left(
\stackketbra{\stateea}{\stateba} +
\stackketbra{\stateea}{\statebb} +
\stackketbra{\stateeb}{\stateba} +
\stackketbra{\stateeb}{\statebb}
\right)\ontop{\rm_{n-1,r}}{\rm_{n,r}} \\
&+ \left(
\stackketbra{\stateca}{\stateea} +
\stackketbra{\stateca}{\stateeb} +
\stackketbra{\statecb}{\stateea} +
\stackketbra{\statecb}{\stateeb}
\right)\ontop{\rm_{n-1,r}}{\rm_{n,r}}.
\end{align*}
As for the upward stage, we have $\ell=(3n-1)r+n+k$ for some $1\leq k<2n-1$. If $1<k<2n-1$ and is an odd number, in this step the qudit in the marker phase $\ket{\statem}$ in the $i$-th row with $i=n-(k-1)/2$ will acquire the value of the qudit on the left of it. Then, the corresponding Hamiltonian term $H_\ell$ can be written as
\begin{align*}
H_\ell&\equiv \stackketbra{\statem}{\stateba}  \ontop{\rm_{i,{r+1}}}{\rm_{i+1,{r+1}}}+  \stackketbra{\statea}{\stateea}  \ontop{\rm_{i-1,r+1}}{\rm_{i,r+1}}-  \left(\ketbra{\stateca,\statem}{\stated,\stateea} +  \ketbra{\stated,\stateea}{\stateca,\statem}\right)_{\rm(i,r)(i,r+1)}\\
&\,+\stackketbra{\statem}{\statebb} \ontop{\rm_{i,r+1}}{\rm _{i+1,r+1}} +
 \stackketbra{\statea}{\stateeb}\ontop{\rm_{i-1,r+1}}{\rm_{i,r+1}}
- \left(\ketbra{\statecb,\statem}{\stated,\stateeb} + \ketbra{\stated,\stateeb}{\statecb,\statem}\right)_{\rm(i,r)(i,r+1)},
\end{align*}
where the first line corresponds to changing the state $\ket{\stateca,\statem}$ to $\ket{\statea,\stateea}$, and the second line corresponds to changing the state $\ket{\statecb,\statem}$ to $\ket{\statea,\stateeb}$.

If $k$ is an even number, in this step the qudit in the flag phase $\ket{\statee}$ in the $i$-th row with $i=n+1-k/2$ will transform into the first phase $\stateb$, while the qudit above it will transform into the marker phase $\statem$ from the unborn phase $\ket{\statea}$. Then, the corresponding Hamiltonian term $H_\ell$ can be written as
\begin{align*}
H_\ell&\equiv \stackketbra{\statem}{\stateba}  \ontop{\rm_{i-1,{r+1}}}{\rm_{i,{r+1}}}+  \stackketbra{\statea}{\stateea}  \ontop{\rm_{i-1,r+1}}{\rm_{i,r+1}}-\left(
\stackket{\statem}{\stateba}\stackbra{\statea}{\stateea}+\stackket{\statea}{\stateea}\stackbra{\statem}{\stateba}\right)\ontop{\rm_{i-1,r+1}}{\rm_{i,r+1}}\\
&\,+\stackketbra{\statem}{\statebb} \ontop{\rm_{i-1,r+1}}{\rm _{i,r+1}} +
 \stackketbra{\statea}{\stateeb}\ontop{\rm_{i-1,r+1}}{\rm_{i,r+1}}
-\left(\stackket{\statem}{\statebb}\stackbra{\statea}{\stateeb}+\stackket{\statea}{\stateeb}\stackbra{\statem}{\statebb}\right)\ontop{\rm_{i-1,r+1}}{\rm_{i,r+1}}.
\end{align*}
For the special case where $k=1$ or $k=2n-1$, we set
\begin{align*}
H_{(3n-1)r+n+1} &\equiv
 \stackketbra{\stateca}{\stateea}\ontop{\rm_{n,r-1}}{\rm_{n,r}}+\stackketbra{\statecb}{\stateea}\ontop{\rm_{n,r-1}}{\rm_{n,r}}
 +  \stackketbra{\statea}{\stateea}\ontop{\rm_{n-1,r+1}}{\rm_{n,r+1}}
 \\
&\, -\left(\ketbra{\stateea,\statea}{\stated,\stateea} +  \ketbra{\stated,\stateea}{\stateea,\statea}\right)_{\rm(n,r)(n,r+1)}
 \\
&\, + 
 \stackketbra{\stateca}{\stateeb}\ontop{\rm_{n,r-1}}{\rm_{n,r}}+\stackketbra{\statecb}{\stateeb}\ontop{\rm_{n,r-1}}{\rm_{n,r}}
 +
 \stackketbra{\statea}{\stateeb}\ontop{\rm_{n-1,r+1}}{\rm_{n,r+1}}\\
&\,-\left(\ketbra{\stateeb,\statea}{\stated,\stateeb} + \ketbra{\stated,\stateeb}{\stateeb,\statea}\right)_{\rm(n,r)(n,r+1)},
\end{align*}
and
\begin{align*}
H_{(3n-1)(r+1)} &\equiv
 \stackketbra{\stateca}{\stated}\ontop{\rm_{1,r}}{\rm_{2,r}}+  \ketbra{\stateba}{\stateba}_{1,r+1}
 -  \left(\ketbra{\stateca,\statem}{\stated,\stateba} +  \ketbra{\stated,\stateba}{\stateca,\statem}\right)_{\rm(1,r)(1,r+1)}
 \\
&\,+\stackketbra{\statecb}{\stated}\ontop{\rm_{1,r}}{\rm _{2,r}} +
 \ketbra{\statebb}{\statebb}_{1,r+1}
- \left(\ketbra{\statecb,\statem}{\stated,\statebb} + \ketbra{\stated,\statebb}{\statecb,\statem}\right)_{\rm(1,r)(1,r+1)}.
\end{align*}
Based on the definition of $H_{\mathrm{input}}$, $H_{\mathrm{clock}}$, and $H_{\ell}$, we present the following result, which is a variant of \cite[Lemma 4.7]{aharonov2008adiabatic}.

\begin{lemma} \label{lem:groundstates_all_close}
For any $\epsilon \geq 1/\poly(n)$, the Hamiltonian $H_{2D}$ defined in \cref{eqn:2D-Hamiltonian-sum} on the $n\times (T+1)$ grid of qudits satisfies the following.
\begin{itemize}
\item The history state of the evolution $\ket{\psi_{\mathrm{ground}}}$ is a ground state of $H_{2D}$. Moreover, any ground state of $H_{2D}$ is $\epsilon$-close to $\ket{\psi_{\mathrm{ground}}}$. 
\item The spectral gap of $H_{2D}$ is at least $\Omega((nT)^{-3})=1/\poly n$.
\end{itemize}\label{lem:H2D-ground-state}
\end{lemma}
The proof of \cref{lem:H2D-ground-state} essentially follows the proof of \cite[Lemma 4.7]{aharonov2008adiabatic}.

\subsubsection{Entanglement properties of the ground state}
In this section, we analyze the entanglement properties of the state $\ket{\psi_{\mathrm{ground}}}$ defined in~\cref{eqn:2d-groundstate}, which is the ground state of the 2D Hamiltonian $H_{2D}$ defined in \cref{eqn:2D-Hamiltonian-sum}. We begin by presenting some definitions and technical lemmas that are useful.

\begin{definition}[Cutwise orthogonality]
For any $K > 0$, let $\ket{\psi_1}$ and $\ket{\psi_2}$ be two $K$-qudit states, and consider any bipartite cut among these $K$ qudits. We denote the density matrices obtained by tracing out the right parts of $\ket{\psi_1}$ and $\ket{\psi_2}$ as $\rho_1^{(A)}$ and $\rho_2^{(A)}$, respectively. Similarly, we define $\rho_1^{(B)}$ and $\rho_2^{(B)}$. We define $\ket{\psi_1}$ and $\ket{\psi_2}$ to be \textit{cutwise orthogonal} with respect to this bipartite cut if the following condition holds:
\begin{align*}
\rho_1^{(A)}\rho_2^{(A)}=\rho_1^{(B)}\rho_2^{(B)}=0.
\end{align*}\label{def:cutwise-orthogonality}
\end{definition}

\noindent The next lemma demonstrates the entanglement properties of a set of mutually cutwise orthogonal states.

\begin{lemma}
For any $m,K>0$, suppose there are $m$ $K$-qudit states $\ket{\psi_1},\ldots,\ket{\psi_m}$ and there is a bipartite cut among these $K$ qudits into subsystems $(A, B)$.
Then for any $\alpha_1,\ldots,\alpha_m$ with $\sum_i\alpha_i^2=1$, the state
\begin{align*}
\ket{\Psi}=\alpha_1\ket{\psi_1}+\cdots+\alpha_m\ket{\psi_m}
\end{align*}
satisfies 
\begin{align*}
S(\Psi_A)=\sum_{i=1}^m|\alpha_i|^2S((\psi_i)_A)+\sum_{i=1}^m|\alpha_i|^2\log\frac{1}{|\alpha_i|^2}\,,
\end{align*}
if $\ket{\psi_i}$ and $\ket{\psi_j}$ are cutwise orthogonal for any $i\neq j\in[m]$.
Here, $\Psi_A$ is the reduced state of $\ket{\Psi}$ on subsystem $A$, and likewise for the other states.
\label{lem:CO-entanglement}
\end{lemma}
\begin{proof}
For any $i\in[m]$, we express the Schmidt decomposition of $\ket{\psi_i}$ as
\begin{align*}
\ket{\psi_i}=\sum_{k=1}^{K_i}\beta_{i,k}\ket{\psi_{i,k}^{(A)}}\otimes\ket{\psi_{i,k}^{(B)}}.
\end{align*}
Consequently, for any $i\in[m]$ we have
\begin{align*}
(\psi_i)_A=\sum_{k=1}^{K_i}|\beta_{i,k}|^2\ket{\psi_{i,k}^{(A)}}\bra{\psi_{i,k}^{(A)}},\qquad(\psi_i)_B=\sum_{k=1}^{K_i}|\beta_{i,k}|^2\ket{\psi_{i,k}^{(B)}}\bra{\psi_{i,k}^{(B)}}.
\end{align*}
Given that $\ket{\psi_i}$ and $\ket{\psi_j}$ are cutwise orthogonal for any $i\neq j\in[m]$, we can then derive that
\begin{align*}
\left|\left\langle\psi_{i,k_i}^{(A)}|\psi_{j,k_j}^{(A)}\right\rangle\right|=\left|\left\langle\psi_{i,k_i}^{(B)}|\psi_{j,k_j}^{(B)}\right\rangle\right|=0,
\end{align*}
for any $k_i\in[K_i]$, $k_j\in[K_j]$, and any $i\neq j\in[m]$. Then, tracing out the right part of $\ket{\Phi}$, the resulting reduced state $\Psi_A$ is given by
\begin{align*}
\Psi_A=\sum_{i=1}^m|\alpha_i|^2\sum_{k=1}^{K_i}|\beta_{i,k}|^2\ket{\psi_{i,k}^{(A)}}\bra{\psi_{i,k}^{(A)}},
\end{align*}
where the components $\{\ket{\psi_{i,k}^{(A)}}\}$ are orthogonal to each other. Hence, we can conclude that
\begin{align*}
S(\Psi_A)&=\sum_{i=1}^m\sum_{k=1}^{K_i}|\alpha_i|^2|\beta_{i,k}|^2\log\frac{1}{|\alpha_i|^2|\beta_{i,k}|^2}\\
&=\sum_{i=1}^m|\alpha_i|^2\sum_{k=1}^{K_i}|\beta_{i,k}|^2\left(\log|\frac{1}{\alpha_i|^2}+\log\frac{1}{|\beta_{i,k}|^2}\right)\\
&=\sum_{i=1}^m|\alpha_i|^2S(\ket{\psi_i})+\sum_{i=1}^m|\alpha_i|^2\log\frac{1}{|\alpha_i|^2}. \qedhere
\end{align*}
\end{proof}

We consider a horizontal cut through the 2D grid of size $n\times (T+1)$ such that the cut has distance at least $\omega(\log n)$ to both the top and the bottom of the grids, and use $S(\cdot)$ to denote the entanglement entropy across this cut. Additionally, we slightly abuse notation by using $S((\psi_{\mathrm{output}})_A)$ to denote the entanglement entropy of the reduced output state $(\psi_{\mathrm{output}})_A$ across the corresponding cut among the $n$ qubits. Formally, if the horizontal cut through the $n\times (T+1)$ 2D grid is positioned between the $c$-th line and the $(c+1)$-th line, then $S((\psi_{\mathrm{output}})_A)$ is defined as the von Neumann entropy of $(\psi_{\mathrm{output}})_A$, which is the reduced state of $\ket{\psi_{\mathrm{output}}}$  across the cut between the $c$-th qubit and the $(c+1)$-th qubit among the $n$ qubits. With this notation, we establish the following fact.
\begin{lemma}
For any $t\in[(3n-1)R,L]$, the $t$-th legal shape $\ket{\gamma(t)}$ defined in \cref{eqn:legal-shape-state} satisfies
$S(\gamma_A(t))=S((\psi_{\mathrm{output}})_A)$.\label{fact:legal-shape-entanglement}
\end{lemma}
\begin{proof}
For any $t\in[(3n-1)R,L]$, observe that $\ket{\gamma(t)}$ can be decomposed as follows:
\begin{align*}
\ket{\gamma(t)}=(\mathbf{dead\ qudits})\otimes(\mathbf{active\ qudits})\otimes(\mathbf{unborn\ qudits}),
\end{align*}
where $(\mathbf{dead\ qudits})$ denotes qudits in the dead phase, $(\mathbf{active\ qudits})$ denotes qudits in the first, second, or the flag phase, and $(\mathbf{unborn\ qudits})$ denotes qudits in the unborn phase. Moreover, note that with respect to the cut $(\mathbf{dead\ qudits})$ can be further divided into
\begin{align*}
(\mathbf{dead\ qudits})=(\mathbf{dead\ qudits})_{\text{above}}\otimes(\mathbf{dead\ qudits})_{\text{below}},
\end{align*}
where $(\mathbf{dead\ qudits})_{\text{above}}$ and $(\mathbf{dead\ qudits})_{\text{below}}$ denote qubits in the dead phase that is above the cut and below the cut, respectively. Similar division can also be applied to the $\mathbf{unborn\ qudits}$. Hence, we have
\begin{align*}
\ket{\gamma(t)}&=(\mathbf{dead\ qudits})_{\text{above}}\otimes (\mathbf{unborn\ qudits})_{\text{above}}\otimes(\mathbf{active\ qudits})\\
&\qquad\otimes (\mathbf{dead\ qudits})_{\text{below}}\otimes (\mathbf{unborn\ qudits})_{\text{below}},
\end{align*}
$S(\gamma_A(t))$ thus equals the von Neumann entropy of the reduced state of $\mathbf{active\ qudits}$, which equals $S((\psi_{\mathrm{output}})_A)$ given that the $\mathbf{active\ qudits}$ encodes $\psioutput$ in their orientations when $t\geq (3n-1)R$.
\end{proof}

Next, we introduce a concept called ``turn'' to analyze the entanglement structure of the ground state $\ket{\psi_{\mathrm{ground}}}$. Formally, we divide the $L+1$ clock states $\ket{\gamma(0)},\ldots,\ket{\gamma(L)}$ in legal shapes into $T+1$ turns, where in each turn, the changes of the legal shapes happen on the same side of the cut. As an example, suppose $n$ is even and the cut is between the $\frac{n}{2}$-th line and the $\left(\frac{n}{2}+1\right)$-th line. Then, the first turn consists of $\{\ket{\gamma(0)},\ldots,\ket{\gamma(n/2)}\}$, the second turn consists of $\{\ket{\gamma(n/2+1)},\ldots,\ket{\gamma(3n/2)}\}$, and so on. An example with $n=6$ and $T=2$ is given in \cref{fig:turns}. It is worth pointing out that the concept of ``turn'' is new and is crucial to our analysis regarding the entanglement structure of the ground states.

\begin{figure}[h!]
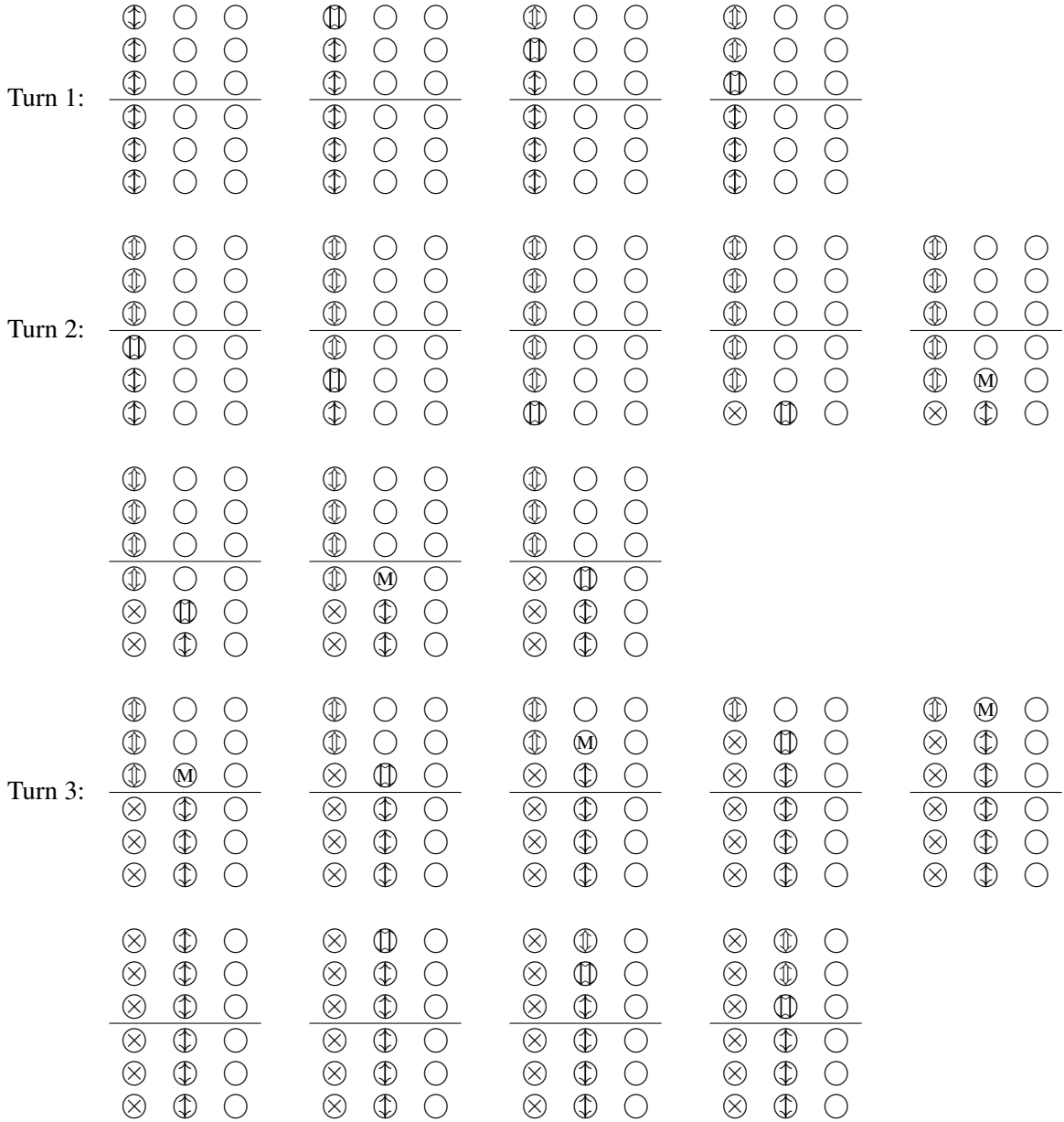

\center{{
\begin{equation*}
\begin{array}{cccccccccc}
\text{Turn 1:} &
\begin{array}{ccc}
\stateb & \statea & \statea \\
\stateb & \statea & \statea \\
\stateb & \statea & \statea \\
\hline
\stateb & \statea & \statea \\
\stateb & \statea & \statea \\
\stateb & \statea & \statea
\end{array}& &
\begin{array}{ccc}
\statee & \statea & \statea \\
\stateb & \statea & \statea \\
\stateb & \statea & \statea \\
\hline
\stateb & \statea & \statea \\
\stateb & \statea & \statea \\
\stateb & \statea & \statea 
\end{array}& & 
\begin{array}{ccc}
\statec & \statea & \statea \\
\statee & \statea & \statea \\
\stateb & \statea & \statea \\
\hline
\stateb & \statea & \statea \\
\stateb & \statea & \statea \\
\stateb & \statea & \statea 
\end{array} & &
\begin{array}{ccc}
\statec & \statea & \statea \\
\statec & \statea & \statea \\
\statee & \statea & \statea \\
\hline
\stateb & \statea & \statea \\
\stateb & \statea & \statea \\
\stateb & \statea & \statea 
\end{array}\\\nonumber
&&&\\\nonumber

\text{Turn 2:} &
\begin{array}{ccc}
\statec & \statea & \statea \\
\statec & \statea & \statea \\
\statec & \statea & \statea \\
\hline
\statee & \statea & \statea \\
\stateb & \statea & \statea \\
\stateb & \statea & \statea 
\end{array}& &
\begin{array}{ccc}
\statec & \statea & \statea \\
\statec & \statea & \statea \\
\statec & \statea & \statea \\
\hline
\statec & \statea & \statea \\
\statee & \statea & \statea \\
\stateb & \statea & \statea 
\end{array}& & 
\begin{array}{ccc}
\statec & \statea & \statea \\
\statec & \statea & \statea \\
\statec & \statea & \statea \\
\hline
\statec & \statea & \statea \\
\statec & \statea & \statea \\
\statee & \statea & \statea 
\end{array} & &
\begin{array}{ccc}
\statec & \statea & \statea \\
\statec & \statea & \statea \\
\statec & \statea & \statea \\
\hline
\statec & \statea & \statea \\
\statec & \statea & \statea \\
\stated & \statee & \statea 
\end{array}& &
\begin{array}{ccc}
\statec & \statea & \statea \\
\statec & \statea & \statea \\
\statec & \statea & \statea \\
\hline
\statec & \statea & \statea \\
\statec & \statem & \statea \\
\stated & \stateb & \statea 
\end{array}\\\nonumber
&&&\\\nonumber

 &
\begin{array}{ccc}
\statec & \statea & \statea \\
\statec & \statea & \statea \\
\statec & \statea & \statea \\
\hline
\statec & \statea & \statea \\
\stated & \statee & \statea \\
\stated & \stateb & \statea 
\end{array}& &
\begin{array}{ccc}
\statec & \statea & \statea \\
\statec & \statea & \statea \\
\statec & \statea & \statea \\
\hline
\statec & \statem & \statea \\
\stated & \stateb & \statea \\
\stated & \stateb & \statea
\end{array}
 & &
 \begin{array}{ccc}
\statec & \statea & \statea \\
\statec & \statea & \statea \\
\statec & \statea & \statea \\
\hline
\stated & \statee & \statea \\
\stated & \stateb & \statea \\
\stated & \stateb & \statea 
\end{array}
\\\nonumber
&&&\\\nonumber

\text{Turn 3:} &
\begin{array}{ccc}
\statec & \statea & \statea \\
\statec & \statea & \statea \\
\statec & \statem & \statea \\
\hline
\stated & \stateb & \statea \\
\stated & \stateb & \statea \\
\stated & \stateb & \statea 
\end{array}& &

\begin{array}{ccc}
\statec & \statea & \statea \\
\statec & \statea & \statea \\
\stated & \statee & \statea \\
\hline
\stated & \stateb & \statea \\
\stated & \stateb & \statea \\
\stated & \stateb & \statea 
\end{array}& & 
\begin{array}{ccc}
\statec & \statea & \statea \\
\statec & \statem & \statea \\
\stated & \stateb & \statea \\
\hline
\stated & \stateb & \statea \\
\stated & \stateb & \statea \\
\stated & \stateb & \statea 
\end{array} & &
\begin{array}{ccc}
\statec & \statea & \statea \\
\stated & \statee & \statea \\
\stated & \stateb & \statea \\
\hline
\stated & \stateb & \statea \\
\stated & \stateb & \statea \\
\stated & \stateb & \statea
\end{array} & &
\begin{array}{ccc}
\statec & \statem & \statea \\
\stated & \stateb & \statea \\
\stated & \stateb & \statea \\
\hline
\stated & \stateb & \statea \\
\stated & \stateb & \statea \\
\stated & \stateb & \statea
\end{array}
\\\nonumber
&&&\\\nonumber

 &
\begin{array}{ccc}
\stated & \stateb & \statea\\
\stated & \stateb & \statea\\
\stated & \stateb & \statea\\
\hline
\stated & \stateb & \statea\\
\stated & \stateb & \statea\\
\stated & \stateb & \statea
\end{array}& & 
\begin{array}{ccc}
\stated & \statee & \statea\\
\stated & \stateb & \statea\\
\stated & \stateb & \statea\\
\hline
\stated & \stateb & \statea\\
\stated & \stateb & \statea\\
\stated & \stateb & \statea
\end{array}& & 
\begin{array}{ccc}
\stated & \statec & \statea \\
\stated & \statee & \statea \\
\stated & \stateb & \statea \\
\hline
\stated & \stateb & \statea \\
\stated & \stateb & \statea \\
\stated & \stateb & \statea 
\end{array}& &
\begin{array}{ccc}
\stated & \statec & \statea \\
\stated & \statec & \statea \\
\stated & \statee & \statea \\
\hline
\stated & \stateb & \statea \\
\stated & \stateb & \statea \\
\stated & \stateb & \statea 
\end{array}
\\\nonumber
&&&\\\nonumber
\end{array}\nonumber
\end{equation*}
\vspace{-13mm}
}}
\caption{An example of the first three turns when $n=6$, $T=2$, and the cut is between the third line and the fourth line.}
\label{fig:turns}
\end{figure}

For the $p$-th turn of legal shapes starting at the $t_s^{(p)}$-th legal shape and ending at the $t_e^{(p)}$-th legal shape, we define the uniform superposition state of this turn to be
\begin{align}
\ket{\zeta(p)}\equiv\frac{1}{\sqrt{t_e^{(p)}-t_s^{(p)}+1}}\sum_{\tau=t_s^{(p)}}^{t_e^{(p)}}\ket{\gamma(\tau)}.\label{eqn:turn-superposition}
\end{align}
Then, the ground state $\psiground$ defined in \cref{eqn:2d-groundstate} can also be expressed as
\begin{align*}
\ket{\psi_{\mathrm{ground}}}=\sum_{p=1}^{T+1}\sqrt{\frac{t_e^{(p)}-t_s^{(p)}+1}{L+1}}\ket{\zeta(p)}.
\end{align*}

\begin{lemma}
For any $p_1\neq p_2\in[1,T+1]$, the uniform superposition states $\ket{\zeta(p_1)}$ and $\ket{\zeta(p_2)}$ of these two turns are cutwise orthogonal with respect to the cut per \cref{def:cutwise-orthogonality}. \label{lem:cutwise-orthogonality}
\end{lemma}
\begin{proof}
We prove this lemma by showing that for any two states $\ket{\gamma(t_1)},\ket{\gamma(t_2)}$ in these two turns, respectively, are cutwise orthogonal.

If $|p_1-p_2|>1$, note that the number of qudits in dead phases are different on both sides of the cut for any possible $\ket{\gamma(t_1)}$ and $\ket{\gamma(t_2)}$. Hence, they are cutwise orthogonal.

If $|p_1-p_2|=1$, i.e., $p_1$ and $p_2$ are two consecutive turns, then the changes of shapes two turns are on different sides of the cut. Without loss of generality we assume $p_2=p_1+1$ and the changes in the $p_2$-th turn happen below the cut. Then, $\ket{\gamma(t_2)}$ contains a qudit in the flag phase $\ket{\statee}$ or the marker phase $\ket{\statem}$ below the cut, which is not present in $\ket{\gamma(t_1)}$. As for the part above the cut, note that $\ket{\gamma(t_2)}$ contains a column whose qudits above the cut are all in the second phase $\ket{\stateb}$, which is not present in $\ket{\gamma(t_1)}$. Hence, $\ket{\gamma(t_1)}$ and $\ket{\gamma(t_2)}$ are cutwise orthogonal.
\end{proof}
\paragraph{Padding does not change entanglement structure within a turn.}
\noindent The next lemma can be viewed as a sanity check. Suppose we are at the time instance corresponding to $t_s = (3n-1)R$, so that $|\gamma(t_s)\rangle$ corresponds to an encoding of $|\psi_{\text{output}}\rangle$. Note that the state $|\psi_{\text{output}}\rangle$ has nice entanglement properties that we desire. For $t \geq t_s$, all we do to the circuit is pad it with unitaries. Ideally, it should not change the entanglement structure of the encoded state, but it is not immediately obvious that it is indeed the case, because of the complicated way in which we encode the clock. In the next lemma, we assuage some of that concern by showing that the entanglement structure indeed \emph{does not} change, when we consider a superposition of temporal states, for $t > t_s$, that are all part of the same turn.

\begin{lemma}\label{lem:legal-shape-entanglement}
For any turn of states in legal shapes that starts at the $t_s$-th shape and ends at the $t_e$-th shape, denote
\begin{align*}
\ket{\zeta}\equiv\frac{1}{\sqrt{t_e-t_s+1}}\sum_{\tau=t_s}^{t_e}\ket{\gamma(\tau)}
\end{align*}
to be the uniform superposition state over all the clock states in this turn as in \cref{eqn:turn-superposition}. Then, if $t_s\geq (3n-1)R$, we have
\begin{align*}
S(\zeta_A)=S(\gamma_A(\tau))=S((\psi_{\mathrm{output}})_A),\quad\forall \tau\in[t_s,t_e].
\end{align*}
Otherwise, we have
\begin{align*}
0\leq S(\zeta_A)\leq 2n.
\end{align*}
\end{lemma}

\begin{proof}
    Without loss of generality we assume the changes of the shapes in this turn happens above the cut. Observe that in all the states in this phase, their \textbf{active qudits} string encodes the same quantum state $\psioutput$ in their orientations and the encodings are cut through at the same position, however in orthogonal subspaces. Hence, if $\psioutput$ admits the following Schmidt decomposition
    \begin{align*}
    \psioutput=\sum_i\nu_i\ket{\psi_{i}^{(A)}}\otimes\ket{\psi_{i}^{(B)}},
    \end{align*}
    across the cut, where $A$ and $B$ denote the subsystems above and below the cut, respectively, the state in the $\tau$-th shape $\ket{\gamma(\tau)}$ admits the following decomposition
    \begin{align*}
    \ket{\gamma(\tau)}=\sum_i\nu_i\ket{\gamma_{i}^{(A)}(\tau)}\otimes\ket{\gamma_{i}^{(B)}},
    \end{align*}
    where $\{\ket{\gamma_{i}^{(A)}(\tau_1)}\}$ and $\{\ket{\gamma_{i}^{(A)}(\tau_2)}\}$ span orthogonal subspaces for any $\tau_1\neq\tau_2\in[t_s,t_e]$, and the RHS of the tensor product is identical for every $t_s\leq\tau\leq t_e$. Hence, the superposition state $\ket{\zeta}$ of this phase satisfies
    \begin{align*}
    \ket{\zeta}
    &=\frac{1}{\sqrt{t_e-t_s+1}}\sum_{\tau=t_s}^{t_e}\sum_i\nu_i\ket{\gamma_{i}^{(A)}(\tau)}\otimes\ket{\gamma_{i}^{(B)}}\\
    &=\sum_i\nu_i\left(\frac{1}{\sqrt{t_e-t_s+1}}\sum_{\tau=t_s}^{t_e}\ket{\gamma_{i}^{(A)}(\tau)}\right)\otimes\ket{\gamma_{i}^{(B)}}.
    \end{align*}
    Hence, for any $t_s\leq\tau\leq t_e$, tracing out the left side of $\ket{\gamma(\tau)}$ yields the same reduced state as tracing out the left side of $\ket{\zeta}$, by which we can conclude that
    \begin{align*}
        S(\zeta_A)=S(\gamma_A(\tau))=S((\psi_{\mathrm{output}})_A),\quad\forall t_s\leq\tau\leq t_e.
    \end{align*}
    where the last equation is from \cref{fact:legal-shape-entanglement}. 

\end{proof}
\begin{remark}
    As for the case where $t_s<(3n-1)R$, it is possible that in this turn there exists some clock state $\ket{\gamma(\tau)}$ encoding an intermediate state of the circuit rather than $\psioutput$. Nevertheless, all the clock states in this turn have their $(\mathbf{active\ qudits})$ concentrated within only two columns of the 2D grid. Hence, $\ket{\zeta}$ is in a Hilbert space spanned by $2n$ basis states which has dimension at most $2^{2n}$, which leads to
    \begin{align*}
0\leq S(\zeta_A)\leq 2n.
\end{align*}
\end{remark}

\paragraph{Showing that the ground state inherits desirable entanglement properties because of padding.}
In this section, we will collect all our ideas together and prove that the entanglement entropy of the ground state of our clock Hamiltonian is, more or less, the same as the entanglement entropy of $|\psi_{\text{output}}\rangle$. At a high level, this is because of padding. Since we pad with sufficiently many identities, the entanglement structure of the ground state looks fairly similar to $|\psi_{\text{output}}\rangle$, because of a preponderence of terms which encode $|\psi_{\text{output}}\rangle$ in the superposition. 
\begin{lemma}
\label{entanglement difference 2D}
Consider a horizontal cut through the 2D grid of size $n\times(T+1)$ such that the cut has distance at least $\omega(\log n)$ to both the top and the bottom of the grid, the entanglement entropy $S((\psi_{\mathrm{ground}})_A)$ of the ground state $\ket{\psi_{\mathrm{ground}}}$ defined in \cref{eqn:2d-groundstate} satisfies
\begin{align*}
\left(1-\frac{1}{n}\right)\cdot S((\psi_{\mathrm{output}})_A)\leq S((\psi_{\mathrm{ground}})_A)\leq S((\psi_{\mathrm{output}})_A)+\poly\log n,
\end{align*}
for any $T=\poly n$ satisfying $T\geq 3nR$. \label{prop:2D-groundstate-entanglement}
\end{lemma}
\begin{proof}
For any $p\in[1,T+1]$, we denote
\begin{align*}
\alpha_p\equiv\sqrt{\frac{t_e^{(p)}-t_s^{(p)}+1}{L+1}},
\end{align*}
where $t_s^{(p)}$ and $t_e^{(p)}$ separately denotes the number of starting shape and ending shape of the $p$-th turn. Then the ground state $\psiground$ defined in \cref{eqn:2d-groundstate} satisfies
\begin{align*}
\psiground=\sum_{p=1}^{T+1}\alpha_p\ket{\zeta(p)}.
\end{align*}
Since the states $\ket{\zeta(p)}$  are shown to be cutwise orthogonal by \cref{lem:cutwise-orthogonality}, we can derive that
\begin{align*}
S((\psi_{\mathrm{ground}})_A)=\sum_{p=1}^{T+1}|\alpha_p|^2S(\zeta_A(p))+2\sum_{p=1}^{T+1}|\alpha_p|^2\log(1/|\alpha_p|)
\end{align*}
according to \cref{lem:CO-entanglement}. Furthermore, by \cref{lem:legal-shape-entanglement} we can establish that
\begin{align*}
S(\zeta_A(p))=S((\psi_{\mathrm{output}})_A),\quad\forall p\in(R+1,T+1],
\end{align*}
which leads to
\begin{align}
S((\psi_{\mathrm{ground}})_A)=\sum_{p=R+2}^{T+1}|\alpha_p|^2S((\psi_{\mathrm{output}})_A)+2\sum_{p=1}^{T+1}|\alpha_p|^2\log(1/|\alpha_p|)+\sum_{p=1}^{R+1}|\alpha_p|^2S(\zeta_A(p))\label{eqn:ground-entanglement-decomposition}
\end{align}
Given that $T=\poly n$, we have
\begin{align*}
0\leq \sum_{p=1}^{T+1}|\alpha_p|^2\log(1/|\alpha_p|)\leq \poly\log n
\end{align*}
for any possible values of $\alpha_1,\ldots,\alpha_{T+1}$. As for the third term in \cref{eqn:ground-entanglement-decomposition}, we have
\begin{align*}
0\leq\sum_{p=1}^{R+1}|\alpha_p|^2S(\zeta_A(p))\leq 2n\cdot\sum_{p=1}^{R+1}|\alpha_p|^2,
\end{align*}
given that $0\leq S(\zeta_A(p))\leq 2n$ by \cref{lem:legal-shape-entanglement}. Furthermore, we have
\begin{align*}
\sum_{p=1}^{R+1}|\alpha_p|^2\leq\sum_{p=1}^{R+1}\frac{t_e^{(p)}-t_s^{(p)}+1}{L+1}\leq\frac{3nR}{2nT+1}\leq\frac{1}{n},
\end{align*}
and
\begin{align*}
\sum_{p=R+2}^{T+1}|\alpha_p|^2=1-\sum_{p=1}^{R+1}|\alpha_p|^2\geq1-\frac{1}{n}.
\end{align*}
Thus, we can conclude that
\begin{align*}
\left(1-\frac{1}{n}\right)\cdot S((\psi_{\mathrm{output}})_A)\leq S((\psi_{\mathrm{ground}})_A)\leq S((\psi_{\mathrm{output}})_A)+\poly\log n.
\end{align*}
\end{proof}

\subsubsection{Learning ground state entanglement structure for 2D geometrically local Hamiltonians}

We can now use the padded circuit-to-2D-Hamiltonian construction to convert the pseudoentangled states ensembles we constructed in \cref{sec:multicut} into families of 2D Hamiltonians for which it is difficult to decide whether the ground state has low entanglement (i.e., $\poly\log n$) or high entanglement (i.e., $\Omega(n)$) entanglement across horizontal cuts through the $n\times\poly(n)$ grid. Formally, we show the following.

\begin{theorem}
\label{thm:LGSES-hard-2d-formal}
For every $n \in \N$, there exist two families $\mathcal{H}^\lo_n$ and $\mathcal{H}^\hi_n$ of \emph{geometrically 2D-local} Hamiltonians on $(n \times \poly n)$ qudits arranged in an $n \times \poly(n)$ grid with spectral gap $\Omega(1/\poly(n))$, and there are efficient procedures that sample (classical descriptions of) Hamiltonians from either family (denoted $H \leftarrow \mathcal{H}^\lo_n$ and $H \leftarrow \mathcal{H}^\hi_n$) with the following properties:
\begin{enumerate}
\item Hamiltonians sampled according to $H \leftarrow \mathcal{H}^\lo_n$ and $H \leftarrow \mathcal{H}^\hi_n$ are computationally indistinguishable under \cref{lwe_assumption_subexp}. \label{item:ham2d_indist}
\item  With overwhelming probability, the ground states of Hamiltonians $H \leftarrow \mathcal{H}^\lo_n$ have $\poly\log n$ entanglement across horizontal cuts through the grid that are at least $\omega(\log n)$ far from the boundary, and Hamiltonians $H \leftarrow \mathcal{H}^\lo_n$ have $\Omega(n)$ entanglement across the same cuts.
\end{enumerate}
\end{theorem}

\begin{proof}
We consider the same construction of $\cH_n^\lo$ and $\cH_n^\hi$ as in \cref{thm:formal-lgses-binary}, except that we now use the 2D circuit-to-Hamiltonian construction described in \cref{sec:2D-Hamiltonian}.
With this, we only need to show \cref{item:hamun_ent} as the condition on the spectral gap follows from~\cite{aharonov2008adiabatic} and the other properties follow in exactly the same way as in the proof of \cref{thm:formal-lgses-binary} 

Consider a Hamiltonian $H \leftarrow \cH^\lo_n$ or $H \leftarrow \cH^\hi_n$.
Let $|\psi_{\mathrm{ground}}\rangle$ be the history state of the circuit as defined in~\cref{eqn:2d-groundstate}.
By \cref{lem:groundstates_all_close} any ground state of $H$ is $\eps$-close for any $\eps = 1/\poly(n)$ that we can choose.
Therefore, choosing $\eps$ small enough and using the continuity bound for the von Neumann entropy (\cref{continuity_vonneumannentropy}), we can ensure that any ground state of $H$ only differs in entanglement entropy from $|\psi_{\mathrm{ground}}\rangle$ by $1/\poly(n)$.
It therefore suffices to show \cref{item:hamun_ent} for $|\psi_{\mathrm{ground}}\rangle$.
This follows directly from \cref{entanglement difference 2D} and the fact that the states $\ket{\psi_{\rm output}}$ are pseudoentangled with entanglement gap $O(\poly\log n)$ vs $\Omega(n)$ by our pseudoentanglement construction (\cref{thm:multi_cut}).
\end{proof}

\begin{remark}
Just as in \cref{first_remark}, under the standard LWE assumption (\cref{lwe_assumption}) \cref{thm:LGSES-hard-2d-formal} still holds, but with the smaller entanglement gap $O(n^\delta)$ vs $\Omega(n)$ for any $\delta > 0$.
\end{remark}

\bibliographystyle{alpha}
\bibliography{main}

\end{document}